\documentclass[11pt]{article}
\usepackage[utf8]{inputenc}
\usepackage{amsmath,amsfonts,amssymb}
\usepackage{amsthm}
\usepackage{latexsym}
\usepackage{mathtools}
\usepackage[noadjust]{cite}
\usepackage{graphicx}
\usepackage{xcolor}
\usepackage{dirtytalk}
\usepackage[margin=1in]{geometry}
\usepackage{thm-restate}
\usepackage[linesnumbered,ruled]{algorithm2e}

\usepackage{cleveref}
\newtheorem{theorem}{Theorem}[section]
\newtheorem{corollary}[theorem]{Corollary}
\newtheorem{lemma}[theorem]{Lemma}

\newtheorem{proposition}[theorem]{Proposition}
\newtheorem{claim}[theorem]{Claim}
\theoremstyle{definition}
\newtheorem{definition}[theorem]{Definition}

\newtheorem{fact}[theorem]{Fact}

\newcommand{\C}{\mathbb{C}}
\newcommand{\N}{\mathbb{N}}

\newcommand{\R}{\mathbb{R}}

\DeclareMathOperator*{\E}{\mathbb{E}}

\DeclareMathOperator{\slft}{slift}
\DeclareMathOperator{\Span}{span}
\newcommand{\slift}[1]{\slft\left(#1\right)}

\newcommand{\ep}{\eps}
\newcommand{\ra}{\rightarrow}
\newcommand{\la}{\leftarrow}
\newcommand{\sym}[1][ ]{\mathbf{U}_{#1}}

\newcommand{\boldVar}[1]{\mathbf{#1}}
\newcommand{\mvar}[1]{\boldVar{#1}}
\newcommand{\vvar}[1]{\vec{#1}}

\newcommand{\mX}{\mvar{X}}

\newcommand{\mA}{\mvar{A}}
\newcommand{\mF}{\mvar{F}}
\newcommand{\mB}{\mvar{B}}
\newcommand{\mS}{\mvar{S}}
\newcommand{\mI}{\mvar{I}}
\newcommand{\mE}{\mvar{E}}
\newcommand{\mW}{\mvar{W}}
\newcommand{\mV}{\mvar{V}}
\newcommand{\mQ}{\mvar{Q}}
\newcommand{\mM}{\mvar{M}}
\newcommand{\mR}{\mvar{R}}
\newcommand{\mL}{\mvar{L}}
\newcommand{\mP}{\mvar{P}}

\newcommand{\tS}{\widetilde{\mS}}
\newcommand{\tA}{\widetilde{\mA}}
\newcommand{\tV}{\widetilde{\mV}}
\newcommand{\tB}{\widetilde{\mB}}
\newcommand{\tW}{\widetilde{\mW}}
\newcommand{\tL}{\widetilde{\mL}}
\newcommand{\tR}{\widetilde{\mR}}
\newcommand{\tM}{\widetilde{\mM}}

\newcommand{\pimin}{\pi_{\min}}

\newcommand*{\defeq}{\stackrel{\text{def}}{=}}

\newcommand{\poly}{\operatorname{poly}}
\newcommand{\polylog}{\operatorname{polylog}}

\renewcommand{\L}{\mathbf{L}}

\newcommand{\md}{\operatorname{md}}

\newcommand{\uc}{\stackrel{\circ}{\approx}}
\newcommand{\sv}{\stackrel{sv}{\approx}}

\newcommand{\eps}{\varepsilon}
\newcommand{\tO}{\tilde{O}}

\newcommand{\zo}{\{0,1\}}

\global\long\def\C{\mathbb{C}}%
\global\long\def\defeq{\stackrel{\mathrm{{\scriptscriptstyle def}}}{=}}%
\global\long\def\E{\mathbb{E}}%
\global\long\def\otilde{\tilde{O}}%

\global\long\def\norm#1{\|#1\|}%
\global\long\def\vx{x}%
 
\global\long\def\vy{y}%

\global\long\def\vzero{\vvar 0}%
 
\newcommand{\vones}{\vvar 1}%

\global\long\def\ma{\mvar A}%
 
\global\long\def\mb{\mvar B}%

\global\long\def\md{\mvar D}%
\global\long\def\mE{\mvar E}%

\global\long\def\mI{\mvar I}%
\global\long\def\mJ{\mvar J}%

\global\long\def\mL{\mvar L}%
\global\long\def\mm{\mvar M}%

\global\long\def\mq{\mvar Q}%
 
\global\long\def\mr{\mvar R}%
 
\global\long\def\ms{\mvar S}%

\global\long\def\mU{\mvar U}%
 
\global\long\def\mv{\mvar V}%
 
\global\long\def\mw{\mvar W}%

\global\long\def\mproj{\mvar P}%

\global\long\def\mzero{\mvar 0}%
\global\long\def\argmin{\mathrm{argmin}}%
\global\long\def\vol{\mathrm{vol}}%
\global\long\def\supp{\mathrm{supp}}%
\global\long\def\poly{\mathrm{poly}}%

\newcommand{\mD}{\mathbf{D}}
\newcommand{\mPi}{\mathbf{\Pi}}

\newcommand{\diag}{\mathrm{diag}}
\newcommand{\mdiag}{\mvar{diag}}%

\newcommand{\uapprox}{\mathbin{\stackrel{sv}{\approx}}}

\newcommand{\dout}{d_{\mathrm{out}}}%
\newcommand{\din}{d_{\mathrm{in}}}%
\newcommand{\mdout}{\mD_{\mathrm{out}}}%
\newcommand{\mdin}{\mD_{\mathrm{in}}}%
\newcommand{\SC}{\mathrm{SC}}%

\newcommand{\Din}{\mathbf{D}_{\mathrm{in}}}

\newcommand{\Dout}{\mathbf{D}_{\mathrm{out}}}

\newcommand{\cTrans}{*}
\newcommand{\rTrans}{\top}

\newcommand{\svapprox}[3]

\newcommand{\svgeneral}[3]{\mathbin{\stackrel{#1,#2}{\approx_{#3}}}}
\newcommand{\svEulerian}[2]{\mathbin{\stackrel{#1}{\approx}_{#2}}}
\newcommand{\svgraph}[1]{\mathbin{\stackrel{\mathit{sv}}{\approx}_{#1}}}
\newcommand{\svn}[1]{\mathbin{\stackrel{\mathit{svn}}{\approx}_{#1}}}

\newcommand{\lker}{\mathrm{lker}}
\newcommand{\rker}{\mathrm{rker}}
\newcommand{\sigmamax}{\sigma_{\mathrm{max}}}
\newcommand{\tC}{\mathbf{\tilde{C}}}
\newcommand{\mC}{\mathbf{C}}
\newcommand{\mEL}{\mathbf{E}}
\newcommand{\mER}{\mathbf{F}}
\newcommand{\mN}{\mathbf{N}}
\newcommand{\tN}{\mathbf{\tilde{N}}}

\newcommand{\algUnweightedOnce}{\textsc{SparsifyCycleUnweighted}~}

\newcommand{\algSparsify}{\textsc{SparsifyGraph}~}
\newcommand{\algProduct}{\textsc{SparsifyProduct}}
\newcommand{\algSparsifyOnce}{\textsc{SparsifyCycle}~}
\newcommand{\algCycleDecomp}{\textsc{CycleDecomp}~}
\newcommand{\algExpDecomp}{\textsc{ExpanderDecomp}}
\newcommand{\TalgExpDecomp}[1]{\mathtt{T}(\textsc{ExpanderDecomp},#1)}
\newcommand{\PhialgExpDecomp}{\Phi(\algExpDecomp)}

\newcommand{\algFindClose}{\textsc{FindClose}}
\newcommand{\algPowerDi}{\textsc{SparsifyScaledPower}}

\newcommand{\algCut}{\textsc{EstimateCut}}

\newcommand{\muedge}{\mu_{\mathrm{edge}}}

\newcommand{\Cut}{\mathrm{Cut}}
\newcommand{\Uncut}{\mathrm{Uncut}}
\newcommand{\tG}{\widetilde{G}}

\newif\ifdraft
\drafttrue

\newif\ifanon
\anonfalse

\newcommand{\ds}{\text{\textcircled{s}}}

\title{Singular Value Approximation and \\ Sparsifying Random Walks on Directed Graphs}

\ifanon
\author{Anonymous submission to FOCS 2023}
\else
\author{AmirMahdi Ahmadinejad\\Amazon\thanks{The work was started prior to joining Amazon and does not relate to Amazon}\\ \texttt{ahmadi@alumni.stanford.edu} \and John Peebles\\Apple\\ \texttt{peebles@apple.com} \and Edward Pyne\thanks{Supported by an Akamai Presidential Fellowship.} \\MIT\\\texttt{epyne@mit.edu} \and Aaron Sidford\thanks{Supported by a Microsoft Research Faculty Fellowship, NSF CAREER Award CCF-1844855, NSF Grant CCF-1955039, a PayPal research award, and a Sloan Research Fellowship.}
\\Stanford University\\\texttt{sidford@stanford.edu} \and Salil Vadhan\thanks{Supported by a Simons Investigator Award.}\\Harvard University\\ \texttt{salil\_vadhan@harvard.edu}}
\fi

\begin{document}
\begin{titlepage}
\maketitle

\begin{abstract}
In this paper, we introduce a new, spectral notion of approximation between directed graphs, which we call \emph{singular value (SV) approximation}. SV-approximation is stronger than previous notions of spectral approximation considered in the literature, including spectral approximation of Laplacians for undirected graphs~\cite{ST04}, standard approximation for directed graphs~\cite{CKPPRSV17}, and unit-circle (UC) approximation for directed graphs~\cite{AKMPSV20}. Further, SV approximation enjoys several useful properties not possessed by previous notions of approximation, e.g., it is preserved under products of random-walk matrices and bounded matrices.

We provide a nearly linear-time algorithm for SV-sparsifying (and hence UC-sparsifying) Eulerian directed graphs, as well as $\ell$-step random walks on such graphs, for any $\ell\leq \poly(n)$. 
Combined with the Eulerian scaling algorithms of \cite{CKKPPRS18}, given an arbitrary (not necessarily Eulerian) directed graph and a set $S$ of vertices, we can approximate the stationary probability mass of the $(S,S^c)$ cut in an 
$\ell$-step random walk to within a multiplicative error of $1/\polylog(n)$ and an additive error of $1/\poly(n)$
in nearly linear time.
As a starting point for these results, we provide a simple black-box reduction from SV-sparsifying Eulerian directed graphs to SV-sparsifying undirected graphs; such a directed-to-undirected reduction was not known for previous notions of spectral approximation.

\end{abstract}
\vfill
\thispagestyle{empty}
\end{titlepage}
\newpage

\tableofcontents
\newpage

\section{Introduction}

Random walks on graphs play a central role in theoretical computer science. In algorithm design, they have found a wide range of applications including, maximum flow~\cite{CKMST11,LRS13,KLOS14,BGJLLPS22,BLLSS0W21},
sampling random spanning trees~\cite{KM09,MST15}, and clustering and partitioning~\cite{AM85,KVV04,ACL06,OSV12}. Correspondingly, new algorithmic results on efficiently accessing properties of random walks have the potential for broad implications.
In particular, in complexity theory, such algorithms have attracted attention as a promising approach to derandomizing space-bounded computation~\cite{SZ99,Rei08,RTV,AKMPSV20}. 

In this paper we consider the well-studied problem of \emph{estimating the $\ell$-step random walk on a directed graph}. Given a strongly connected, weighted, directed graph $G = (V, E, w)$, its associated random walk matrix $\mW \in \R^{V \times V}$, and an integer $\ell > 0$, we seek to 
approximate key properties of the $\ell$-step random walk, $\mW^\ell$, more efficiently than
we could computing $\mW^\ell$ explicitly.  For example, we may wish to estimate individual entries of $\mW^\ell$, the conductance or cut probabilities of subsets of vertices, or (expected) hitting times between pairs of vertices.

In recent years, \emph{graph sparsification} has emerged as a powerful approach for efficiently solving such problems. When the graph is undirected, we look for {\em spectral sparsifiers} of the {\em Laplacian} $\mL = \mD-\mA$, where
$\mD$ is the diagonal matrix of degrees and $\mA$ is the adjacency matrix. 
It is known that for all $\epsilon \in (0, 1)$, that there exist $\epsilon$-spectral sparsifiers with sparsity $\otilde(|V| \epsilon^{-2})$; that is, a Laplacian matrix $\tilde{\mL}$ with $\otilde(|V| \epsilon^{-2})$ non-zero entries such that 
\begin{equation} \label{eqn:spectral-sparsifier}
(1 - \epsilon) x^\top \mL x \leq x^\top \tilde{\mL} x \leq (1 + \epsilon) x^\top \mL x
\text{ for all }
x \in \R^V\,.
\end{equation}

Spectral sparsifiers can be computed in
nearly linear time~\cite{ST04,SS08,BSS12,PS13}. 
Normalizing such a sparsifier $\tilde{\mL}$ by $\mD^{-1/2}$ on both sides, we obtain a spectral approximation of 
the {\em normalized Laplacian} $\mD^{-1/2} \mL \mD^{-1/2}$, which directly gives information about random walks because it is equivalent (up to a change of basis) to the  {\em random-walk Laplacian}, $\mL\mD^{-1} = \mI-\mW$.
Indeed, 
from any spectral sparsifier $\tilde{\mL}$ satisfying \Cref{eqn:spectral-sparsifier}, we
can approximate any desired cut $(S,S^c)$ in the original graph in nearly linear time by evaluating $x^\top \tL x$ for $x$ equal to the indicator vector of $S$. 
Furthermore, there are nearly linear-time algorithms for computing sparse $\epsilon$-spectral sparsifiers corresponding to the $\ell$-step random walk, i.e., sparsifiers of the weighted graph whose random-walk Laplacian is $\mI-\mW^\ell$, for
any polynomial length $\ell$~\cite{ChengCLPT15,MRSV2}.

Obtaining analogous results for sparsifying $\mI - \mW^\ell$ for directed graphs has been more challenging. For a directed graph, we consider the \emph{directed Laplacian}~\cite{CKPPSV16} $\L=\Dout-\mA^{\rTrans}$ where $\Dout$ is the associated diagonal matrix of out-degrees and $\mA^{\rTrans}$ is the transpose of the associated weighted adjacency matrix, $\mA$.  In comparison to their symmetric counterparts for undirected graphs, nearly linear-time sparsification algorithms  (which approximate more than the associated undirected graph) were developedl more recently~\cite{CKPPRSV17,CGPSSW18} and have yet to be extended to handle long random walks.
Here we describe challenges in sparsifying $\mI-\mW^\ell$ for directed graphs.

\paragraph{Unknown Stationary Distribution.} While the kernel of an undirected Laplacian matrix is the all ones vector, computing the kernel of a directed Laplacian matrix $\L$ corresponds to computing the stationary distribution $\pi$ of the random walk on the directed graph ($\L\Dout^{-1}\pi = 0$). Without explicitly knowing the kernel, it is not known how to efficiently perform any kind of
useful sparsification or approximately solve linear systems in $\L$. 
This difficulty was overcome in \cite{CKPPSV16, AhmadinejadJSS19} which provide reductions from solving general directed Laplacian systems to the case where the graph is {\em Eulerian}, meaning that every vertex has the same in-degree as out-degree. In Eulerian graphs, the stationary distribution is simply proportional to the vertex degrees and the all ones vector is both the left and right kernels of the associated directed Laplacian.

\paragraph{Defining Approximation.} Undirected Laplacians $\L$ are symmetric and positive semidefinite (PSD), i.e., $x^\top \L x \geq 0$ for all $x$.
This leads to the natural Spielman--Teng~\cite{ST04} definition of multiplicative approximation
given in \Cref{eqn:spectral-sparsifier}. That is, we say that $\tL$ is an $\eps$-approximation of $\mL$ if $(1-\epsilon)\mL \preceq \tL \preceq (1+\epsilon)\mL$, where $\preceq$ is the L\"owner order on PSD matrices. However, even though Laplacians of directed graphs are potentially asymmetric, the quadratic form $x^\top \mL x$ depends only on a symmetrization of the Laplacian ($x^\top \mL x = x^\top ((\mL + \mL^\top)/2)x$). Consequently, the quadratic form discards key information about the associated directed graph (e.g. the quadratic form of a directed cycle and an undirected cycle are the same). Thus, defining approximation for directed graphs (even Eulerian ones) is more challenging than for undirected graphs and a more complex notion of approximation was introduced in \cite{CKPPSV16}. 
This additional complexity requires designing new sparsification algorithms that take into account the directedness of the graph. 

\paragraph{Preservation under Powering.}  Even for undirected graphs, the standard definition of spectral approximation
in~\Cref{eqn:spectral-sparsifier} 
is not preserved under powering.  That is, $\mI -\tW \approx \mI-\mW$ does not
imply that $\mI - \tW^2 \approx \mI-\mW^2$. Indeed, in graphs that are bipartite (and connected), $\mI-\mW^2$ has a two-dimensional kernel, corresponding to the $\pm1$ eigenvalues of $\mW$, whereas $\mI-\mW$ has only a one-dimensional kernel.  Standard spectral approximation requires perfect preservation of the kernel of $\mI-\mW$, but not of $\mI-\mW^2$.  Graphs that are nearly bipartite (i.e., where $\mW$ has an eigenvalue near $-1$) can also experience a large loss in quality of approximation when squaring.

Cheng et al.~\cite{ChengCLPT15} addressed this issue by (implicitly) strengthening spectral approximation to require that $\mI+\tW \approx \mI+\mW$ in addition to $\mI-\tW\approx \mI-\mW$.
This notion of approximation enabled algorithms for sparsifying
$\mI-\mW^\ell$ for undirected graphs in randomized near-linear time \cite{ChengCLPT15,MRSV2} and deterministic logspace~\cite{MRSV2,DMVZ}.  For directed graphs, the problem comes not just from bipartiteness, but general periodic structures (e.g. a directed cycle), which give $\mW$ complex eigenvalues
on or near the unit circle.  This led Ahmadenijad et al.~\cite{AKMPSV20} to propose the notion of {\em unit-circle (UC) approximation}, which amounts to requiring that $\mI-z\tW$ approximate $\mI-z\mW$ for all complex numbers $z$ of magnitude 1, with respect to the standard notion of approximation for directed graphs proposed in \cite{CKPPSV16}.
UC approximation has the property that it is preserved under taking arbitrary powers, with
no loss in the approximation error. As such, sparsification techniques for UC and stronger notions must exactly preserve periodicity. 

\paragraph{Preservation of Periodic Structures.}  Sparsifying directed graphs under UC approximation is more challenging due to the need to preserve periodic structures in the graph, which can be easily lost or introduced by common sparsification techniques such as random sampling~\cite{SS08} or patching to fix degrees.  Thus in \cite{AKMPSV20}, it
was only shown how to partially sparsify the {\em square} of a graph; that is obtain a graph with random-walk matrix
$\widetilde{\mW^2}$ such that $\mI-\widetilde{\mW^2}$ UC-approximates $\mI-\mW^2$, but has fewer edges than the true square $\mW^2$.  Still, the number of edges in $\widetilde{\mW^2}$ is larger than in $\mW$ by at least a constant factor, so if we iterate to obtain a sparse approximation of $\mW^\ell$, the number of edges will grow by a factor of $c^{\log \ell} = \poly(\ell)$ and our approximations will quickly become dense.  This was affordable in the deterministic logspace algorithms of \cite{AKMPSV20}, but is not in our setting of nearly linear time.

\subsubsection*{Our Work}

In this paper we provide several tools for overcoming these challenges, advancing both algorithmic and structural tools regarding graphs sparsification. First, we introduce a new notion of directed graph approximation called \emph{singular value (SV) approximation}. We then show that that this notion of approximation strictly strengthens unit-circle approximation and show that it has a number of desirable properties, 
such as preservation under not only powers but arbitrary products of random-walk matrices, and implying approximation
of stationary probabilities of all cuts.  
Then we provide an efficient near linear-time randomized algorithm for computing nearly linear-sized SV-sparsifiers for arbitrary Eulerian directed graphs; this implies the first proof that nearly linear-sized UC-sparsifiers exist for Eulerian directed graphs.  
As a starting point for this result, we provide a simple reduction from SV-sparsifying Eulerian directed graphs to SV-sparsifying {\em undirected graphs}; no such reduction was known for the previous, weaker forms of spectral approximation of directed graphs, and shows that SV approximation is a significant strengthening even for undirected graphs.

Combined with the Eulerian scaling algorithms of \cite{CKKPPRS18}, we obtain an algorithm for approximating
the stationary probabilities of cuts (as well as ``uncuts'') in random walks on arbitrary directed graphs, which we define as follows:
\begin{definition}[Cut values] \label{def:Cut}
    For a strongly connected, weighted digraph $G$ on $n$ vertices, let $\mu$ be the unique stationary distribution of the random walk on $G$, and let $\muedge$ be the stationary distribution on edges $(i,j)$ of $G$ (i.e., pick $i$ according to $\mu$ and $j$ by following an edge from $i$ proportional to its weight).  For subsets $S$ and $T$ of vertices, define:
    \begin{itemize}
        \item $\Cut_G(S,T) = \Pr_{(i,j)\sim \muedge}[i \in S, j\in T]$,
        \item $\Cut_G(S) = \Cut_G(S,S^c) = (\Cut_G(S,S^c)+\Cut_G(S^c,S))/2$, and
        \item $\Uncut_G(S) = (\Cut_G(S,S)+\Cut_G(S^c,S^c))/2$. 
    \end{itemize}
    If $G$ has random-walk matrix $\mW$, we may write $\Cut_\mW$ and $\Uncut_{\mW}$ instead of $\Cut_G$ and $\Uncut_G$.
\end{definition}

\begin{definition}[Powering]  For a weighted digraph $G$ with adjacency matrix $\mA$, out-degree matrix $\Dout$, and random-walk matrix $\mW=\mA\Dout^{-1}$, we write $G^\ell$ for the weighted digraph with adjacency matrix $(\mA\Dout^{-1})^\ell\cdot \Dout$ (and thus out-degree matrix $\Dout$ and random-walk matrix $\mW^\ell$).
\end{definition}

With these definitions, the main application of our SV sparsification results is the following:
\begin{theorem}[informal, see~\Cref{thm:dirCut}] \label{thm:dirCut-intro}
    There is a randomized algorithm that, given a strongly connected $n$-node $m$-edge directed graph $G$ with integer edge weights in $[1,U]$, a walk length $\ell$, an error parameter $\eps>0$, and lower bound $s$ on the minimum stationary probability of the random walk on $G$, 
   runs in time
   $O((m+n\eps^{-2})\cdot \poly(\log (U\ell/s))$
   and outputs an $O(n\eps^{-2}\cdot \poly(\log(U\ell/s)))$-edge graph $H$ such that for every two sets $S,T$ of vertices, we have:  
   \[\left|\Cut_{H}(S,T) - \Cut_{G^\ell}(S,T)\right|
    \leq \frac{\eps}{2} \cdot \sqrt{\min\left\{\Cut_{G^\ell}(S),\Uncut_{G^\ell}(S)\right\} \cdot
    \min\left\{\Cut_{G^\ell}(T),\Uncut_{G^\ell}(T)\right\}}.\]
   In particular:
    \[ \left(1-\eps\right)\cdot \Cut_{G^\ell}(S) \leq 
    \Cut_{H}(S)
    \leq  \left(1+\eps\right)\cdot \Cut_{G^\ell}(S),\]
and
    \[ \left(1-\eps\right)\cdot \Uncut_{G^\ell}(S) \leq 
    \Uncut_{H}(S)
    \leq  \left(1+\eps\right)\cdot \Uncut_{G^\ell}(S).\]   
\end{theorem}
Note that when $U,\ell\leq \poly(n)$, $s\geq 1/\poly(n)$, and $\eps\geq 1/\poly(\log n)$, our algorithm runs in time $\tO(m)$.  
For comparison, note that, given a set $S$, we can estimate the cut value for $S$ using random walks in time roughly $\tO(\ell/(\eps^2\Cut_G(S)))$, which is slower when $\ell/\Cut_G(S)$ is $m^{1+\Omega(1)}$. (Note that $\Cut_G(S)$ can be as small as $1/m$.) 

It is also worth comparing to the following approaches that yield {\em high-precision} estimates (i.e. replacing multiplicative error $\eps$ with polynomially small additive error):
\begin{itemize}
\item Use matrix powering via repeated squaring to compute $\mW^\ell$.  This takes time $n^\omega\cdot \log(\ell)$, where $\omega$ is the matrix multiplication exponent. This is slower than our algorithm assuming $\omega>2$ or $m\leq n^{2-\Omega(1)}$. 
\item Use the algorithm of \cite{CKKPPRS18} to obtain a high-precision estimate of the stationary distribution $\mu$ of $G$ in time $\tO(m)$, and then use repeated matrix-vector multiplication to compute $\mW^\ell \mu$.  This takes time $\tO(m\ell)$, so is slower than our algorithm except when $\ell=\polylog(n)$.
\item Use the algorithm of Ahmadenijad et al.~\cite{AKMPSV20}. This also gives high-precision estimates of $\mW^\ell$, and does so in nearly logarithmic space, but the running time is superpolynomial. A running time of $\Omega(m \cdot \ell)$ seems inherent in the approach as it works by reducing to solving a directed Laplacian system of size $m\cdot \ell$.
\end{itemize}
It remains an interesting open problem to estimate any desired entry of $\mW^\ell$ to high precision in nearly linear time.

\paragraph{Other Work on SV Approximation.} 

The definition of SV approximation and some of our results on it (obtained in collaboration between Jack Murtagh and the authors) were first presented in the first author's PhD thesis~\cite{AmirThesis} in August 2020. 
Independently, Kelley~\cite{Kelley21} used a variant of SV approximation to present an alternative proof of a result of \cite{HozaPyVa21}, who used unit-circle approximation to prove that the Impagliazzo-Nisan--Wigderson pseudorandom generator~\cite{ImpagliazzoNiWi94} fools permutation branching programs.  Golowich and Vadhan~\cite{GolowichVa22} used SV approximation and some of our results (presented in the aforementioned thesis) to prove new pseudorandomness properties of expander walks against permutation branching programs.  Most recently, Chen, Lyu, Tal, and Wu~\cite{ChenLyTaWu22} have used a form of SV approximation to present alternative proofs of the results of \cite{AKMPSV20,PyneVa21a}.

\subsection{Singular-Value Approximation}
\label{sec:intro:sv}
In this paper, we present a stronger and more robust notion for addressing the challenge of defining approximation between directed graphs. Specifically, we introduce a novel definition of approximation for 
asymmetric matrices,
which we call {\em singular-value approximation (or SV approximation)}.  

For simplicity in the rest of this introduction, we focus on the case of regular directed graphs, i.e.\ directed graphs where for some value $d\geq 0$, every vertex has in-degree $d$ and out-degree $d$. (In the case of digraphs with non-negative edge weights, we obtain the in- and out-degrees by summing the in-coming or out-going edge weights at each vertex.)  However, all of our results generalize to Eulerian digraphs and some generalize to wider classes of complex matrices.

To introduce SV approximation, let $\mA$ be the adjacency matrix of a $d$-regular digraph, i.e., $\mA$ is a non-negative real matrix where every row and column sum equals $d$. Then the (in- and out-) degree matrix is simply $d\mI$.  Dividing by $d$, it is equivalent to study approximation of the {\em random-walk matrix} $\mW = \mA/d$, which is doubly stochastic, and has degree matrix $\mI$.  

\begin{definition}[SV approximation for doubly stochastic matrices]
\label{def:SV-intro}
For doubly stochastic matrices $\mW$, $\tW\in \R^{n\times n}$ we say that
$\tW$ is an {\em $\eps$-singular-value (SV) approximation of $\mW$}, written
$\tW \svgraph{\eps} \mW$, if for all ``test vectors'' $x,y\in \R^n$, we have
\begin{equation}
        \left|x^\rTrans (\tW - \mW) y\right| \leq \frac{\eps}{4} \cdot \left[ x^\rTrans (\mI - \mW\mW^{\rTrans}) x + y^\rTrans (\mI - \mW^{\rTrans}\mW) y\right]. \label{ineq:SV-intro-doublystoch}
\end{equation}
\end{definition}
This formulation of SV-approximation is one of several equivalent formulations we provide later in the paper (\Cref{SVequiv}). We can equivalently define SV approximation between doubly stochastic matrices by requiring \Cref{ineq:SV-intro-doublystoch} to hold for all complex test vectors $x,y\in \C^n$.  SV-approximation can also be defined equivalently by replacing condition \eqref{ineq:SV-intro-doublystoch} with 
\begin{equation}
        \left|x^{\rTrans} (\tW - \mW) y\right| \leq \frac{\eps}{4} \cdot \sqrt{ \left[x^{\rTrans} (\mI - \mW\mW^{\rTrans}) x\right]\cdot \left[ y^{\rTrans} (\mI - \mW^{\rTrans}\mW) y\right]}. 
        \label{ineq:SV-intro-doublystoch-geometric}
\end{equation}
These two formulations, \eqref{ineq:SV-intro-doublystoch} and \eqref{ineq:SV-intro-doublystoch-geometric}, differ only in using the geometric mean or the arithmetic mean of the terms involving $x$ and $y$ on the right-hand side.  The formulation in terms of the geometric mean implies the one in terms of the arithmetic mean (since the geometric mean is no larger than the arithmetic mean); the converse follows by optimizing over scalar multiples of $x$ and $y$ (as was done in e.g.~\cite{CKPPRSV17,AKMPSV20}).
 Both formulations can be rewritten more simply by noting that
$$ x^{\rTrans} (\mI - \mW\mW^{\rTrans}) x = \|x\|^2 - \|x^{\rTrans} \mW\|^2 \text{ and }
y^{\rTrans} (\mI - \mW^{\rTrans}\mW) y = \|y\|^2 - \|\mW y\|^2,
$$
but the description in terms of quadratic forms will be more convenient for comparison with previous notions of approximation.
In Section~\ref{sec:sv}, we provide more general definitions of SV approximation which also apply to unnormalized directed Laplacians and even to complex matrices.

We prove that SV approximation is strictly stronger than previous notions of spectral approximation considered in the literature, even for undirected graphs, and enjoys several useful properties not possessed by the previous notions.  Most notably, there is a simple black-box reduction from SV-sparsifying Eulerian directed graphs to SV-sparsifying undirected graphs; no such reduction is known for prior notions of asymmetric spectral approximation. 

Furthermore, we give efficient algorithms for working with SV approximation.  These include nearly linear-time algorithms for SV-sparsifying undirected and hence also Eulerian directed graphs (\Cref{thm:sparsify-intro}), as well as random-walk polynomials of directed graphs (\Cref{thm:sparsify-intro}).  We also show that  a simple repeated-squaring and sparsification algorithm for solving Laplacian systems also works for Eulerian digraphs whose random-walk matrix is normal (i.e., unitarily diagonalizable), if we use SV-sparsification at each step (\Cref{thm:squaringBased-intro}). Prior Laplacian solvers for Eulerian graphs are more complex. We elaborate on these results in the next several subsections.

\subsection{Comparison to Previous Notions of Approximation}
Let us compare \Cref{def:SV-intro} to previous definitions of approximation.  

\paragraph{Undirected spectral approximation.} Let's start with the undirected case, where $\mW=\mW^{\rTrans}$.  In this case, it can be shown that we can without loss of generality restrict~\Cref{def:SV-intro} to $x=y$, obtaining the following: $\tW \svgraph{\eps} \mW$ requires that for all $x\in \R^n$,
    \begin{equation}
        \left|x^{\rTrans} (\tW - \mW) x\right| \leq \frac{\eps}{2} \cdot \left[ x^{\rTrans} (\mI - \mW^2)x\right]. \label{ineq:SV-intro-undirected}
    \end{equation}
In contrast, the standard definition of spectral approximation (introduced by Spielman and Teng~\cite{ST04}), which we denote by $\tW\approx_\eps \mW$, is equivalent to requiring that for all $x\in \R^n$, we have
    \begin{equation}
        \left|x^{\rTrans} (\tW - \mW) x\right| \leq \eps \cdot \left[ x^{\rTrans} (\mI - \mW)x\right].
        \label{ineq:ST-intro}
    \end{equation}
To compare inequalities~\eqref{ineq:SV-intro-undirected} and \eqref{ineq:ST-intro}, we write $x=\sum_i c_iv_i$, where $v_1,\ldots,v_n$ is an 
orthonormal eigenbasis for $\mW$ with associated eigenvalues $\lambda_1,\ldots,\lambda_n$. Since $\mW$ is stochastic, $|\lambda_i|\leq 1$ for all $i \in [n]$, the right-hand side of SV inequality~\eqref{ineq:SV-intro-undirected} becomes
$$\frac{\eps}{2}\cdot \sum_{i \in [n]} c_i^2 \cdot (1-\lambda_i^2),$$
whereas the right-hand side of ST inequality~\eqref{ineq:ST-intro}
becomes
$$\eps\cdot \sum_{i \in [n]} c_i^2 \cdot (1-\lambda_i).$$
Since each $|\lambda_i| \leq 1$, the fact that SV approximation implies ST approximation
then follows from 
$$(1-\lambda_i^2) = (1-\lambda_i)(1+\lambda_i) \leq 2(1-\lambda_i).$$
However, we also see that inequality~\eqref{ineq:SV-intro-undirected} can be much stronger than inequality~\eqref{ineq:ST-intro} when $\mW$ has eigenvalues $\lambda_i$ close to -1 (e.g. in a bipartite graph with poor expansion) because then $1-\lambda_i^2$ is close to 0, but $1-\lambda_i$ is bigger than 2. More generally,  
inequality~\eqref{ineq:SV-intro-undirected} requires that $\tW$ approximates $\mW$ very well on every test vector $x$ that is concentrated on the eigenvectors whose eigenvalues have {\em magnitude} close to 1, whereas inequality~\eqref{ineq:ST-intro} only requires close approximation on the (signed) eigenvalues that are close to 1. 

We remark that another way of ensuring $\tW$ preserves unit singular values
is to 
replace $\mW^2$ in the SV inequality~\eqref{ineq:SV-intro-undirected} with
the matrix $|\mW|$ where we replace all eigenvalues of $\mW$ with their absolute value rather than their square,\footnote{Another way of describing $|\mW|$ is 
as the psd square root of the psd matrix $\mW^2$.}
so that we have:
$$x^\top (\mI - |\mW|)x = \sum_{i \in [n]} c_i^2 \cdot (1-|\lambda_i|).$$
Using $|\mW|$ instead of $\mW^2$
results in an equivalent definition up to a factor of 2 in $\eps$, and $\mW^2$ turns out to be convenient to work with.\footnote{$|\mW|=(\mW^2)^{1/2}$, i.e., $|\mW|$ is the PSD square root of $\mW^2$.  In Definition~\ref{def:SV-intro}, we could similarly replace $\mW\mW^\top$ and $\mW^\top\mW$ with their PSD square roots and obtain a definition that is equivalent up to a factor of 2 in $\eps$.}  This viewpoint also explains why we stop at $\mW^2$ in the definition and don't explicitly use higher powers; it is simply a convenient proxy for $|\mW|$, which captures all powers.  Indeed, for all $k\in \N$
$$\mI-\mW^2 \preceq 2\cdot \left(1-|\mW|\right) \preceq 2\cdot \left(\mI-\mW^k\right).$$

\paragraph{Directed spectral approximation.} Turning to previous notions of spectral approximation for directed graphs, \emph{standard approximation}~\cite{CKPPRSV17} generalizes the definition of Spielman and Teng~\cite{ST04} by saying $\tW \approx_\eps \mW$ if for all $x,y\in \R^n$
    \begin{equation}
        \left|x^{\rTrans} (\tW - \mW) y\right| \leq \frac{\eps}{2} \cdot \left[x^{\rTrans} (\mI - \mW) x + y^{\rTrans} (\mI - \mW) y\right]. \label{ineq:strong-intro}
    \end{equation}
Equivalently, we can require that for all $x,y\in \R^n$,
    \begin{equation}
        \left|x^{\rTrans} (\tW - \mW) y\right| \leq \eps \cdot \sqrt{\left[x^{\rTrans} (\mI - \mW) x\right]\cdot\left[y^{\rTrans} (\mI - \mW) y\right]}. \label{ineq:strong-intro-geometric}
    \end{equation}
It can be shown that for undirected graphs, standard approximation is equivalent to the condition of Equation~\ref{ineq:ST-intro}, so we will also refer to it as standard approximation.

The use of different left and right test vectors $x$ and $y$ on the left-hand side is crucial for capturing the asymmetric information in $\tW$ and $\mW$.  As before, if $x$ or $y$ is concentrated on eigenvectors of $\mW$ whose eigenvalues are close to 1, then the right-hand side of ST inequality~\eqref{ineq:strong-intro-geometric} is close to 0 and $\tW$ must approximate $\mW$ very well.  However, like the standard undirected ST inequality~\eqref{ineq:ST-intro}, not much is required on eigenvalues near -1.  Moreover, asymmetric matrices can have eigenvalues that are not real and are equal to or close to complex numbers of magnitude 1.  For example, the eigenvalues of a directed $n$-cycle are the complex $n$'th roots of unity.

To address this issue, {\em unit-circle (UC) approximation}~\cite{AKMPSV20}, written $\tW\uc_\eps \mW$, requires that for all {\em complex} test vectors $x,y\in \C^n$, we have 
    \begin{equation}
        \left|x^* (\tW - \mW) y\right| \leq \frac{\eps}{2} \cdot \left[\|x\|^2+\|y\|^2 -\left|x^*\mW x + y^*\mW y\right|\right]. \label{ineq:UC-intro}
    \end{equation}
That is, we take the complex magnitude of the terms involving $\mW$ on the right-hand side of ST inequality~\eqref{ineq:strong-intro}.  That way, if $x$ and $y$ are concentrated on eigenvectors of $\mW$ that have eigenvalue near some complex number $\mu$ of magnitude 1, we require that $\tW$ approximates $\mW$ very well.  
For example, consider the case where $\mW$ is normal, i.e., has an orthonormal basis of complex eigenvectors $v_1,\ldots,v_n$ and with corresponding complex eigenvalues $\lambda_1,\ldots,\lambda_n$.  Then if we write $x=\sum_i c_i v_i$ and $y=\sum_i d_i v_i$, the right-hand side of UC inequality~\eqref{ineq:UC-intro} becomes:
\begin{equation} \label{expr:UC-intro}
    \frac{\eps}{2}\cdot \sum_{i \in [n]} (|c_i|^2+|d_i|^2) - \left|\sum_{i \in [n]} (|c_i|^2+|d_i|^2)\cdot \lambda_i\right|.
\end{equation}
If $x$ and $y$ are concentrated on eigenvalues $\lambda_i\approx \mu$ where $|\mu|=1$, then this expression will be close to 0.
Unit-circle approximation has valuable properties not enjoyed by standard approximation, in particular being preserved under powering: If $\tW$ is an $\eps$-UC approximation of $\mW$, then for every positive integer $k$, $\tW^k$ is an $O(\eps)$-UC approximation of $\mW^k$; note that the quality of approximation does not degrade with $k$. This property was crucial for the results of \cite{AKMPSV20}.

However, UC approximation has two limitations compared to SV approximation.  First,  UC expression~\eqref{expr:UC-intro} is only small if $x$ and/or $y$ is concentrated on eigenvalues that are all close to the {\em same} point $\mu$ on the complex unit circle. Even in the undirected case, if $x$ and $y$ are  mixtures of eigenvectors with eigenvalue close to 1 and eigenvalue close to -1, then there will be cancellations in the second term of
UC expression~\eqref{expr:UC-intro} and the result will not be small.  Second, some properties of asymmetric matrices are more directly captured by {\em singular values} than eigenvalues, since singular values 
treat the domain and codomain as distinct. For example, the second-largest singular value of $\mW$ equals 1 if and only if there is a probability distribution $\pi$ on vertices that does not mix at all in one step (i.e., $\|\mW\pi - u\| = \|\pi - u\|$, where $u=\vones/n$ and $\pi\neq u$), but the latter can hold even when all nontrivial eigenvalues have magnitude strictly smaller than 1.\footnote{For instance, consider a directed graph on $\{1,2,3,4\}$ with edges $\{(1,2),(1,2),(2,3),(2,4),(4,1),(3,1),(3,4),(4,3)\}$. A walk started on vertex $1$ will not mix at all after the first step.}

To see how SV approximation addresses these limitations, let $\sigma_1,\ldots,\sigma_n\geq 0$ be the singular values of $\mW$, let $u_1,u_2,\ldots,u_n\in \C^n$ the corresponding left-singular vectors of $\mW$, and let $v_1,\ldots,v_n\in \C^n$ the corresponding right-singular vectors.  If we write $x=\sum_i c_iu_i$ and $y=\sum_i d_iv_i$, then the right-hand side of SV inequality~\eqref{ineq:SV-intro-doublystoch} becomes:
\begin{equation} \label{expr:SV-intro}
\frac{\eps}{4}\cdot \left[\sum_{i \in [n]} (|c_i|^2+|d_i|^2)\cdot (1-\sigma_i^2)\right].
\end{equation}
Consequently, SV-approximation requires high-quality approximation if $x$ is concentrated on left-singular vectors of singular value close to 1 and/or $y$ is concentrated on right-singular vectors of singular value close to 1. (For the ``or'' interpretation, use the formulation of SV approximation in terms of inequality~\eqref{ineq:SV-intro-doublystoch-geometric}.)
To compare with UC expression~\eqref{expr:UC-intro}, let us consider what happens with a {\em normal} matrix, where $u_i=v_i$ and $\sigma_i=|\lambda_i|$.  In this case, SV expression~\eqref{expr:SV-intro} amounts to bringing the absolute value of UC expression~\eqref{expr:UC-intro} inside the summation (and squaring, which only makes a factor of 2 difference), to avoid cancellations between eigenvalues of different phases. 

Furthermore, for non-normal matrices, SV approximation retains the asymmetry of $\mW$ even on the right-hand side, by always using $x$ on the left of $\mW$ (thus relating to its decomposition into left singular vectors) and $y$ on the right of $\mW$ (thus relating to its decomposition into right singular vectors).  Indeed, this feature allows us to even extend the definition of SV approximation to non-square matrices. (See Section~\ref{sec:sv}.)

Following the above intuitions, we prove
that SV approximation is indeed strictly stronger than the previous notions of approximation, even for undirected graphs:
\begin{theorem}
    For all doubly stochastic matrices $\mW$ and $\tW$, if $\tW \svn{\eps} \mW$, then $\tW \uc_{\eps} \mW$ (and hence $\tW \approx_\eps \mW$). On the other hand, for every $n\in \N$ there exist random walk matrices $\tW,\mW$ for $n$-node undirected graphs such that $\tW\uc_{O(1/\sqrt{n})}\mW$, but it is not the case that $\tW\svn{.3}\mW$.
\end{theorem}
Since UC approximation implies standard approximation, we likewise separate SV from standard approximation. Finally, we note that our separation implies that several useful properties enjoyed by SV approximation, such as preservation under products (\Cref{prop:UC_prop_fail}), are not satisfied by UC approximation.

\subsection{Properties of SV Approximation}
SV approximation enjoys a number of novel properties not known to be possessed by previous notions of spectral approximation.  Most striking is the fact that directed approximation reduces to undirected approximation. To formulate this, we define the symmetric lift of a matrix:
\begin{definition}
    Given $\mW\in \C^{m\times n}$, let the \emph{symmetric lift} of $\mW$ be defined as
    \[
    \slift{\mW} \defeq \begin{bmatrix}\mzero^{n\times n} & \mW^*\\\mW & \mzero^{m\times m}\end{bmatrix}.
    \]
\end{definition}
Graph theoretically, the symmetric lift of $\mW$ is the following standard operation:  Given our directed graph $G$ on $n$ vertices with random-walk matrix $\mW$, we lift $G$ to an undirected bipartite graph $H$ with $n$ vertices on each side, where we connect left-vertex $i$ to right-vertex $j$ if there is a directed edge from $i$ to $j$ in $G$.  Then $\slift{\mW}$ is the random-walk matrix of $H$. 

\begin{theorem} \label{thm:directed2undirected-intro}
    Let $\mW$ and $\tW$ be doubly stochastic matrices.  Then
    $\mW \svgraph{\eps} \tW$ if and only if $\slift{\tW}\svgraph{\eps} \slift{\mW}$.
\end{theorem}
Thus, for the first time (as far as we know), sparsification of directed graphs reduces directly to sparsification of undirected graphs.  It would be very interesting to obtain a similar reduction for other algorithmic problems in spectral graph theory, such as solving Laplacian systems.

Another novel property of SV approximation is that it is preserved under products:
\begin{theorem} \label{thm:products-intro}
Let $\mW_1,\ldots,\mW_k$ and $\tW_1,\ldots,\tW_k$ be doubly stochastic matrices such that $\tW_i \svgraph{\eps} \mW_i$ for each $i \in [k]$.  Then $\tW_1\tW_2\cdots\tW_k \svgraph{\eps+O(\eps^2)} \mW_1\mW_2\cdots \mW_k$.
\end{theorem}
Notably the approximation error does not grow with the number $k$ of matrices being multiplied.  This property does not hold for UC approximation, only the weaker property of preservation under powering, i.e., $\mW_1=\mW_2=\cdots=\mW_k$ and $\tW_1=\tW_2=\cdots=\tW_k$.

In addition, SV approximation is preserved under multiplication on the left and right by arbitrary matrices of bounded spectral norm. Indeed, it can be seen as the ``closure'' of standard approximation under this operation (up to a factor of 2).
\begin{theorem}\label{thm:rotate_invariant_informal}~
The following hold for all doubly stochastic matrices $\mW$ and $\tW$:
\begin{enumerate}
    \item If $\tW \svgraph{\eps} \mW$ then for all complex matrices $\mU$ and $\mV$ of spectral norm at most 1, we have $\mU\tW\mV \svgraph{\eps} \mU\mW\mV$, and hence  $\mU\tW\mV \approx_{\eps} \mU\mW\mV$.
    \item If for all complex matrices $\mU$ and $\mV$ of spectral norm at most 1, we have $\mU\tW\mV \approx_{\eps} \mU\mW\mV$ then
    $\tW \svgraph{2\eps} \mW$.
\end{enumerate}
\end{theorem}
Since $\mU\mW\mV$ and $\mU\tW\mV$ need not be doubly stochastic matrices, \Cref{thm:rotate_invariant_informal} uses the generalization of SV approximation to more general matrices, which can be found in Section~\ref{sec:sv}.

Recall that standard spectral sparsifiers~\cite{ST04} are also {\em cut sparsifiers}~\cite{BenczurKa00}.
That is, if $\widetilde{G}$ is an $\eps$-approximation of $G$, then for every set $S$ of vertices, 
the weight of the cut $S$ in $\widetilde{G}$ is within a $(1\pm \eps)$ factor of the weight of $S$ in $G$.  
Indeed, if we take the test vector $x$ to be the characteristic vector of the set $S$ in inequality~\eqref{ineq:ST-intro},
we obtain
\begin{equation}
\left|\Cut_{\tG}(S) - \Cut_G(S)\right| 
\leq \eps\cdot 
\Cut_G(S),
\label{eq:cut_approximator}
\end{equation}
where $\Cut(\cdot)$ is as in \Cref{def:Cut}.

Similarly, we can obtain a combinatorial consequence of SV approximation, by taking $x$ to be a characteristic vector of a set $S$ of vertices and taking $y$ to be a characteristic vector of a set $T$ of vertices.  This yields:

\begin{proposition} \label{prop:SV-combinatorial-intro}
Let $\tW$ and $\mW$ be doubly stochastic $n\times n$ matrices and suppose that $\tW \svgraph{\eps} \mW$.  Then for every two subsets $S,T\subseteq [n]$, we have 

\[\left|\Cut_{\tW}(S,T) - \Cut_\mW(S,T)\right| 
\leq \frac{\eps}{2} \cdot \sqrt{\Cut_{\mW\mW^\top}(S) \cdot \Cut_{\mW^\top\mW}(T)}. 
\]
\end{proposition}
Note that $\mW^\top\mW$ (resp., $\mW \mW^\top$) is the transition matrix for the {\em forward-backward walk} (resp. backward-forward walk), namely where we take one step using a forward edge of the graph followed by one step using a backward edge.

Let us interpret Proposition~\ref{prop:SV-combinatorial-intro}. 
First, consider the case that $\mW=\mJ$, the matrix with every entry equal to $1/n$ (the random-walk matrix for the complete graph with self-loops).  Then the distribution $\muedge$ on pairs $(i,j)$ in the definition of $\Cut_\mW$ (\Cref{def:Cut}) has $i$ and $j$ as uniform and independent vertices, and the same is true for
$\Cut_{\mW \mW^\top}$ and
$\Cut_{\mW^\top \mW}$.  Thus,
Proposition~\ref{prop:SV-combinatorial-intro} says:
\[\left|\Cut_{\tW}(S,T) - \mu(S)\cdot \mu(T)\right| 
\leq \frac{\eps}{2} \cdot \sqrt{\mu(S)\cdot (1-\mu(S)) \cdot \mu(T)\cdot (1-\mu(T))},\]
where $\mu(S)=|S|/n$ and $\mu(T)=|T|/n$ are the stationary probabilities of $S$ and $T$, respectively.  This amounts to a restatement of the Expander Mixing Lemma (cf., \cite[Lemma 4.15]{Vad}); indeed 
$\tW \svgraph{\eps} \mJ$ if and only if $\tW$ is a spectral expander with all nontrivial singular values at most $\eps/2$.

Next, let's consider the case that $T=S^c$.  Since 
$\Cut_{\mW\mW^\top}(S) = \Cut_{\mW\mW^\top}(S^c)$, 
SV approximation implies that: 
\begin{equation}
\left|\Cut_{\tW}(S) - \Cut_{\mW}(S)\right| 
\leq \frac{\eps}{2}\cdot
\sqrt{
\Cut_{\mW\mW^\top}(S)
\cdot 
\Cut_{\mW^\top\mW}(S)}
\label{eq:sv_conductance_bound}
\end{equation}

We claim that \eqref{eq:sv_conductance_bound} is stronger than the standard notion of a cut approximator 
\eqref{eq:cut_approximator}.
Indeed, it can be shown that
$$\Cut_{\mW\mW^\top}(S) \leq 2\cdot\Cut_\mW(S),$$
and similarly for $\Cut_{\mW^\top\mW}(T)$.
The reason is that if a backward-forward walk crosses between $S$ and $S^c$, then it must cross between $S$ and $S^c$ in either the first step or in the second step.
Similar reasoning shows that
$$\Cut_{\mW\mW^\top}(S) \leq 2\cdot\Uncut_\mW(S),$$
and similarly for $\Cut_{\mW^\top\mW}(T)$.
Thus SV approximation also implies:
\begin{equation}
\left|\Uncut_{\tG}(S) - \Uncut_G(S)\right| 
\leq \eps\cdot 
\Uncut_G(S),
\label{eq:uncut_approximator}
\end{equation}
Thus, we conclude that an SV-approximator not only approximates every cut to within a small additive error that is scaled by the weight of the cut edges (as in \eqref{eq:cut_approximator}), but also scaled by the weight of the {\em uncut} edges.

\subsection{Algorithmic Results}

Even though SV approximation is stronger than previously considered notions of spectral approximation, we show that it still admits sparsification: 
\begin{theorem} \label{thm:sparsify-intro}
There is a randomized nearly-linear time algorithm that given a regular directed graph $G$ with $n$ vertices and $m$ edges, integer edge weights in $[0,U]$, and random-walk matrix $\mW$,  
and $\eps>0$,
 whp outputs a  weighted graph $\widetilde{G}$ with at most 
 $O(n\eps^{-2}\cdot \poly(\log(nU)))$ 
 edges such that its random-walk matrix $\tW$ satisfies $\tW \svgraph{\eps} \mW$.
\end{theorem} 
A more general theorem that also applies to Eulerian digraphs is stated in the main body of the paper (\Cref{thm:rwSparsifier}).
Prior to this work, it was open whether or not even UC-sparsifiers with $O(n \cdot \poly(\log n,1/\eps))$ edges existed for all unweighted regular digraphs.
Instead, it was only known how to UC-sparsify powers of a random walk matrix 
in such a way that the number of edges increases by at most a polylogarithmic factor
compared to the original graph (rather than decrease the number of edges)~\cite{AKMPSV20}. 

By Theorem~\ref{thm:directed2undirected-intro}, it suffices to prove \Cref{thm:sparsify-intro} for undirected bipartite graphs. We obtain the latter via an undirected sparsification algorithms based on expander partitioning \cite{ST04}.
It remains an open question whether algorithms based on edge sampling can yield SV approximation or unit circle approximation, even in undirected graphs.  The standard approach to spectral sparsification of undirected graphs via sampling, namely keeping each edge independently with probability proportional to its effective resistance~\cite{SS08}, does not work for SV or UC approximation. For example, this method does not exactly preserve degrees, which we show is necessary for SV sparsification (\Cref{lem:svStationary}).\footnote{Unlike standard spectral approximation, degrees cannot be fixed just by adding self loops; indeed, self-loops ruin bipartiteness and periodicity, which are properties that UC and SV approximation retain (as they are captured by eigenvalues like -1 or other roots of unity).}

However, we remark that the work of Chu, Gao, Peng, Sawlani, and Wang~\cite{CGPSSW18}
does yield something closer to sparsification via \emph{degree preserving} sampling for standard approximation~\cite{CKPPRSV17} but not unit circle approximation. They show that if one has a directed graph and decomposes it into short ``cycles'' without regard for the direction of the edges on the cycle, then one can sparsify by randomly eliminating either the clockwise or counterclockwise edges on each such cycle. 
We build on their procedure and use it to obtain SV sparsification (and hence, unit circle) by showing that this technique obtains SV approximation, even if the cycles are not short, as long as (a) all the cycles are within expanding subgraphs, and (b) the cycles alternate between forward and backward edges.  (Note that such alternating cycles in a directed graph correspond to ordinary cycles in the undirected lift given by Theorem~\ref{thm:directed2undirected-intro}.)
\begin{center}
    \begin{tabular}{c|c|c|c|c}
  Sparsity & Approximation & Time & Subgraph? & Citation\\\hline
  $O(n\eps^{-2}\log^{c}n)$ & Standard & $O(m\log^cn)$ & No & \cite{CKPPRSV17}\\
  $O(n\eps^{-2}\log^cn)$ & Standard & $O(m^{1+o(1)})$ & Yes & \cite{CGPSSW18,PY19}\\
  $O(n\eps^{-2}\log^4 n)$ & Standard & $O(nm)$ & Yes & \cite{CGPSSW18}\\
  $O(n\eps^{-2}\log^{12} n)$ & SV & Existential & Yes & \Cref{cor:sparseExis}\\
  $O(n\eps^{-2}\log^{20} n)$ & SV & $O(m\log^{7}n)$ & Yes & \Cref{cor:sparseFast}\\
\end{tabular}
\end{center}

Given
\Cref{thm:sparsify-intro}, we obtain our algorithm for longer walks 
(\Cref{thm:dirCut-intro})
as follows:
\begin{enumerate}
\item First, we show that we can SV-sparsify the {\em squares} of random-walk matrices of Eulerian digraphs; we follow the approach of~\cite{CKKPPRS18} by locally sparsifying the bipartite complete graphs that form around each vertex when squaring, and then applying \Cref{thm:sparsify-intro} to globally sparisfy further. We likewise show the ``derandomized square'' approach used in \cite{RozenmanVa05,PS13,MRSV2,AKMPSV20} gives a square sparsifier.
\item Then we SV-sparsify arbitrary powers of 2 by repeatedly squaring and sparsifying, using the fact that SV approximation is preserved under powering. During this process, we need to ensure that the ratio between the largest and smallest edge weights remains bounded.  We 
do this by restricting to graphs that have second-largest singular value bounded away from 1 by $1/\poly(nU\ell)$, which allows us to discard edge weights that get too small and make small patches to preserve degrees.  We can achieve this assumption on the second-largest singular value by adding a small amount of laziness to our initial graph.
\item Then to sparsify arbitrary powers $\mW^\ell$, we can multiply sparsifiers for the powers of 2 appearing in the binary representation of $\ell$.  For example, to get a sparsifier for $\mW^7$, we multiply sparsifiers for $\mW^4$, $\mW^2$, and $\mW^1$, sparsifying and eliminating small edge weights again in each product. The use of SV approximation plays an important role in the analysis of this algorithm, because  it has the property that the product of the approximations of the powers still approximates the product of the true powers (\Cref{thm:products-intro}).

\item Given \Cref{thm:sparsify-intro}, we obtain Theorem~\ref{thm:dirCut-intro} for general directed graphs by using \cite{CKKPPRS18} to compute a high-precision estimate of the stationary distribution, which allows us to construct an Eulerian graph whose random-walk matrix closely approximates that of the original graph.  SV-sparsifying the $\ell$'th power of the Eulerian graph gives us a graph all of whose $\Cut$ and $\Uncut$ values approximate the $\ell$'th power of our input graph. 
The use of \cite{CKKPPRS18} to estimate the stationary distribution and the introduction of laziness to $\mW$ both incur a small
additive error $\delta$, but we can absorb that into $\eps$ by 
setting $\delta=1/\poly(nU/s)$ and
observing
that $\Cut_{G^\ell}(S)$ and $\Uncut_{G^\ell}(S)$ are at
least $1/\poly(nU/s)$ (if nonzero).
\end{enumerate}

Our final contribution concerns algorithms for solving directed Laplacian systems. The recursive identities used for solving undirected Laplacian systems, while behaving nicely with respect to PSD approximation, do not behave as nicely with respect to the previous approximation definitions for directed graphs.
This led to different, more sophisticated recursions with a more involved analysis of the error~\cite{CKPPRSV17,CKKPPRS18,AKMPSV20,KMPG22}.
We make progress towards simplifying the recursion and analysis of solving directed Laplacian linear systems in the following way. We show that a simpler recursion, a variant of the one used by Peng and Spielman~\cite{PS13} (\Cref{eqn:pengspielman-intro} below), and a simpler analysis suffice if the directed Laplacian is normal (i.e., unitarily diagonalizable) and we perform all sparsification with respect to SV approximation.  Note that this result is the only result in our paper that relies on a normality assumption; the aforementioned sparsification results hold for all Eulerian directed graphs.

\begin{theorem}\label{thm:squaringBased-intro}
For a doubly stochastic normal matrix $\mw \in \R^{n\times n}$ with $\| \mw \| \leq 1$, let $\mw = \mw_0, \ldots, \mw_{k-1}$ be a sequence of matrices such that for $\eps \leq 1/4k$ we have
\begin{equation} \label{eqn:repeatedsquare-intro}
\mw_{i} \sv_{\epsilon} \mw_{i-1}^2 \quad \forall 0 < i < k,
\end{equation}
and 
\begin{equation} \label{eqn:pengspielman-intro}
\mproj_i = \frac{1}{2} \left[\mI + (\mI + \mw_i) \mproj_{i+1} (\mI + \mw_i)\right] \quad \forall 0 \leq i < k
\end{equation}
defining the Peng-Spielman squaring recursion. \\
Then, for a matrix $\mproj_k$, such that $ \left\| (\mI - \mw^{2^k})^{\frac{1}{2}} \left[ \mproj_k - (\mI - \mw^{2^k})^+ \right] (\mI - \mw^{2^k})^{\frac{1}{2}}\right\| \leq O(k \epsilon)$, we have
\begin{equation} \label{eqn:precond-intro}
\left\| \mproj_0 (\mI - \mw) - \mI \right\|_{\mb}
    =
    \left\| (\mI - \mw)^{\frac{1}{2}} \left[ \mproj_0 - (\mI - \mw)^+ \right] (\mI - \mw)^{\frac{1}{2}}\right\| 
    \leq 
    O(k^2 \epsilon)
\end{equation}
where $\mb = ((\mI - \mw)^{1/2})^* (\mI - \mw)^{1/2}$.
\end{theorem}

\Cref{thm:squaringBased-intro} says that that we can compute a good {\em preconditioner} $\mP_0$ for the Laplacian $\mI-\mW$ (eq.~\ref{eqn:precond-intro}) by repeatedly computing SV-approximate squares (eq.~\ref{eqn:repeatedsquare-intro}) and use  the simple recurrence (eq.~\ref{eqn:pengspielman-intro}, starting with a preconditioner $\mP_k$ for a sufficiently large power of $\mW$.
Generally, $\mP_k$ is easy to obtain for a large enough $k=O(\log n)$ since $\mW^{2^k}$ is well-approximated by a complete graph (assuming the original graph is connected and aperiodic).

\subsection{Open Problems}
One open problem is to determine whether or not it is possible to obtain {\em linear-sized} sparsifiers.  Recall that undirected graphs have sparsifiers with respect to standard spectral approximation that have only $O(n/\eps^2)$ nonzero edge weights~\cite{BSS12}.
If this result could be extended to obtain linear-sized SV-sparsifiers of undirected graphs, we would also have linear-sized sparsifiers for directed graphs by Theorem~\ref{thm:directed2undirected-intro}, which would be a new result even for standard approximation~\cite{CKPPRSV17}.

\subsection{Roadmap}
The rest of the paper is organized as follows:
\begin{itemize}
    \item \Cref{sec:prelim} gives definitions we use throughout the paper.
    
    \item \Cref{sec:sv} defines SV-approximation and provides several equivalent characterizations of it which are useful when analyzing algorithms and also interesting in their own right. These results immediately imply a reduction from the directed to undirected case for SV approximation.

    \item \Cref{sec:sparsification} gives a direct SV sparsification algorithm based on sparsifying alternating cycles within expanders, and shows that the derandomized square gives SV square sparsifiers, and shows how to sparsify powers of random walk matrices with respect to SV approximation.
    
    \item \Cref{sec:normal} gives our simpler algorithm and analysis for solving Laplacian linear systems when the matrix is unitarily diagonalizable (normal).
    
    \item In \Cref{app:SVproofs}, we prove equivalences, properties, and separations, including several deferred proofs from the body of the paper.
\end{itemize}

\section{Preliminaries}\label{sec:prelim}

We begin by providing notation and definitions that we will use throughout the paper.

\paragraph{Complex Numbers.}
Let $\C$ denote the set of complex numbers. For $z\defeq a+bi\in \C$, let $z^* \defeq a-bi$ and $\Re[z] \defeq a$ and $\Im[z] \defeq b$. Define the {\em magnitude} of $z$ as $|z| \defeq \sqrt{z^*z}$.

\paragraph{Matrix Notation.} 
For $x \in \C^n$ we let $\mdiag(x) \in \C^{n \times n}$ denote the diagonal matrix with $[\mdiag(x)]_{i,i} = x_i$ for all $i \in [n]$. 
For complex matrix $\mM$ that is not necessarily square, we write $\mM^*$ to denote the conjugate transpose of $\mM$ and $\mM^+$ to denote the Moore-Penrose pseudoinverse of $\mM$.
We let $\lker(\mM)$ and $\rker(\mM)$ be the left and right kernels of $\mM$ respectively.
We say $\mM \in \C^{n \times n}$ is \emph{Hermitian} if $\mM^*=\mM$. 
$\mM$ is {\em positive semidefinite (PSD)} if it is Hermitian and for every vector $x\in \C^n$, we have
$x^*\mM x\geq 0$.
$\mM$ is {\em normal} if $\mM^*\mM=\mM\mM^*$.
For a square matrix $\mM$, let $\mS_{\mM} = (\mM+\mM^*)/2$ denote its symmeterization. Note that $\mS_{\mM}$ is Hermitian, and if $\mM$ is a scalar then $\mS_{\mM} = \Re[\mM]$.  
Given any $\mM \in \C^{n \times n}$ and subsets $S,T\subset [n]$, we let $\mM_{S,T} \in \C^{S\times T}$ be the submatrix of $\mM$ with rows specified by $S$ and columns specified by $T$. We let $S^c \defeq [n]\setminus S$ be the complement of $S$. 

\paragraph{Schur Complements.} Given an $n\times n$ matrix $\mM$ and $S\subset [n]$, we define the \emph{Schur Complement of $\mM$ onto $S$} as
\[\SC_S(\mM) = \mM_{S^c,S^c}-\mM_{S^c,S}\mM_{S,S}^+\mM_{S,S^c}.
\]

\paragraph{Eigenvalues and Singular Values.} For a symmetric matrix $\mA$ with nonnegative entries, let $\lambda_i(\mA)$ denote the $i$th largest eigenvalue of $\mA$. For a matrix $\mA$, let $\sigma_i(\mA) \defeq \lambda_i(\mA^*\mA)$. Given an undirected graph $G$ with adjacency matrix $\mA$ and degree matrix $\mD\defeq \mdiag(\mA \vones)$, let $\lambda(G)\defeq \lambda_2(\mD^{+/2} \mA\mD^{+/2})$.
    
\paragraph{Norms.}   
Throughout we use $\|\cdot\|$ to denote the spectral norm, where for any $\mA \in \C^{n\times m}$,
\[
\|\mA\| \defeq \sup_{x \in \C^m \setminus \{\vzero\}}\frac{\|\mA x\|}{\|x\|}
\,.
\]

\paragraph{L\"owner Order.}
Given Hermitian $\mA,\mB \in \C^{n\times n}$ we write $\mA \preceq \mB$ if $\mB-\mA$ is PSD, i.e., for every $x\in \C^n$ we have $x^*\mA x \leq x^*\mB x$.

\section{Singular Value Approximation}\label{sec:sv}

In this section we formally define the notions of approximation we work with throughout the paper (\Cref{sec:sv:matrix_approx}), provide a number of equivalent definitions of SV approximation (\Cref{sec:sv:equivalences}), and give key properties of SV approximation  (\Cref{sec:sv:prop}). All proofs in this section are deferred to their analogous subsection in \Cref{app:SVproofs}.

\subsection{Matrix Approximation}
\label{sec:sv:matrix_approx}

First we define a general notion of matrix approximation with respect to arbitrary PSD matrices (\Cref{def:leftrightapprox}). This definition and the equivalences are a generalization of those established in \cite{CKPPRSV17}.

Intuitively, we say that a matrix $\tA$ is an \emph{$\epsilon$-approximation of $\mA$ with respect to error matrices $\mEL$ and $\mER$} if we can bound the bilinear form of their difference, i.e., $x^* (\tA - \ma) y$, by the quadratic forms of $x$ with $\mEL$ and $y$ with $\mER$, i.e., $x^* \mEL x$ and $y^* \mER y$. We use the term \emph{error matrix} here to distinguish this notion of approximation from standard approximation (\Cref{def:stdapprox}) which we define later.

\begin{definition}[Matrix approximation] \label{def:leftrightapprox}
    Let $\mA, \tA \in \C^{m\times n}$, and let $\mEL\in \C^{m\times m},\mER\in \C^{n\times n}$ be PSD matrices.
    For $\eps\geq 0$, we say that $\tA$ is an \emph{$\eps$-approximation of $\mA$ with respect to error matrices  $\mEL$ and $\mER$} if any of the following equivalent conditions hold:
    \begin{enumerate}
\item 
$\left|x^*(\tA-\mA)y\right|\leq \frac{\eps}{2}\left(x^*\mEL x+y^*\mER y\right)$ for all $x \in \C^m,y\in \C^n$ 
    \item $\left|x^*(\tA-\mA)y\right|\leq \eps \cdot \sqrt{x^*\mEL x}\cdot \sqrt{y^*\mER y}$ for all $x \in \C^m,y\in \C^n$ 
\item $\left\|\mEL^{+/2} (\tA-\mA) \mER^{+/2}\right\| \leq \eps$,
$\lker(\tA-\mA) \supseteq \ker(\mEL)$, and
$\rker(\tA-\mA) \supseteq \ker(\mER)$.

\end{enumerate}
If $\mEL=\mER$, we say that $\tA$ is an \emph{$\eps$-approximation of $\mA$ with respect to error matrix $\mE$}.
\end{definition}

In the following \Cref{lem:Hermitianapprox} we specialize and simplify the equivalences presented in \Cref{def:leftrightapprox} to Hermitian matrices.

\begin{lemma}[Hermitian matrix approximation] \label{lem:Hermitianapprox}
    Let $\mA, \tA \in \C^{m\times m}$ be Hermitian matrices and let $\mE\in \C^{m\times m}$ be a PSD matrix.
    Then the following are equivalent conditions.
    
    \begin{enumerate}
        \item $\tA$ is an $\eps$-approximation of $\mA$ with respect to error matrix $\mE$.
        \item  For all $x \in \C^m$ we have $\left|x^*(\tA-\mA)x\right|\leq \eps\cdot \left(x^*\mE x\right)$.
        \item $-\eps\mE \preceq \tA-\mA \preceq \eps\mE$.
    \end{enumerate}
\end{lemma}

Next, for matrices $\ma, \tA, \mD \in \C^{m \times m}$ we define a natural notion of approximation between $\mD - \mA$ and $\mD - \tA$ which we call \emph{standard approximation} (\Cref{def:stdapprox}). The choice of this name is because when $\mA$ and $\tA$ are adjacency matrices of undirected graphs with the same degrees, standard approximation coincides with spectral approximation of the Laplacian matrices of the associated graphs. Further, when $\mA$ and $\tA$ are adjacency matrices of directed Eulerian digraphs with the same degrees then standard approximation coincides with $\epsilon$-approximation from \cite{CKPPRSV17}. Standard approximation generalizes these two cases in a natural way even when $\mD$, which we call a {\em degree matrix}, is not diagonal (as it was in these two cases and as we will often choose it to be). We use the term degree matrix to emphasize when we are using standard approximation (\Cref{def:stdapprox}) and its variants rather than (\Cref{def:leftrightapprox}). 

We now present a series of definitions and equivalences and properties. For clarity, the reader can initially think of the case that $\mD=\mI$ and $\mA=\mW$ is a doubly stochastic matrix.
\begin{definition}[Standard approximation] \label{def:stdapprox}
    Let $\mA,\tA \in \C^{m\times m}$
    and $\mD\in \C^{m\times m}$ be a PSD matrix.
     For $\eps\geq 0$, we say $\tA$ is a \emph{standard $\eps$-approximation of $\mA$ with respect to degree matrix $\mD$} if $\mE=\mD - \mS_\mA$ is PSD and $\tA$ is an $\eps$-approximation of $\mA$ with respect to error matrix $\mE$.\footnote{Recall by \Cref{sec:prelim} that $\mS_\mA \defeq (\mA+\mA^\cTrans)/2$.}
\end{definition}

With matrix approximation and standard approximation established, we can now define unit-circle (UC) approximation and singular-value (SV) approximation. UC approximation as we present it here, is a generalization of UC approximation as it was introduced for random walk matrices in~\cite{AKMPSV20}. As discussed in the introduction, SV approximation is a new notion of approximation introduced in this paper; it has a number of natural desirable properties and facilitates our results on sparsification (\Cref{sec:sparsification}) and linear system solving (\Cref{sec:normal}).

\begin{definition}[UC approximation] \label{def:ucapprox}
    Let $\mA,\tA \in \C^{m\times m}$ and $\mD\in \C^{m\times m}$ be PSD.
     For $\eps\geq 0$, we say $\tA$ is a \emph{unit-circle (UC) $\eps$-approximation of $\mA$ with respect to degree matrix $\mD$} if for every $z\in \C$ with $|z|=1$, $z\tA$ is a standard $\eps$-approximation of $z\mA$ with respect to degree matrix $\mD$. If this holds, we write $\tA\uc_\eps \mA$ with respect to $\mD$. We omit \say{respect to $\mD$} when $\mD = \mI$.
\end{definition}
Note that for doubly stochastic matrices $\tW,\mW$, the above definition (with $\mD=\mI$) corresponds to requiring for every $z\in \C$ and $x,y\in \C^n$,
\begin{equation}
\label{eq:UC-formula}
\left|x^*(z\tW-z\mW)y\right| \leq \frac{\eps}{2}\left(\|x\|^2+\|y\|^2-x^*z\mW x-y^*z\mW y\right)\,.
\end{equation}
Optimizing over $z$ in \eqref{eq:UC-formula} implies \Cref{ineq:UC-intro}.

\begin{definition}[SV approximation] \label{def:svapprox}
    Let $\mA,\tA \in \C^{m\times n}$ and let $\Din\in \C^{m\times m}$ and $\Dout\in \C^{n\times n}$ be PSD matrices.
     For $\eps\geq 0$, we say $\tA$ is a \emph{$\eps$-singular-value (SV) approximation of $\mA$ with respect to degree matrices $\Din$ and $\Dout$}, 
     if  
     \begin{enumerate}
         \item $\ker(\Din)\subseteq \lker(\mA)$, and $\ker(\Dout)\subseteq\rker(\mA)$,
         \item $\mEL = \Din -\mA\Dout^{+}\mA^*$ and
     $\mER = \Dout-\mA^*\Din^{+}\mA$ are PSD, and 
        \item 
     $\tA$ is an $\eps/2$-approximation of $\mA$ with respect
     to error matrices $\mEL$ and $\mER$.
     \end{enumerate}
    If this holds, we write
    $\tA \svgeneral{\Din}{\Dout}{\eps} \mA$.
    If $\mA\in \R^{n\times m}_{\geq 0}$ and $\Din = \mdiag(\mA \vones_n)$, $\Dout = \mdiag(\vones_m^{\rTrans} \mA)$, then 
    we write $\tA \svgraph{\eps} \mA$.
    If $\Din=\mI_m$ and $\Dout=\mI_n$, then
    we write $\tA \svn{\eps} \mA$, which we call \emph{normalized SV approximation}.
    If $m=n$ and $\Din=\Dout=\mD$, then we write
    $\tA \svEulerian{\mD}{\eps} \mA$.
\end{definition}
Note that normalized SV approximation does not require the relevant matrices to be non-negative, whereas $\tA\svgraph{\eps}\mA$ is only defined for real non-negative matrices (such as the adjacency matrices of graphs). The maximally general definition captures both cases, so we will prove properties with respect to this notion and note their implications for the specialized notions.

\subsection{Equivalent Definitions of SV Approximation}
\label{sec:sv:equivalences}

Here we give several equivalent formulations of SV approximation. We first give conditions under which the error matrices $\mEL$ and $\mER$ in the definition of SV approximation are PSD.
\begin{restatable}[Conditions for SV approximation to be defined]{lemma}{svDefined} \label{lem:svdefined}
    Let $\mA,\tA \in \C^{m\times n}$, 
    and let $\Din\in \C^{m\times m}$ and $\Dout\in \C^{n\times n}$ be PSD matrices such that
    $\ker(\Din)\subseteq \lker(\mA)$ and $\ker(\Dout)\subseteq\rker(\mA)$.
    Then the following are equivalent:
    \begin{enumerate}
        \item $\sigmamax(\Din^{+/2} \mA \Dout^{+/2})\leq 1$. \label{itm:norm}
        \label{itm:specnorm}
        \item $\Din -\mA\Dout^{+}\mA^*$ is PSD. \label{itm:leftpsd}
        \item $\Dout-\mA^*\Din^{+}\mA$ is PSD.
        \label{itm:rightpsd}
        \item 
        For some scalar $z\in \C$ with $|z|=1$,
        $\begin{bmatrix}
    \mdin & z\ma \\
    z^\cTrans\ma^\cTrans & \mdout
    \end{bmatrix}$ is PSD. \label{itm:twobysomez} 
        \item For every scalar $z\in \C$ with $|z|\leq 1$,
        $\begin{bmatrix}
    \mdin & z\ma \\
    z^\cTrans\ma^\cTrans & \mdout
    \end{bmatrix}$ is PSD. \label{itm:twobytwoallz}
    \end{enumerate}
    Suppose further that $\Din=\mdiag(\din)$ and $\Din=\mdiag(\dout)$
    for $\din\in \R_{\geq 0}^m$, $\dout\in \R_{\geq 0}^n$.
    Then Condition~\ref{itm:graph} below implies Condition~ \ref{itm:diagdom} below, which implies Conditions~\ref{itm:norm}--\ref{itm:twobytwoallz} above.
    \begin{enumerate}
        \item $\mA$ is nonnegative, $\dout = \mA \vones_n$, $\din = \vones_m^{\rTrans} \mA$. \label{itm:graph}
        \item For all $i\in [n]$, $(\dout)_i\geq \sum_j |\mA_{i,j}|$, and for all $j\in [m]$, $(\din)_j\geq \sum_i |\mA_{i,j}|$.
        \label{itm:diagdom}
    \end{enumerate}
\end{restatable}

Next we give several equivalent definitions of SV approximation.

\begin{restatable}[Equivalent formulations of SV approximation]{lemma}{svLifts}
    \label{SVequiv}
    Let $\mA,\tA \in \C^{m\times n}$ and let $\Din\in \C^{m\times m}$ and $\Dout\in \C^{n\times n}$ be PSD matrices.
    Then the following are equivalent
    \begin{enumerate}
        \item \label{SVequiv:main} $\tA \svgeneral{\Din}{\Dout}{\eps} \mA$. 
    
    \item \label{SVequiv:bip} $\slift{\tA} \svEulerian{\mD}{\eps} \slift{\tB}$, where 
    $$\slift{\tA} 
    = \begin{bmatrix}
     \mzero^{m \times m} & \mA\\
    \mA^* & \mzero^{n \times n}
    \end{bmatrix},\  
    \slift{\mA}
    = \begin{bmatrix}
     \mzero^{m \times m} & \tA\\
    \tA^* & \mzero^{n \times n}
    \end{bmatrix},\ 
    \mD
    = \begin{bmatrix}
     \Din & \mzero^{m \times n}\\
    \mzero^{n \times m} & \Dout
    \end{bmatrix}.$$
    
    \item\label{SVequiv:somez} For some scalar $z\in \C$ with $|z|=1$, $\tC$ is $\eps/2$-approximation of $\mC$ with respect to
    error matrix $\mE$, where 
    $$\mC
    = \begin{bmatrix}
     \mzero^{m \times m} & z\mA\\
    \mzero^{n \times m} & \mzero^{n \times n}
    \end{bmatrix},\  
    \tC
    = \begin{bmatrix}
     \mzero^{m \times m} & z\tA\\
    \mzero^{n \times m} & \mzero^{n \times n}
    \end{bmatrix},\ 
    \mE =
    \begin{bmatrix}
    \mdin & z\ma \\
    z^\cTrans\ma^\cTrans & \mdout
    \end{bmatrix}
    $$ 
    
    \item\label{SVequiv:allz} \Cref{SVequiv:somez} holds for every $z \in \C$ such that $|z| \leq 1$.
    \end{enumerate}
\end{restatable}
Each formulation of SV approximation in \Cref{SVequiv} has useful properties. \Cref{SVequiv:bip} implies that SV approximation between directed graphs is equivalent to a natural related statement between \textit{undirected} graphs, which we use for sparsification (See \Cref{sec:sparsification}). A version of \Cref{SVequiv:bip} is not known to hold for prior definitions of approximation between directed graphs, such as standard approximation and unit circle approximation. 
\Cref{SVequiv:somez} and
\Cref{SVequiv:allz} have an error matrix $\mE$ that is linear in $\ma$, $\mdin$, and $\mdout$, which enables short proofs of properties such as summability. In addition, \Cref{SVequiv:somez} and \Cref{SVequiv:allz} characterize SV approximation of $\mA$ in terms of $\eps$-approximation of the \say{asymmetric lift} of $\mA$ with respect to its symmetrization. This enables us to leverage results developed for $\eps$-approximation, such as preservation under Schur complements (\Cref{thm:schurStandard}).

We next note the relation between SV and normalized SV approximation.
\begin{restatable}{lemma}{normalizedSV}\label{lem:normalizedSV}
Let $\mA,\tA \in \C^{m\times n}$, and let $\Din\in \C^{m\times m}$ and $\Dout\in \C^{n\times n}$ be PSD matrices such that $\ker(\Din)\subseteq \lker(\mA)$, and $\ker(\Dout)\subseteq\rker(\mA)$.
    Let $\mN = \Din^{+/2} \mA \Dout^{+/2}$
    and $\tN = \Din^{+/2} \tA \Dout^{+/2}$
    Then 
     $\tA \svgeneral{\Din}{\Dout}{\eps} \mA$
     if and only if
     $\tN \svn{\eps} \mN$.
\end{restatable}
Note that if $\mA$ is nonnegative and $\Din = \mdiag(\mA \vones_n)$, $\Dout = \mdiag(\vones_m^{\rTrans} \mA)$, then $\mN$ and $\tN$ are the normalized adjacency matrices of $\mA$ and $\tA$. Furthermore, when $\mA,\tA$ are the adjacency matrices of regular digraphs, then $\mN$ and $\tN$ are the random-walk matrices of $\mA$ and $\tA$ respectively. If instead $\mA$ and $\tA$ are the adjacency matrices of Eulerian digraphs, we obtain that $\mN$ and $\tN$ are similar to the random-walk matrices of $\mA$ and $\tA$. 

Normalized SV approximation is implied by standard $\epsilon/2$-approximation with respect to the original degree matrix holding for all unitary multiples of the adjacency matrix.
\begin{restatable}[Unitary transformation characterization of SV-approximation]{theorem}{svRegular}\label{lem:svregular}
    For $\tN,\mN \in \C^{n\times n}$, we have that $\tN \svn{\eps} \mN$ if
    for every pair of unitary matrices $\mU,\mV$, $\mU\tN\mV$ is a standard $\eps/2$-approximation of $\mU\mN\mV$ with respect to degree matrix $\mI$. Moreover, if $\tN \svn{\eps} \mN$ then
    for every pair of matrices $\mU,\mV$ satisfying $\|\mU\|\leq 1,\|\mV\|\leq 1$, we have that $\mU\tN\mV\svn{\eps}\mU\mN\mV$, and hence $\mU\tN\mV$ is a standard $\eps$-approximation of $\mU\mN\mV$ with respect to degree matrix $\mI$.
\end{restatable}

\subsection{Comparison to Prior Notions of Approximation}

First, we show that SV approximation implies unit-circle approximation:
\begin{restatable}{lemma}{SVimpliesUC}\label{lem:SVimpliesUC}
    If $\mA,\tA \in \C^{n\times n}$ with $\tA \svEulerian{\mD}{\eps} \mA$ then $\tA \uc_{\eps} \mA$ with respect to degree matrix $\mD$.
\end{restatable}

We show that SV approximation can be separated arbitrarily from UC approximation, even for undirected graphs and symmetric $2\times 2$ matrices.
\begin{restatable}{proposition}{SVUCsepGraphs}\label{prop:SVUCsepGraphs}~
    \begin{enumerate}
        \item\label{itm:sepGraph}  There is $c>0$ such that for all $n\in \N$, there are random walk matrices of undirected graphs $\tM,\mM\in \R^{n\times n}$ such that $\tM\uc_{1/c\sqrt{n}}\mM$, yet $\tM$ is not a $.3$-normalized SV approximation of $\mM$.
        \item\label{itm:sepMatr} For every $\alpha,\eps \in (0,1)$, there are symmetric matrices $\mW,\tW\in \R^{2\times 2}$ with $\|\mW\|,\|\tW\|\leq 1$ such that $\tW \uc_\eps\mW$ but $\tW$ is not an $\eps$-normalized SV approximation of $\mW$ for any $\eps'<\frac{\eps}{\sqrt{1-\alpha^2}}$.
    \end{enumerate}
\end{restatable}

In prior work, (\cite[Proposition 4.1]{AKMPSV20}), it was also shown that UC approximation can be arbitrary separated from standard approximation, even for undirected graphs.

\begin{proposition}[\cite{AKMPSV20}]
For every $\epsilon \in (0,1)$, there exist undirected regular graphs with random walk matrices $\tW,\mW$ such that $\tW$ is an $\eps$-approximation of $\mW$ with respect to $\mI-\mW$ but $\tW \uc_c \mW$ does not hold for every $c\in \N$.
\end{proposition}

Further, since UC approximation trivially implies standard approximation (as was also argued in \cite{AKMPSV20}) we see that SV approximation can be viewed of a strengthening of both UC and standard approximation.

\subsection{Properties of SV Approximation}\label{sec:sv:prop}
We now show that SV approximation enjoys several properties that are provably not enjoyed by prior notions of approximation. We summarize these differences in the following lemma:
\begin{restatable}{proposition}{UCpropFail}\label{prop:UC_prop_fail}
    SV approximation is preserved under multiplication by permutation matrices (\Cref{cor:svPerm}), embedding into arbitrary block matrices (\Cref{lem:svarblifts}), and products (\Cref{lem:svproducts}). None of these properties hold for UC approximation.
\end{restatable}

SV approximation between adjacency matrices is preserved under multiplication on each side of the adjacency matrix by (possibly different) permutations:
\begin{restatable}[SV preservation under multiplication by permutation matrices]{corollary}{svPerm} \label{cor:svPerm}
Let $\mA,\tA \in \C_{\geq 0}^{m\times n}$ and suppose $\tA \svgeneral{\Din}{\Dout}{\eps} \mA$. Let $\mU,\mV$ be arbitrary permutation matrices. Then $\mU \tA \mV \svgeneral{\Din'}{\Dout'}{\eps} \mU \mA \mV$ where $\Din'=\mU \Din \mU^*$ and $\Dout' = \mV^* \Dout \mV$.
Consequently, if $\tA\svgraph{\eps}\mA$, then $\mU \tA \mV \svgraph{\eps} \mU \mA \mV$ and if $\tA\svn{\eps}\mA$ then $\mU\tA\mV \svn{\eps}\mU\mA\mV$.
\end{restatable}

A corollary of~\Cref{cor:svPerm} is that SV approximation is preserved under embedding in a block matrix. In particular, it shows the stronger fact that approximation is preserved even if we use a different block structure for rows and columns (i.e., we are not embedding into principal submatrices).
In contrast, unit circle approximation is only known to be preserved when the embedding pattern is a block pattern given by the directed cycle; i.e., tensoring the adjacency matrix with that of the directed cycle. 

\begin{restatable}[SV preservation under arbitrary lifting]{lemma}{svArbLifts} \label{lem:svarblifts}
    Let $\mA, \tA \in \C^{m \times n}$ be matrices such that $\tA \svgeneral{\Din}{\Dout}{\eps} \mA.$
    Then for all integers $i, j, k, \ell \geq 0$
    \[
\begin{bmatrix}
    \mzero^{i \times j} & \mzero^{i \times n} & \mzero^{i \times k} \\
    \mzero^{m \times j} & \tA & \mzero^{m \times k} \\
    \mzero^{\ell \times j} & \mzero^{\ell \times n} & \mzero^{\ell \times k}
\end{bmatrix}
\svgeneral{\Din'}{\Dout'}{\eps}
\begin{bmatrix}
    \mzero^{i \times j} & \mzero^{i \times n} & \mzero^{i \times k} \\
    \mzero^{m \times j} & \mA & \mzero^{m \times k} \\
    \mzero^{\ell \times j} & \mzero^{\ell \times n} & \mzero^{\ell \times k}
\end{bmatrix}.
    \] 
Where 
\[
\Din' = \begin{bmatrix}
\mzero^{i\times i}& \mzero^{i\times n} & \mzero^{i\times \ell}\\
\mzero^{n\times i} & \Din & \mzero^{n\times \ell}\\
\mzero^{\ell\times i} &\mzero^{\ell\times n} & \mzero^{\ell\times \ell}
\end{bmatrix},\quad \Dout' =\begin{bmatrix}
\mzero^{j\times j}& \mzero^{j\times m} & \mzero^{j\times k}\\
\mzero^{m\times j} & \Din & \mzero^{m\times k}\\
\mzero^{k\times j} &\mzero^{k\times m} & \mzero^{k\times k}
\end{bmatrix}.
\]
Consequently, if $\tA\svgraph{\eps}\mA$ then
    \[
\begin{bmatrix}
    \mzero^{i \times j} & \mzero^{i \times n} & \mzero^{i \times k} \\
    \mzero^{m \times j} & \tA & \mzero^{m \times k} \\
    \mzero^{\ell \times j} & \mzero^{\ell \times n} & \mzero^{\ell \times k}
\end{bmatrix}
\svgraph{\eps}
\begin{bmatrix}
    \mzero^{i \times j} & \mzero^{i \times n} & \mzero^{i \times k} \\
    \mzero^{m \times j} & \mA & \mzero^{m \times k} \\
    \mzero^{\ell \times j} & \mzero^{\ell \times n} & \mzero^{\ell \times k}
\end{bmatrix}.
    \] 
\end{restatable}
By taking sums of different such liftings, one can obtain approximation for arbitrary tensorings: if $\tA\svgraph{\eps}\mA$, for any $\mM \in \zo^{i\times j}$, we have $\tA\otimes \mM\svgraph{\eps}\mA\otimes \mM$. A special case of this is concatenation:
\begin{restatable}[SV preservation under concatenation]{corollary}{svConcat} \label{lem:svconcat}
    Let $\mA_1, \mA_2 \tA_1, \tA_2 \in \C^{m \times n}$ be matrices such that
    $\tA_1 \svgraph{\eps} \mA_1$ and $\tA_2 \svgraph{\eps} \mA_2$ then 
    $
\begin{bmatrix}\tA_1 & \tA_2\end{bmatrix}  \svgraph{\eps} \begin{bmatrix} \mA_1 & \mA_2 \end{bmatrix}
    $. 
\end{restatable}

Another corollary is that SV approximation is preserved under products of different walk matrices with essentially no loss in the approximation quality. In contrast, the closest property known to be achieved by definitions of approximation considered in prior work is that unit circle approximation is preserved with no loss under only \emph{powers} of the same walk matrix. 
\begin{restatable}[SV preservation under products]{lemma}{svProducts} \label{lem:svproducts}
    Let $(\mN_i)_{i\in [\ell]}, (\tN_i)_{i\in [\ell]} \in \C^{n\times n}$ be such that for every $i$, $\tN_i \svn{\eps} \mN_i$.
    Then
    $\tN_\ell\cdots\tN_2\tN_1\svn{\eps + O(\epsilon^2)}\mN_\ell\cdots\mN_2\mN_1
    $.
\end{restatable}
A useful application of this result is to sparsifying powers of a walk matrix of an Eulerian digraph (\Cref{thm:rwSparsifier}).

Finally, we show that SV approximation \textit{does} satisfy useful properties known for other notions of approximation. SV approximation between digraphs requires exact preservation of the degrees: 
\begin{restatable}{lemma}{svStationary}\label{lem:svStationary}
    If $\tA \svgraph{\eps} \mA$,
    then $\tA \vones = \mA \vones$ and $\tA^{\rTrans}\vones=\mA^{\rTrans}\vones$.
\end{restatable}

SV sparsification is additive. We note that the analgous statement for normalized SV approximation is that it is preserved under convex combinations.
\begin{restatable}{lemma}{summability}\label{lem:summability}
    If $\tA_i \svgeneral{(\Din)_i}{(\Dout)_i}{\eps} \mA_i$ for all $i \in [k]$, letting $\Din \defeq \sum_{i\in [k]}(\Din)_i$ and $\Dout \defeq \sum_{i\in [k]} (\Dout)_i$, then $\sum_{i \in [k]} \tA_i \svgeneral{\Din}{\Dout}{\eps} \sum_{i\in [k]} \mA_i$. Consequently, if $\tA_i \svgraph{\eps} \mA_i$ for every $i$, then $\sum_{i \in [k]} \tA_i \svgraph{\eps} \sum_{i\in [k]} \mA_i$. 
\end{restatable}
SV approximation satisfies an approximate triangle inequality:
\begin{restatable}{lemma}{triangle}\label{lem:triangle}
    If $\mA_3 \svgeneral{\Din}{\Dout}{\delta} \mA_2$ and $\mA_2 \svgeneral{\Din}{\Dout}{\eps} \mA_1$ then $\mA_3 \svgeneral{\Din}{\Dout}{\eps+\delta+\eps\delta} \mA_1.$ Consequently, if $\mA_3 \svgraph{\delta} \mA_2$ and $\mA_2 \svgraph{\eps} \mA_1$ then $\mA_3 \svgraph{\eps+\delta+\eps\delta} \mA_1$, and if $\mA_3 \svn{\delta} \mA_2$ and $\mA_2 \svn{\eps} \mA_1$ then $\mA_3 \svn{\eps+\delta+\eps\delta} \mA_1$. Moreover, if for $\delta \in (0,1/2)$ and $\mA_0,\ldots,\mA_\ell$ we have $\mA_i \svgraph{\delta/2\ell}\mA_{i-1}$ for every $i$, then $\mA_\ell \svgraph{\delta} \mA_0$. Moreover, the equivalent claim holds for normalized SV approximation. 
\end{restatable}
A regular undirected graph SV approximates the complete graph with error equal to its expansion.
\begin{restatable}{lemma}{expanderSVappx}\label{lem:expander-uapprox}
    Let $G$ be a strongly connected, $d$-regular directed multigraph on $n$ vertices with adjacency matrix $\mA$ and let $\mJ\in\mathbb{R}^{n\times n}$ be a matrix with $1/n$ in every entry (i.e., $\mJ$ is the walk matrix of the complete graph with a self loop on every vertex). Then $\lambda(G)\leq 1-\lambda/2$ if and only if $\mA/d\svn{\lambda}\mJ$.
\end{restatable}

\section{Sparsification}\label{sec:sparsification}
Given the Laplacian of an Eulerian directed graph, we wish to compute a sparsifier with respect to SV approximation. In this section, we show how to solve a more general problem of sparsifiying a non-negative rectangular matrix $\mA \in \R_{\geq 0}^{m \times n}$ with respect to SV approximation. Further, in \Cref{sec:derandomized-square} we show how to construct square sparsifiers via the derandomized square, and in \Cref{sec:rwsparsify} we show how to sparsify powers of walk matrices. We first give an informal overview of our approach. 

We give a series of reductions that reduce this problem to degree- and bipartition-preserving sparsification of expander graphs with respect to $\eps$-approximation. By \Cref{SVequiv} (\Cref{SVequiv:bip}), to SV-sparsify (possibly directed) $\mA$, it suffices to SV-sparsity the undirected bipartite lift
\[\begin{bmatrix}\mzero & \mA^\rTrans\\\mA & \mzero\end{bmatrix}
\]
with respect to SV approximation.
In particular, given a symmetric matrix satisfying $\slift{\tA}\svgraph{\eps} \slift{\mA}$,
we have by \Cref{SVequiv} that $\tA \svgraph{\eps} \mA$. 
Thus, any symmetric bipartition-preserving SV sparsifier of undirected graphs immediately gives an SV sparsifier for Eulerian digraphs.
In fact, every SV sparsifier of undirected graphs must be bipartition-preserving, and an asymmetric sparsifier can be used to read off sparsifiers of the underlying directed graphs. For a formal statement of this stronger claim, see \Cref{lem:sv_undirected_suffices}.

We show (\Cref{lem:diag_to_sv}) that for undirected bipartite \textit{expanders}, degree- and bipartition-preserving approximation with respect to $\mD$ implies SV approximation. 

From this, we can construct SV sparsifiers (\Cref{thm:SVsparse}) by computing an expander decomposition of the bipartite lift, then sparsifying the relevant subgraphs in a suitable way. While we can use existing techniques to obtain this (with worse log factors), we also develop a new sparsification approach based on cycle decompositions. We show that in a bipartite expander, decomposing edges indicent to high-degree vertices into edge-disjoint cycles, then taking the odd or even edges of each cycle with probability $1/2$ produces a sparsifier with high probability. This sparsification procedure automatically preserves degrees and bipartitions, which we require to lift approximation with respect to $\mD$ to SV approximation. Prior work~\cite{CGPSSW18} showed an equivalent result in general Eulerian digraphs, except they required the cycles to be short, resulting in almost-linear runtime~\cite{CGPSSW18,PY19}.

\newcommand{\LH}{\L_{H}}
\newcommand{\LG}{\L_G}
\newcommand{\LC}{\L_C}
For convenience, we define notation for the adjacency and degree matrices of graphs:
\begin{definition}
    For the remainder of the section, for a weighted undirected graph $G$ we let $\mA_G$ denote the adjacency matrix of $G$ and $\mD_G\defeq \mA_G\vones$ denote the degree matrix. As such, an $\eps$-SV sparsifier of $G$ is a graph $H$ such that $\mA_H \svgraph{\eps}\mA_G$.
\end{definition}

First, we show that for undirected \textit{expanders}, degree- and bipartition-preserving sparsification with respect to the degree matrix implies sparsification with respect to SV approximation.
\begin{definition}[Bipartiteness] 
    We say $\ma \in \R^{n \times n}_{\geq 0}$ is {\em bipartite} with {\em bipartition} $S,T \subseteq [n]$ if $S$ and $T$ partition $V$ and $\mA_{S,S} = \mzero_{S \times S}$ and $\mA_{T,T} = \mzero_{T \times T}$. Given a graph $G$ with adjacency matrix $\ma$, we say that a sparsifier $H$ is \emph{bipartition-preserving} if for every bipartition $S,T$ in $G$, we have $H_{S,S}=\mzero_{S\times S}$ and $H_{T,T}=\mzero_{T\times T}$ (where the statement is vacuously satisfied if $G$ is not bipartite).
\end{definition}

In the following lemma we show that for non-negative matrices that are bipartite with the same bipartition, SV approximation is implied by approximation with respect to the error matrix $\mD$ (\Cref{def:leftrightapprox}), up to a loss in the condition number.

\begin{lemma}[From diagonal approximation to SV approximation]\label{lem:diag_to_sv}
    Suppose that symmetric $\tA \in \R^{n \times n}_{\geq 0}$ is a degree and bipartition-preserving $\epsilon/2$-approximation of symmetric bipartite $\mA \in \R^{n \times n}_{\geq 0}$ with respect to error matrix $\mD\defeq \diag(\mA\vones)=\diag(\mA^{\rTrans}\vones)$. 
    Then $\tA\svgraph{\gamma}\mA$ for $\gamma\defeq \frac{\eps}{\lambda^2}$ where 
    $\lambda \defeq 1 - \lambda_2(\mD^{+/2}\mA\mD^{+/2})$.
\end{lemma}
\begin{proof}
Recall by \Cref{SVequiv} that it suffices to show that 
$\tA$ is an $\eps/2$-approximation of $\mA$ with respect to $\mE\defeq \mD-\mA\mD^+\mA$.
First, we have
\[\mE\succeq (\mD-\mA)\mD^+(\mD+\mA) \succeq \mzero.\]
Furthermore by \Cref{lem:Hermitianapprox} and the assumption that $\tA$ is an $\eps/2$-approximation of $\mA$ with respect to $\mD$ we have
\[-\eps/2\cdot\mD\preceq \mA-\tA \preceq \eps/2\cdot \mD.\]
By assumption that $\mA$ has expansion $\lambda$ and is bipartite (and hence its spectrum is symmetric), we have
\[(\mD-\mA)^+ \preceq \frac{1}{\lambda}\mD^+, \quad (\mD+\mA)^+ \preceq \frac{1}{\lambda}\mD^+.
\]
Furthermore, as $\mA$ is the adjacency matrix of an undirected graph the only singular values of $\mA$ of magnitude $1$ can be $-1$ and $1$. Thus since $\tA$ matches degrees exactly and preserves bipartitions we have that $\ker(\mA-\tA)\supseteq \ker(\mE)$.
Thus we have
\begin{align*}
    \left\|\mE^{+/2}(\mA-\tA)\mE^{+/2}\right\| &\leq
    \left\|(\mD+\mA)^{+/2}\mD^{1/2}(\mD-\mA)^{+/2}(\mA-\tA)(\mD-\mA)^{+/2}\mD^{1/2}(\mD+\mA)^{+/2}\right\|\\
    &\leq \frac{1}{\lambda}\left\|\mD^{+/2}\mD^{1/2}(\mD-\mA)^{+/2}(\mA-\tA)(\mD-\mA)^{+/2}\mD^{1/2}\mD^{+/2}\right\|\\
    &\leq \frac{1}{\lambda}\left\|(\mD-\mA)^{+/2}(\mA-\tA)(\mD-\mA)^{+/2}\right\|\\
    &\leq \frac{1}{\lambda^2}\left\|\mD^{+/2}(\mA-\tA)\mD^{+/2}\right\|
    \leq \frac{\eps}{\lambda^2}. \qedhere
\end{align*}
\end{proof}

\subsection{Cycle and Expander Decompositions}
We give a new analysis of degree-preserving expander sparsification.

Chu, Gao, Peng, Sachdeva, Sawlani, and Wang~\cite{CGPSSW18} constructed sparsifiers for Eulerian digraphs, which we now discuss. They decomposed the low-importance edges of a graph into a series of edge-disjoint cycles $C_1,\ldots,C_t$. Then for each cycle, they chose a random orientation (clockwise or counterclockwise) and kept every edge in this orientation, and upweighted these by a factor of two. One can see that this procedure gives a subgraph and maintains Eulerianness. They bounded the impact of a single such sample in terms of the length of the cycle, and by giving a sophisticated algorithm for decomposing a graph into many short cycles concluded an almost-linear time sparsifier for standard approximation.

We follow this approach but make two modifications. First, we require the cycles to be \textit{Forward-Backward (FB)}, where we always alternate between clockwise and counterclockwise edges (and all cycles are of even length). This ensures that we exactly preserve \textit{degrees}, not just Eulerianness, which is essential for SV approximation. We take advantage of the reduction from directed to undirected SV sparsification to perform this decomposition in the undirected bipartite lift, where undirected cycles correspond to alternating cycles in the original graph. Our second modification is to prove a bound on the importance of the cycles that is independent of their length, as long as the cycles lie inside an expander. This enables us to use a simple linear time algorithm for finding the decomposition.

We first recall the standard procedure that enables us to efficiently decompose an undirected graph into a union of (potentially long) FB cycles.
\newcommand{\Eex}{E_{ex}}
\newcommand{\cC}{\mathcal{C}}
\begin{lemma}\label{lem:naiveCycleDecomp}
    Given an unweighted undirected graph $G=(V,E)$, the algorithm \algCycleDecomp returns a collection of edge-disjoint cycles $C_1,\ldots,C_T\subseteq E$ in time $O(|E|+|V|)$ such that at most $n$ edges are not contained in some cycle.
\end{lemma}

\begin{algorithm}\caption{\algCycleDecomp$(G=(V,E))$,  $G$ undirected}
    Initialize $\Eex,\cC = \{\}$.\\
    \For{$v\in V$}{
        \While{Greedily construct a non-backtracking path $v=u_0,u_1,\ldots,u_k$ in an arbitrary fashion and mark visited vertices.}{
        \If{$u_k$ has no neighbors other than $u_{k-1}$}{Add $\Eex \la \Eex \cup \{(u_{k-1},u_k)\}$ and remove $(u_{k-1},u_k)$ from $E$.\\
        Continue search from $u_{k-1}$.}
        \If{$u_k =u_{k'}$ for some $k'<k-1$}{Let $\cC \la \cC \cup (u_{k'},u_{k'+1},\ldots,u_{k-1})$ and let $E\la E\setminus \{(u_{k'},u_{k'+1}),(u_{k'+1},u_{k'+2}),\ldots,(u_{k-1},u_{k})\}$.\\
        Continue search from $u_{k}$.}
        }
    }
    \Return{$\cC,\Eex$}
\end{algorithm}

\begin{proof}
    The runtime and the fact that the cycles are edge disjoint is direct from the description. Then note that for every edge $(u,v)$ that we add to $\Eex$ (and remove from $E$) we must create an isolated vertex $v$. But this can happen at most $n$ times, so we are done.
\end{proof}

Furthermore, we recall the definition of an expander partition:
\begin{definition}\label{def:expDecomp}
    Given $\phi,\gamma$ and an undirected graph $G$, a \emph{$(\phi,\gamma)$-expander partition} of $G$ is a partition $V_1,\ldots,V_t$ such that at most $\gamma\cdot m$ edges cross between partitions and for every $i\in [t]$, $\lambda(G[V_i])\leq 1-\phi$. 
\end{definition}
We will use algorithms for expander decomposition as a key subroutine in our sparsifier. We require a small constant fraction of edges to lie between clusters, and thus state the algorithms fixing $\gamma=1/16$ such that $m/16$ edges cross subgraphs. In particular, given an algorithm \algExpDecomp~ computing a $(\Phi,m/16)$ expander decomposition, let $\TalgExpDecomp{m}$~ be its runtime (on graphs with $m$ edges) and $\PhialgExpDecomp=\Phi$. We recall the best known nearly-linear expander partitioning algorithm, as well as the (optimal) existential result:
\begin{theorem}[\cite{ADK22}]\label{alg:expDecompFast}
    There is an algorithm \algExpDecomp$(G)$ running in randomized time $\TalgExpDecomp{m}=O(m\log^7 m)$ that computes a $(\Phi=\log^{-4}m,1/16)$ expander decomposition of $G$ with high probability.
\end{theorem}

\begin{theorem}[\cite{ST04}]\label{alg:expDecompExis}
    There exists a $(\Phi=\log^{-2}m,1/16)$ expander decomposition of $G$.
\end{theorem}

\subsection{Sparsifying Bipartite Expanders}
The core of our algorithm is a way to sparsify bipartite, unweighted expanders in a degree- and bipartition- preserving way.

The algorithm decomposes the graph into a series of edge-disjoint cycles, then takes either twice the odd or twice the even edges with probability $1/2$. In the case that the undirected graph corresponds to the bipartite lift of an Eulerian digraph $G$, this precisely corresponds to decomposing $G$ into a union of Forward-Backward cycles and taking the forward or backward edges for each cycle. We bound the error in terms of the degree matrix $\mD$. This is tolerable as we apply this procedure inside expanders, so by \Cref{lem:diag_to_sv} we can transform this into an SV approximation guarantee with $\polylog(n)$ loss in approximation quality.

\newcommand{\dlow}{d_{low}}
\newcommand{\Elow}{E_{low}}
\begin{lemma}\label{lem:unweighted_sparsify_once}
    There is a constant $c>0$ such that given $\delta \in (0,1)$ and an unweighted, undirected bipartite graph $G$ with $m$ edges, \algUnweightedOnce returns in time $O(m)$ a graph $H$ with edge weights in $\{1,2\}$ such that with high probability $\mA_H$ is a $\delta$-approximation of $\mA_G$ with respect to error matrix $\mD_G$,
    and $H$ has at most $\max\{cn\delta^{-2}\log n,(7/8)m\}$ edges and exactly preserves degrees and bipartiteness.
\end{lemma}

\begin{algorithm}
\caption{\algUnweightedOnce$(\delta,G=(V,E))$}
    \lIf{$\delta < \sqrt{cn\log n/m}$}
    {
        \Return{$G$}
    }
    Let $\dlow\defeq\frac{m}{2n}$.\\
    Let $\Elow$ be the edges adjacent to vertices $v$ with $d(v)\leq \dlow$.\\
    Let $\Eex,(C_1,\ldots,C_t) =$\algCycleDecomp$(V,E\setminus \Elow)$.\\
    Initialize $H=(V,\Eex \cup E\setminus \Elow)$.\\
    \For{$i=1,\ldots,t$}{
        Let $(a_1,b_1),(b_1,a_2),\ldots,(b_k,a_1)=C_i$.\\
        Choose $b\gets {0,1}$ uniformly at random.\\
        \lIf(\tcp*[f]{Add with weight $2$}){$b=0$}{$H\la H\cup \{2\cdot (a_i,b_i)\}_{i\in [k]}$}
        \lElse{$H\la H\cup \{2\cdot (b_i,a_{i+1})\}_{i\in [k]}$
        }
    }
    \Return{$H$}
\end{algorithm}

To prove correctness, we recall a result on matrix concentration of Tropp:
\begin{lemma}[\cite{Tropp12}]\label{lem:matrixChernoff}
    Let $\{\mX_i\}$ be $d\times d$ random PSD matrices such that $\E[\sum_i \mX_i] \preceq \mX$. For $\delta\leq 1$,
    \[\Pr\left[\sum_i\mX_i - \E\left[\sum_i \mX_i\right] \succeq \delta \mX\right] \leq 1-d\cdot e^{\delta^2/3R},\quad \Pr\left[\sum_i\mX_i - \E\left[\sum_i \mX_i\right] \preceq -\delta \mX\right] \leq 1-d\cdot e^{\delta^2/3R},
    \]
  provided that 
    $R \geq \left\|\mX^{+/2}\mX_i \mX^{+/2}\right\|$
    almost surely.
\end{lemma}

Using this matrix concentration result, we can prove the theorem.

\begin{proof}[Proof of \Cref{lem:unweighted_sparsify_once}]
    First, note that if we do not sparsify $G$ at all we have $cn\delta^{-2}\log n < m$ and hence the sparsity requirement is already satisfied. Subsequently we assume this is not the case.
    
    Recall $\Elow$ is the set of edges adjacent to a vertex with degree at most $\dlow = (m/2n)$, and observe that there are at most $\dlow\cdot n\leq m/2$ such edges. Further recall $C_1,\ldots,C_t$ are the cycles found by \algCycleDecomp applied to $E\setminus \Elow$. By assumption that $G$ is bipartite, all such cycles are of even length.
    
    \newcommand{\tD}{\widetilde{\mD}}
    Let $\mX_i\defeq \tD_i-\tA_i$ be the random variable of the Laplacian of the $i$th sampled cycle and observe $\mzero \preceq \mX_i\preceq 2\mI_{C_i}.$
    Furthermore, let $\mX \defeq \mD_G \succeq \mzero$. We have
    \[\E\left[\sum_i \mX_i\right] = \sum_{i=1}^{\rTrans}\L_{C_i} \preceq \mX.
    \]
    Furthermore, for every $i\in[t]$ we have
    \[
        \left\|\mX^{+/2}\mX_i\mX^{+/2}\right\|
        \leq 2\left\|\mD^{+/2}\mI_{C_i}\mD^{+/2}\right\|
       \leq \frac{4n}{m}
       \]
    where the final step follows from all cycles being exclusively incident to vertices with degree at least $m/2n$.
    Thus, we apply \Cref{lem:matrixChernoff} with $\{\mX_i\}_{i\in [t]}$, $R = 4n/m$, $\mX = \mX$ and $\delta = \delta/c$. Letting $\mD_C-\mA_C$ be the random graph obtained by the cycle sampling procedure, with high probability we have
    \[-\delta\mD_G \preceq (\mD_C-\mA_C)-\sum_{i=1}^{t}(\mD_{C_i}-\mA_{C_i}) \preceq \delta\mD_G.
    \]
    By adding the non-random components to we obtain
    \[-\delta\mD_G \preceq \mA_H-\mA_G \preceq \delta\mD_G
    \]
    and thus $\mA_H$ is an $\delta$-approximation of $\mA_G$ with respect to $\mD_G$ with high probability.
    
    Furthermore, by \Cref{lem:naiveCycleDecomp} at least $m/2-n\geq m/4$ edges are included in the cycles $C_1,\ldots,C_t$, and hence $H$ has at most $7m/8$ edges under the assumption that $m\geq cn$, and $H$ preserves degrees and bipartiteness by construction.
\end{proof}

\subsection{Sparsifying a Constant Fraction of Edges}
We then construct our routine for sparsifying a graph with respect to SV approximation. At each iteration, we decrease the number of edges by a constant factor. We first bucket edges into powers of two, then apply \algUnweightedOnce.
\begin{lemma}\label{lem:sparsifyOnce}
    Given $\eps>0$ and $b\in \N$ and an undirected bipartite graph $G$ with $m$ edges, each with weight
    in $\{1,2,2^2,\ldots,2^b\}$, \algSparsifyOnce returns a graph $H$ with 
    \[O\left(b\cdot n\eps^{-2}\Phi^{-4}\log n\right)+\frac{15}{16}m
    \]
    edges where $\Phi=\PhialgExpDecomp$, each with weight
    in $\{1,2,2^2,\ldots,2^{b+1}\}$, such that with high probability $\mA_H\svgraph{\eps} \mA_G.$
    Moreover, \algSparsifyOnce runs in randomized time $O(m+bn+\TalgExpDecomp{m})).$
\end{lemma}
\begin{algorithm}
\caption{\algSparsifyOnce$(\eps,b,G=(V,E))$}
    Let $\delta = \eps^{-1}\Phi^{2}$.\\
    Initialize $H=(V,\emptyset)$.\\
    \For{$i =0,\ldots,b$}{
        Initialize (unweighted) $G_i$ as $G$ restricted to edges with weight $2^i$.\\
        Let $V_{i,1},\ldots,V_{i,t},\Eex \defeq$\algExpDecomp$(G_i)$.\\
        Let $H\la H\cup 2^i\cdot \Eex$.\\
        \For{$j \in 1,\ldots,t$}{
            Let $H_{i,j}\defeq$\algUnweightedOnce$(\delta, G_i[V_{i,j}])$\\
            Let $H \la H \cup 2^i \cdot H_{i,j}$.
        }
    }
    \Return{$H$}
\end{algorithm}

\begin{proof}
    Let $G_0,\ldots,G_b$ be the unweighted bipartite graphs where $G_i$ contains an edge $(u,v)$ if there is an equivalent edge with weight $2^i$. Thus, observe that 
    \[\sum_{i=0}^b2^i\cdot  (\mD_{G_i}-\mA_{G_i}) = \mD_{G}-\mA_{G}
    \]
    and it suffices to sparsify unweighted bipartite graphs. Furthermore let $m_i\defeq |E(G_i)|$ and observe $\sum_im_i=m$.
    
    Next, for every $i$ observe that \algExpDecomp~ computes a decomposition $V_{i,1},\ldots,V_{i,t}$ such that letting $G_{i,j}\defeq G_i[V_{i,j}]$ we have $\lambda(G_{i,j})\leq 1-\Phi$, and at most $m_i/16$ edges do not lie in one of these subgraphs. Let $n_{i,j}\defeq |G_{i,j}|$ and observe $\sum_{j}n_{i,j}\leq n$ and let $m_{i,j}\defeq|E(G_{i,j})|$ and observe $\sum_jm_{i,j}\leq m_i$. By \Cref{lem:unweighted_sparsify_once}, we have that for every $i,j$, $\mA_{H_{i,j}}$ is a degree and bipartition-preserving $\delta$-approximation of $\mA_{G_{i,j}}$ with respect to $\mD_{G_{i,j}}$ with high probability, and thus by \Cref{lem:diag_to_sv}
    \[\mA_{H_{i,j}} \svgraph{\delta\cdot \Phi^{-2}} \mA_{G_{i,j}}.
    \]
    Thus by \Cref{lem:summability} we have
    \[\mA_H \svgraph{\delta\cdot \Phi^{-2}} \mA_G
    \]
    and so by our choice of $\delta = \eps\Phi^{2}$ we obtain the desired approximation.
    
    We now analyze the number of edges. 
    The number of edges not placed into a subgraph $G_{i,j}$ is at most
    \[\sum_{i=1}^b\frac{m_i}{16}=\frac{m}{16}.
    \]
    By \Cref{lem:unweighted_sparsify_once} the number of edges in the sparsified graphs is at most
    \begin{align*}
        \sum_{i=1}^b\sum_{j=1}^{t}\max\{cn_{i,j}\delta^{-2}\log n,\frac{7}{8}m_{i,j}\} &\leq \sum_{i=1}^b\sum_{j=1}^{t}\left(c n_{i,j}\delta^{-2}\log n+\frac{7}{8}m_{i,j}\right)\\
        &= O\left(bn\eps^{-2}\Phi^{-4}\log n\right)+\frac{7}{8}\sum_{i=1}^b\sum_{j=1}^{t}m_{i,j}\\
        &= O\left(bn\eps^{-2}\Phi^{-4}\log n\right)+\frac{7}{8}m
    \end{align*}
    and hence the total number of edges is bounded as desired.
    Finally, we analyze the time. We can compute every $G_i$ with a single pass through the edge list, so the total work is bounded by $\sum_i\sum_j O(m_{i,j}+n_{i,j})+\sum_{i=1}^b\TalgExpDecomp{m_i} = O(m+bn+\TalgExpDecomp{m})$ as claimed.
\end{proof}

\subsection{Nearly Linear-Sized Sparsifers}
We now apply our routine for sparsifying a constant fraction of edges in the natural recursive way.
\begin{theorem}\label{thm:SVsparse}
    Given $\eps \in (0,1)$ and an undirected bipartite graph $G$ with $m$ edges with integer edge weights in $[1,U]$, \algSparsify returns a graph $H$ with 
    \[O\left(\log(nU)\cdot n\eps^{-2}\Phi^{-4}\log^{3}n\right)
    \]
    edges where $\Phi=\PhialgExpDecomp$, such that with high probability  $\mA_H\svgraph{\eps} \mA_G.$
    Moreover, \algSparsify runs in time $O(m+nb+\TalgExpDecomp{m\log(U)})$.
\end{theorem}

\begin{algorithm}
\caption{\algSparsify$(\eps,G=(V,E))$}
    Initialize $H_0=(V,\emptyset)$.\\
    Let $b\defeq\lceil\log(U)\rceil$.\\
    \For{$i=0,\ldots, b$}{
        \For{$(u,v) \in E$}{
            Let $\langle w_G(u,v) \rangle$ be the binary expansion of $w_G(u,v)$.\\
            \lIf{$\langle w_G(u,v) \rangle_i=1$}{Add $(u,v)$ to $H_0$ with weight $2^i$.}
        }
    }
    \For{$j = 1,\ldots, t\defeq O(\log n)$}{
        Let $H_j\la $\algSparsifyOnce$(\eps/2\log n,b+j,H_{j-1})$.
    }
    \Return{$H_t$}
\end{algorithm}

\begin{proof}
    First observe that in Phase 1 of the algorithm, we obtain $H_0$ such that $H_0$ has at most $mb$ edges (recalling $b=\lceil \log(U)\rceil$) and each edge has weight $\{1,2,2^2,\ldots,2^b\}$, and $\mA_{H_0}=\mA_G$. Then in each phase we obtain $H_j$ such that
    \[\mA_{H_j}\svgraph{\eps/2\log n}\mA_{H_{j-1}}\]
    with high probability by \Cref{lem:sparsifyOnce}.
    Then applying \Cref{lem:triangle} we obtain
    \[\mA_{H_t}\svgraph{\eps}\mA_{G}
    \]
    with high probability as desired.
    We next analyze the sparsity. Observe that the final edge weights are bounded by $2^{b+t}=\poly(nU)$ and hence the final sparsity is
    \[O\left((b+t)n\eps^{-2}\Phi^{-4}\log^{3} n\right) = O\left(\log(nU)n\eps^{-2}\Phi^{-4}\log^{3}n\right).
    \]
    Finally, we analyze the time complexity. Our initial call has $mb$ edges and thus runs in time $O(mb+bn+\TalgExpDecomp{mb})$  by \Cref{lem:sparsifyOnce}, and all subsequent calls are to sparsify graphs with a constant factor fewer edges, so the time is dominated by the initial term.
\end{proof}
By plugging in the algorithmic (\Cref{alg:expDecompFast}) and existential (\Cref{alg:expDecompFast}) results for expander partitions, we obtain our sparsifier routines:
\begin{corollary}\label{cor:sparseFast}
    Given $\eps \in (0,1)$ and an undirected bipartite graph $G$ with $m$ edges with integer edge weights in $[1,U]$, there is a randomized algorithm that returns a graph $H$ with 
    \[O\left(\log(nU)\cdot n\eps^{-2}\log^{19}n\right)
    \]
    edges such that with high probability $\mA_H\svgraph{\eps} \mA_G.$
    Moreover, the algorithm runs in time $O(\log(U)m\log^7 m)$.
\end{corollary}

\begin{corollary}\label{cor:sparseExis}
    Given $\eps \in (0,1)$ and an undirected bipartite graph $G$ with $m$ edges with integer edge weights in $[1,U]$, there is a graph $H$ with 
    \[O\left(\log(nU)\cdot n\eps^{-2}\log^{11}n\right)
    \]
    edges such that $\mA_H\svgraph{\eps} \mA_G.$
\end{corollary}

\subsection{Derandomized Square Sparsification}\label{sec:derandomized-square}

Here we show that the derandomized square of Rozenman and Vadhan~\cite{RozenmanVa05} gives an SV sparsification of the square graph. We first explain the derandomized square. Given a $d$-regular graph $G$, the squaring operation can be though of as placing a $d\times d$ complete graph between the in- and out-neighbors of every vertex $v$. The derandomized square replaces this complete graph with an expander. As an expander approximates the complete graph, this sparsification operation produces an approximate square. 

We first formally define the derandomized square. 
\begin{definition}
    Given a $d$-regular graph $G=([n],E)$ with neighbor function $\Gamma_G:[n]\times[d]\ra [n]$ and $c$-regular graph $H=([d],E')$ with neighbor function $\Gamma_H:[d]\times [c]\ra [d]$, the \emph{derandomized square of $G$ with respect to $H$}, denoted $G\ds H$, is the $dc$-regular graph on $n$ vertices with neighbor function
    \[
        \Gamma_{G\ds H}(v,(i,j)) = \Gamma_G(\Gamma_G(v,i),\Gamma_H(i,j)).
    \]
\end{definition}    

Prior analyses of the derandomized square have shown that it improves expansion approximately as well as the true square~\cite{RozenmanVa05} and that it produces unit-circle square sparsifiers~\cite{AKMPSV20}. We strengthen the latter conclusion to SV square sparsifiers.
\newcommand{\tAs}{\widetilde{\mA^2}}
\begin{lemma}\label{lem:drsq_sv}
    Let $G=(V,E)$ be a $d$-regular directed multigraph with random walk matrix $\mW$. Let $H$ be a $c$-regular expander on $d$ vertices with $\lambda(H)\leq 1-\epsilon$ and let $\tW$ be the adjacency matrix of $G \ds H$. Then
    \[
   \tW\svn{\eps}\mW^2.
    \]
\end{lemma}
\begin{proof}
    Let $\mW_H$ be the random walk matrix corresponding to $H$, and $\mJ\in \R^{d\times d}$ be the random walk matrix of a complete graph with self loops. By \Cref{lem:expander-uapprox,lem:svarblifts} we have
    \begin{equation}\label{eq:krock-approx}
        \mI_n \otimes \mW_H \svgraph{\eps} \mI_n \otimes \mJ.
    \end{equation}
    Let $\mER$ be an $nd\times nd$ edge-rotation permutation matrix such that $\mER((v,j),(u,i))=1$ if the $i$th edge leaving $u$ is the $j$th edge entering $v$. It is straight-forward to verify that
    $$
    \mW^2 = \mq^\top \mER (\mI \otimes \mJ) \mER \mq
    $$
    and
    $$
    \tW = \mq^\top \mER (\mI \otimes \mW_H) \mER \mq
    $$
    where $\mQ \defeq \mI \otimes \frac{1_{d}}{\sqrt{d}}$ is an $nd \times n$ matrix. Now note that $\|\mER\| \leq 1$ and $\|\mQ\| \leq 1$ and thus $\|\mER \mQ\|, \|\mQ^\top \mER\| \leq 1$. 
    Thus by \Cref{lem:svregular}, and \Cref{eq:krock-approx} we have
    \[
    \widetilde{\mw} \uapprox_\epsilon \mw^2. \qedhere
    \]
\end{proof}

\newcommand{\algSquare}{\textsc{SparsifySquare}}
\newcommand{\algPower}{\textsc{SparsifyPower}}
\newcommand{\algPowerTwo}{\textsc{SparsifyPowerTwo}}
\subsection{Sparsification of Eulerian Walks}\label{sec:rwsparsify}
We give the first nearly-linear time algorithm for sparsifying random walk polynomials of Eulerian digraphs with second singular value bounded away from $1$. In particular, given an Eulerian digraph $\mW \in \R^{n\times n}_{\geq 0}$ with second singular value bounded away from $1$, we compute a sparse graph such that the walk matrix $\tW$ approximates $\mW^\ell$
with respect to SV approximation. For this result, we crucially require that SV approximation is preserved under products, a feature that is not obtained by prior notions of approximation.
\begin{definition}
    Given a strongly connected Eulerian digraph $G$, we let $\sigma(G)=\sigma_2(\mD^{-1/2}_G\mA_G\mD^{-1/2})$ and refer to this as the \emph{second normalized singular value} of $G$.
\end{definition}

\begin{restatable}{theorem}{rwSparsifier}\label{thm:rwSparsifier}
    Given $\ell \in \N$ and $\eps>0$ and an Eulerian digraph $G=(\mA,\mD)$ where $\sigma(G)\leq 1-1/\tau$ with $m$ edges with integer edge weights in $[1,U]$, \algPower~returns in time $\tO(m\log(U)+n\eps^{-2}\polylog(\tau U\ell))$ an Eulerian digraph $H=(\mA_H,\mD)$ with $\tO(n\eps^{-2}\polylog(\tau U\ell))$ edges with integer edge weights in $[1,\poly(n,\tau,U,\ell,1/\eps)]$ such that 
    \[\mD^{+/2}\mA_H\mD^{+/2}\svn{\eps} (\mD^{+/2}\mA\mD^{+/2})^\ell.\]
\end{restatable}

To prove these, we first recall the square sparsification algorithm implicit in the work of Cohen et al.~\cite{CKPPRSV17}.
\begin{theorem}[Implicit in~\cite{CKPPRSV17}]\label{thm:sqSparse}
    Given an Eulerian
    weighted digraph $G$ with Laplacian $\mD-\mA_G$, let $G^2$ be the graph with Laplacian $\mD-\mA_G \mD^{+} \mA_G$. Let $\mU$ denote the Laplacian of the undirected graph corresponding to $(\mA_{G^2}+\mA_{G^2}^\rTrans)/2$. Then \algSquare~of \cite{CKPPRSV17} produces in time $\tO(m)$, a weighted graph $H$ with the following properties:
    \begin{enumerate}
        \item $H$ has the same weighted in- and out- degree sequences as those of $G^2$.
        \item $H$ has $\tO(m/\eps^{2})$ nonzero edges.
        \item $\mA_H$ is an $\epsilon/2$-approximation of $\mA_G \mD^{-1} \mA_G$ with respect to $\mS_{\mD-\mA_G \mD^{+} \mA_G}$.
    \end{enumerate}
\end{theorem}
We can use a block embedding and \Cref{thm:sqSparse} to obtain a nearly-linear time algorithm for sparsifying products with respect to SV approximation.
\begin{lemma}\label{lem:sparsify_product}
    Given $\eps>0$ and Eulerian digraphs $G_1,G_2$ each with at most $m$ edges with matching degree matrix $\mD$, where each edge has integer weight in $[1,U]$, \algProduct~ returns an Eulerian digraph $H$ with $\tO(m/\eps^2)$ edges such that with high probability
    \[
        \mA_H \svgraph{\eps} \mA_{G_2}\mD^{+}\mA_{G_1}
    \]
    Moreover, \algProduct~ runs in time $\tO(m/\eps^2)$.
\end{lemma}
\begin{algorithm}
\caption{\algProduct$(\eps,G_1,G_2)$}
    Let $G$ be the graph with $\mA_G = \begin{bmatrix}\mzero & \mzero & \mzero\\
    \mA_{G_1} & \mzero & \mzero\\
    \mzero & \mA_{G_2}& \mzero
    \end{bmatrix}$\\
    Let $H\la [\algSquare(\eps/3,G)]_{3,1}$.\\
    \Return{$H$}
\end{algorithm}
\begin{proof}
    We first analyze correctness. Note by \Cref{thm:sqSparse} we have that 
    \[
        \begin{bmatrix}\mzero & \mzero & \mzero\\
    \mzero & \mzero & \mzero\\
    \mA_{H_0} & \mzero& \mzero
    \end{bmatrix} \text{ is an $\eps/6$-approximation of } \begin{bmatrix}\mzero & \mzero & \mzero\\
    \mzero & \mzero & \mzero\\
    \mA_{G_2}\mD^{+}\mA_{G_1} & \mzero& \mzero
    \end{bmatrix} \text{ with respect to $\mR$}
    \]
    where 
    \[\mR=\begin{bmatrix}\mD & \mzero & -(\mA_{G_2}\mD^{+}\mA_{G_1})^{\rTrans}\\
    \mzero & \mzero & \mzero\\
    -\mA_{G_2}\mD^{+}\mA_{G_1} & \mzero& \mD
    \end{bmatrix}\]
    This is precisely \Cref{SVequiv:somez} of \Cref{SVequiv}, so we conclude $\mA_{H}\svgraph{\eps}\mA_{G_2}\mD^+\mA_{G_1}$.
\end{proof}

\newcommand{\algDropEdges}{\textsc{FixEdgeWeights}}
Finally, we require an algorithm that removes very small edge weights. This is not trivial, as we must do so in a way that exactly preserves degrees for SV approximation to hold.
\begin{lemma}\label{lem:drop_small_edges}
    There is an algorithm \algDropEdges~ such that the following holds. Let $G$ be an Eulerian digraph with $m$ edges where the degree matrix $\mD_G$ can be represented as a subset sum of $\{1,\ldots,2^b\}$.
    Then for $t\in \N$, \algDropEdges$(t,G)$ outputs in time $O(mb)$ a graph $H$ such that:
    \begin{itemize}
        \item $H_t$ has $m+2n$ edges and is Eulerian and $\mD_H=\mD_G$.
        \item Edge weights in $H$ can be represented as a subset sum of $\{2^{-t},\ldots,1,\ldots,2^b\}$. 
        \item $\|\mA_H-\mA_G\|\leq 2n\cdot 2^{-t}$.
    \end{itemize}
\end{lemma}
\begin{proof}
    First, let $\lfloor G\rfloor$ be the graph where all edge weights are rounded down to a multiple of $2^{-t}$. Observe that this satisfies
    \[
        \|\mA_{\lfloor G\rfloor}-\mA_G\|\leq n\cdot 2^{-t}.
    \]
    Furthermore, note that each vertex $v$ has an in-degree shortfall equal to $s_{in}(v)=\mD_{v,v}-\sum_u (\mA_{\lfloor G\rfloor})_{v,u}$. Note that $s_{in}(v)\geq 0$ since we only remove weight from edges. This can be expressed as a subset sum of $\{2^{-t},\ldots,1,\ldots,2^b\}$ as all numbers in the summand can be. An equivalent statement holds for the out-degree shortfall. Furthermore, the total in- and out- degree shortfalls must be equal, since every truncation contributes to both sums. Thus, by greedily matching in- shortfalls to out- shortfalls, we obtain the desired graph $H$. It is easy to see this satisfies the approximation claim.
\end{proof}

We are now prepared to state our algorithm for sparsifying an arbitrary power of an Eulerian digraph.
\begin{algorithm}
\caption{\algPower$(\eps,G,\ell)$}
    Let $t\defeq \lceil \log \ell\rceil$.\\
    Let $\langle \ell \rangle$ be the binary expansion of $\ell$.\\
    Let $l \defeq O(\log(nU\ell/\gamma \eps))$.\\
    Let $G_1\defeq \algSparsify(\eps/2,G)$.\\
    \For{$i=2,\ldots,t$}{
        Let $P_i \defeq \algProduct(\eps/t,G_{i-1},G_{i-1})$.\\
        Let $H_i \defeq \algDropEdges(l,P_i)$.\\
        Let $G_i \defeq \algSparsify(\eps/t,H_i)$.\\
        \If{$\langle \ell\rangle_{t-i}=1$}{
            Let $P_i \la \algProduct(\eps/t,G_{i-1},G_{1})$.\\
            Let $H_i \la \algDropEdges(l,P_i)$.\\
            Let $G_i \la \algSparsify(\eps/t,H_i)$.\\
        } 
    }
    \Return{$G_t$}
\end{algorithm}

We are now prepared to prove the result:
\rwSparsifier*
\begin{proof}
    We first show the approximation guarantee.  Note that by the guarantees of \Cref{lem:sparsify_product} and \Cref{cor:sparseFast} and \Cref{lem:drop_small_edges}, for every $i$,
    \[\mD\defeq \mD_{P_i}=\mD_{H_i}=\mD_{G_i}=\mD_G.\]
    Let $\mN_{P_i}\defeq \mD^{+/2}\mA_{P_i}^{+/2}\mD^{+/2}$ and $\mN_{H_i}\defeq \mD^{+/2}\mA_{H_i}^{+/2}\mD^{+/2}$ and $\mN_{G_i}\defeq \mD^{+/2}\mA_{G_i}^{+/2}\mD^{+/2}$. The proof of correctness does not change if we assume $\ell=2^t$ for some $t$ (and hence the if condition in the main loop never occurs), so we analyze it assuming this for simplicity. The fact that this does not matter follows from the fact that SV approximation is preserved under arbitrary products, not merely powers (\Cref{lem:svproducts}).
    
    We first argue that the singular value of all graphs we sparsify remain bounded below $1$.

    \begin{claim}\label{clm:goodCond}
        For every $i$, $\sigma(H_i),\sigma(P_i),\sigma(G_i) \leq 1-1/O(\tau)$.
    \end{claim}
    \begin{proof}
        We show the claim via induction. By assumption, we have $\sigma(G)=1-1/\tau$. By \Cref{lem:SVSingular} we have $\sigma(G_1)\leq 1-(1+\eps/t)/\tau$. Subsequently, assume that $\sigma(G_{i-1})\leq 1-(1+2\eps/t)^{2(i-1)}/\tau$. 
        \begin{itemize}
            \item By \Cref{lem:singularPower} we have $\sigma(G_{i-1}^2) \leq \sigma(G_{i-1}) \leq 1-(1+2\eps/t)^{2(i-1)}/\tau$.
            \item By \Cref{lem:SVSingular} we have $\sigma(P_i) \leq 1-(1+\eps/t)(1+2\eps/t)^{2(i-1)}/\tau$.
            \item By \Cref{lem:singularPerturb} we have $\sigma(H_i) \leq 1-(1+\eps/t)(1+2\eps/t)^{2(i-1)}/\tau+1/\poly(\tau)$.
            \item By \Cref{lem:SVSingular} we have $\sigma(G_i) \leq 1-(1+2\eps/t)^{2i}/\tau$
        \end{itemize}
        So the claim follows assuming $\eps\leq 1/c'$ for constant $c'$.
    \end{proof}

    We now show that SV approximation is preserved at every step of the loop. Fixing $i$ and $G_{i-1}$:
    \begin{enumerate}
        \item  By \Cref{lem:sparsify_product} we have that $\mN_{P_i} \svn{\eps/3t} \mN_{G_{i-1}}^2.$
        
        \item By \Cref{lem:drop_small_edges} we have 
        $\|\mA_{H_i} - \mA_{P_i}\|\leq 2n2^{-l}$, so $\mA_{H_i}$ is a $2n2^{-l}\leq \eps/\poly(\tau)$ approximation of $\mA_{P_i}$ with respect to $\mD$. Taking the bipartite lift, $\begin{bmatrix}
        \mzero & \mA_{H_i}^{\rTrans}\\ \mA_{H_i} & \mzero
        \end{bmatrix}$ is a degree- and bipartition- preserving $\eps/\poly(\tau)$-approximation of $\begin{bmatrix}
        \mzero & \mA_{P_i}^{\rTrans}\\ \mA_{P_i} & \mzero
        \end{bmatrix}$
        with respect to $\begin{bmatrix} \mD & \mzero \\ \mzero & \mD\end{bmatrix}$ (where we use that $P_i$ is not itself bipartite to ensure that the lift of $H_i$ preserves all bipartitions). Thus by \Cref{lem:diag_to_sv} (using that the bipartite lift has normalized second eigenvalue at least $1/O(\tau)$ by \Cref{clm:goodCond}) and \Cref{SVequiv}, we have that $\mN_{H_i} \svgraph{\eps/3t} \mN_{P_i}.$

        \item By \Cref{cor:sparseFast}, $\mN_{G_i}\svn{\eps/3t}\mN_{H_i}.$
    \end{enumerate}
    
    Applying \Cref{lem:triangle} we obtain $\mN_{G_i} \svn{\eps/t} \mN_{G_{i-1}}^2$ for every $i$.
    Therefore by \Cref{lem:svproducts},
    \[
        \mN_{G_t} \svn{\eps/t} \left(\mN_{G_{i-1}}\right)^2 \svn{\eps/t} \ldots \left( \mN_{G_2}\right)^{2^{t-1}} \svn{\eps/t} \left(\mN_{G_1}\right)^{2^t}
    \]
    and so we conclude by \Cref{lem:triangle} (and taking $\eps\la \eps/10$).

    To analyze the time, note that by \Cref{lem:drop_small_edges} we always call \algSparsify on graphs with (up to rescaling) integer edge weights in $[1,\poly(\tau Un\ell/\eps)]$, and by \Cref{thm:sqSparse} $P_i$ (and thus $H_i$) has $\tO(n\eps^{-2}\polylog(U\tau \ell/\eps))$ edges, so we obtain the desired runtime and sparsity by \Cref{cor:sparseFast}.
\end{proof}

\subsection{Sparsification of Directed Random Walks}

We first recall the definition of the statioary distribution for an arbitrary directed graph:
\begin{definition}[\cite{CKPPSV16}] 
    Given a strongly-connected directed graph $G$ with Laplacian $\mD-\mA$, there exists a vector $\vec{s}$, which we refer to as the \emph{stationary distribution} or \emph{stationary vector}, with $\|s\|_1=1$ and $\mS:=\diag(\vec{s})$, which refer to as the \emph{stationary matrix}, such that the following equivalent conditions hold:
    \begin{enumerate}
        \item $\mA\mD^{-1}\vec{s}=\vec{s}$
        \item $(\mD-\mA)\mD^{-1}\mS$ is an Eulerian Laplacian.
    \end{enumerate}
\end{definition}
We then state the main theorem:
\newcommand{\tD}{\mathbf{\widetilde{D}}}
\begin{restatable}{theorem}{dirRW}\label{thm:dirRW}
    There is a randomized algorithm \algPowerDi~satisfying the following. The algorithm takes as input $\ell \in \N$ and $\eps,\delta>0$ and a strongly connected directed graph $G$ with:
    \begin{itemize}
        \item $m$ edges with integer edge weights in $[1,U]$,
        \item  minimum statationary probability at least $s$,
    \end{itemize}
    and let $R=U\ell/s\delta$. Then \algPowerDi~returns in time 
    $\tO\left(\left(m+n\eps^{-2}\right)\polylog(R)\right)$
    an Eulerian digraph $H=(\mD_H,\mA_H)$ with $\tO(n\eps^{-2}\polylog(R))$ edges with integer edge weights in $[1,\poly(R)]$ such that
    \[\mA_H\svgraph{\eps} \mB\]
    where $\mB$ is $\delta/n$-entrywise close to $(\mA\mD^{-1})^\ell\mS=\mW^\ell \mS$ where $\mA,\mD,\mS$ are the adjacency, degree, and stationary matrix of $G$ respectively. In particular, $|\Cut_{\mB}(S,T)-\Cut_{G^\ell}(S,T)|\leq \delta$ for every $S,T\subseteq [n]$.
\end{restatable}

To obtain this result, we recall that we can find the stationary distribution of an arbitrary directed graph in nearly linear time.
\begin{lemma}[Lemma 53~\cite{CKPPSV16}]\label{lem:findStat}
    Given $\delta>0$, there is a randomized algorithm \algFindClose~ that, given a directed graph $G$ with weak mixing time\footnote{The mixing time of the $1/2$-lazy random walk.} $T$ and minimum stationary probability $s$, returns $(\tD,\tA,\tS)$ in time $\tO(m\polylog(nT/s\delta))$, where $\tS$ is the stationary matrix of $\tA$. Moreover, $\|\tA-\mA_G\|\leq \delta$.
\end{lemma}

We use this primitive to find an Eulerian rescaling of the lazified random walk of $G$.

\begin{proof}[Proof of~\Cref{thm:dirRW}]
    We first claim that we can compute an approximate rescaling. Let $T$ be the weak mixing time of $G$, and note we have $T\leq \poly(U/s)$ (where we use that $1/s\geq n$). Note that $(\mD-\mA)\mD^{-1}\mS$ is an Eulerian Laplacian with edge weights in $[s/U, 1]$ and weak mixing time at most $T$. 
    Now let 
    \[(\tD,\tA,\tS):=\algFindClose(G,\rho)\]
    for $\rho$ to be chosen later.
    We have that 
    \[\tL:=(\tD-\tA)\tD^{-1}\tS = \tS-\tA\tD^{-1}\tS
    \]
    is an Eulerian Laplacian, and moreover (taking $\rho\la \rho/n$):
    \[\|\tD-\mD\|\leq \rho,\quad  \|\tA-\mA\|\leq \rho, 
    \text{ and }
    \|\tS-\mS\|\leq \rho.\]
    In particular, choosing $\rho$ sufficiently small compared to $U/s$ we have that $\tL$ has weak mixing time at most  $\poly(T)$.
    Now let 
    \[\tM:=(1-\gamma)\tN+\gamma \mI\] 
    where $\tN = \tS^{-1/2}(\tA\tD^{-1}\tS) \tS^{-1/2}$, for $\gamma$ to be chosen later. By \Cref{lem:eulerian_lazy} we have 
    \[
    \sigma(\tM)\leq 1-1/\poly(TU/s\gamma).
    \]
    Moreover, $\tM$ is the normalized adjacency matrix of an Eulerian graph with adjacency and degree matrix $\mB$ and $\mD_H$, i.e., $\mB=\mD_H^{1/2}\tM\mD_H^{1/2}$. By \Cref{thm:rwSparsifier} we can compute in the specified time bound a graph $H$ with adjacency $\mA_H$ such that 
    \[
    \mD_H^{-1/2}\mA_H\mD_H^{-1/2} \svn{\eps}  \mD_H^{-1/2}\mB \mD_H^{-1/2}=\tM^\ell.
    \]
    \begin{claim}\label{clm:diagClose}
        $\mD_H$ is $\delta\cdot s$ close to $\mS$. 
    \end{claim}
    \begin{proof}
        We have $\|\mD_H-\mS\| \leq  \|\mD_H-\tD\|+\|\tD-\mS\| \leq \gamma +\rho$
        and hence taking $\gamma$ and $\rho$ small relative to $\delta\cdot s$ this holds.
    \end{proof}
    \begin{claim}
        $\mB$ is $\delta$ close to $(\mA\mD^{-1})\mS$. 
    \end{claim}
    \begin{proof}

        We have 
        \begin{align*}
        \|\mB - (\mA\mD^{-1})\mS\| &\leq \|\mS^{-1/2}\mB\mS^{-1/2} - \mS^{-1/2}(\mA\mD^{-1})\mS^{1/2}\| && \|\mS\|\leq 1\\
        &\leq 
        \delta + \|\tM^\ell - \mS^{-1/2}(\mA\mD^{-1})^\ell\mS^{1/2}\| && \text{\Cref{clm:diagClose}}\\
        &= \delta+\|\tM^\ell -(\mS^{-1/2}(\mA\mD^{-1}\mS)\mS^{-1/2})^\ell\|\\
        &\leq \delta + \|\tM^\ell -\tN^\ell\| + \|\tN^\ell-(\mS^{-1/2}(\mA\mD^{-1}\mS)\mS^{-1/2})^\ell\|\\
        &=  \delta + \|((1-\gamma)\tN+\gamma \mI)^\ell - \tN^\ell\|+\delta/2 \leq O(\delta). 
        \end{align*}
        where the second to last inequality follows from making $\rho$ small relative to $\delta/s$, and we conclude by taking $\delta \la \delta/cn$.
    \end{proof}
    Thus, we have the desired approximation statement. Finally, for arbitrary $S,T$, we have
    \begin{align*}
        |\Cut_{\mB}(S,T)-\Cut_{G^\ell}(S,T)| &= |\vones_T\mB \vones_S - \vones_T\mW^\ell\mS \vones_S|\leq \delta. \qedhere
    \end{align*}
\end{proof}

We can then state the implication in terms of cut sparsification.
\begin{theorem}\label{thm:dirCut}
    There is a randomized algorithm \algCut~ satisfying the following. Fix $\ell \in \N$, $\eps>0$, and a strongly connected directed graph $G$ with $m$ edges and integer edge weights in $[1,U]$, and $s\leq \pimin(G)$, 
    and let $R=U\ell/s$.
    Then \algCut~ returns a graph $H$ with $\tO(n\eps^{-2}\polylog(R))$ edges in time 
    \[\tO\left(\left(m+n\eps^{-2}\right)\polylog(R)\right)\] 
    such that the following holds. For every $S,T\subset [n]$:
    \[\left|\Cut_{H}(S,T) - \Cut_{G^\ell}(S,T)\right|
    \leq \eps \cdot \sqrt{\min\left\{\Cut_{G^\ell}(S),\Uncut_{G^\ell}(S)\right\} \cdot
    \min\left\{\Cut_{G^\ell}(T),\Uncut_{G^\ell}(T)\right\}}.\]
\end{theorem}
\begin{proof}
    Apply \Cref{thm:dirRW} with $\eps=\eps$ and $\delta= \eps \cdot (s/2U)^3/2n$, so that we obtain $(\mA_H,\mD_H)$ in the claimed time bound. We now prove the cut approximation property. Fix arbitrary $S,T\subset [n]$, and recall that $\Cut_H(S,T) = \vones_{T}^\rTrans\mA_H\vones_S/|S|$.
    Recall that by \Cref{thm:dirRW} we have
    \[
    \mD_H^{-1/2}\mA_H\mD_H^{-1/2} \svn{\eps} \mD_H^{-1/2}\mB\mD_H^{-1/2}.
    \]
    Set
    \[x^\rTrans:=\vones_{S^c}^\rTrans\mD_{H}^{1/2},\quad\quad  y:=\mD_{H}^{1/2}\vones_S\]
    and apply the definition of SV approximation:
    \begin{align*}
        \MoveEqLeft{\left|\Cut_H(S,T)- \vones_{T}^\rTrans(\mB\mD_H^{-1})\mD_H\vones_{S}\right|}\\ 
        &\leq \eps \sqrt{\vones_{T}^\rTrans\mD_H^{1/2}(\mI-\mD_H^{-1/2}(\mB\mD_H^{-1})(\mB^\rTrans\mD_H^{-1})\mD_H^{1/2})\mD_H^{1/2}\vones_{T}}\\ 
        &\quad\quad \cdot  \sqrt{\vones_{S}^\rTrans\mD_H^{1/2}(\mI-\mD_H^{-1/2}(\mB^\rTrans\mD_H^{-1})(\mB\mD_H^{-1})\mD_H^{1/2}) \mD_H^{1/2}\vones_{S}
        }\\
        &= \eps \cdot \sqrt{\Cut_{B^TB}(S)}\cdot \sqrt{\Cut_{B^TB}(T)}\\
        &\leq 2\eps \cdot\sqrt{\min\{\Cut_{B}(S),\Uncut_{B}(S)\}}\cdot \sqrt{\min\{\Cut_{B}(T),\Uncut_{B}(T)\}} := \eps\cdot\mathtt{err}(B,S,T)
    \end{align*}
    where the final line follows as $\Cut_{B^TB}(S) \leq 2\cdot \Cut_B(S)$ and $\Cut_{B^TB}(S)\leq 2\Uncut_B(S)$ and likewise for $T$. Finally, delete all edges from $H$ with weight less than $\delta$. Applying~\Cref{thm:dirRW}, we obtain
    \[
    \left|\Cut_H(S,T)- \Cut_{G^\ell}(S,T)\right| \leq 2\delta\cdot n+ \eps\cdot\mathtt{err}(B,S,T).
    \]
    as all cut values in $B$ are within $\delta$ of cut values in $G^\ell$ by~\Cref{thm:dirRW}.

    By~\Cref{prop:statNonNeg}, $\mathtt{err}(G^\ell,S,T)$ is either $0$ or at least $(s/2U)^3$. In the latter case, we have $2\delta\cdot n+ \eps\cdot\mathtt{err}(B,S,T)\leq 2\eps \cdot\mathtt{err}(B,S,T)$ by choice of $\delta$. In the former case, as $|\mathtt{err}(B,S,T)-\mathtt{err}(G^\ell,S,T)|\leq \delta$ and $|\Cut_{B}(S,T)-\Cut_{G^\ell}(S,T)|\leq \delta$, it must have been the case that the $(S,T)$ cut in (pre-deletion) $H$ consisted entirely of edges of weight at most $\delta$. Thus, deleting these edges brings $\Cut_H(S,T)$ to zero, so we again obtain approximation with no additive loss.
\end{proof} 
\section{Squaring-based Solver for Normal Directed Laplacian Systems}\label{sec:normal}

In \cite{PS13}, Peng and Speilman proposed a squaring-based parallel algorithm for computing an approximate pseudo-inverse (preconditioner) of a Laplacian matrix. A key ingredient of their algorithm is the following recursion, which we refer to as the \emph{Peng-Spielman squaring recursion (PS-recursion)}, for inverting $\mI - \mW$ when $\norm{\mW} < 1$.\footnote{Technically the recursion was introduced for symmetric $\mW$. 
}
\begin{equation}
    (\mI - \mw)^{-1}
= \frac{1}{2} \left[\mI + (\mI + \mw) (\mI - \mw^2)^{-1} (\mI + \mw)\right]
\end{equation}

Leveraging the PS-recursion a natural appproach to find an approximate inverse of $\mI - \mw$, is to compute a sparse matrix $\widetilde{\mw}$ such that $\mI - \widetilde{\mw} \approx \mI - \mw^2$, and then compute the inverse of $\mI - \widetilde{\mw}$ recursively. \cite{PS13} showed that for symmetric Laplacians, using spectral approximation this recursion can be leveraged to produce a constant preconditioner with logarithmic depth.

Here we prove that the same algorithm gives approximate pre-conditioner for $\mI - \mw$ if $\mw$ is normal but not necessarily symmetric and the square approximations are normalized SV approximations.  
We think \Cref{thm:squaringBased} gives hope that a similar approach can expand the proof to non-normal matrices. If proved, this will result in simple but space and time efficient algorithms for solving directed Laplacian systems. For example, this could simplify the space efficient solver introduced in \cite{AKMPSV20} which in turn is used to estimate random walk probabilities to a high precision.

Below we give the main theorem which formally defines the PS-recursion and states the main result of this section. This theorem uses a generalization of the spectral norm to a (semi-)norm induced by a PSD matrix, which we first define.  

\begin{definition}
    Given PSD $\mF\in \C^{n\times n}$, we let $\|\cdot\|_{\mF}$ be the (semi)norm on vectors where $\|x\|_{\mF}\defeq \sqrt{x^*\mF x}$, and define the matrix norm $\|\mA\|_{\mF} \defeq \sup_{x\in \C^n \setminus \{\vzero\}} \frac{\|\mA x\|_{\mF}}{\|x\|_{\mF}}$.
\end{definition}

\begin{theorem}\label{thm:squaringBased}
For a normal matrix $\mw \in \C^{n\times n}$ with $\| \mw \| \leq 1$, let $\mw = \mw_0, \ldots, \mw_{k-1}$ be a sequence of matrices such that for all $i \in [k - 1]$ we have $\mw_{i} \svn{\epsilon} \mw_{i-1}^2$ for $\eps \leq 1/(4k)$.
Define
\[
\mproj_i \defeq \frac{1}{2} \left[\mI + (\mI + \mw_i) \mproj_{i+1} (\mI + \mw_i)\right] 
\quad \text{for all} \quad
0 \leq i < k
\]
where $\mproj_k$ is a matrix such that 
\[
\left\| (\mI - \mw^{2^k})^{\frac{1}{2}} \left[ \mproj_k - (\mI - \mw^{2^k})^+ \right] (\mI - \mw^{2^k})^{\frac{1}{2}}\right\| \leq O(k \epsilon)\,.
\]
Then for $\mb \defeq ((\mI - \mw)^{1/2})^* (\mI - \mw)^{1/2}$ 
\begin{equation}
\left\| \mproj_0 (\mI - \mw) - \mI \right\|_{\mb}
    =
    \left\| (\mI - \mw)^{\frac{1}{2}} \left[ \mproj_0 - (\mI - \mw)^+ \right] (\mI - \mw)^{\frac{1}{2}}\right\| 
    \leq 
    O(k^2 \epsilon)\,.
\end{equation}
\end{theorem}
Note that in the above statement, the square root of $\mI - \mw$ is well defined by the assumption that $\mw$ is normal. In the case of non-normal matrices, which is outside the scope of this paper, one can consider the Maclaurin series for $(1-z)^{1/2}$ to work with $(\mI - \mw)^{1/2}$. 

We do not give a full algorithm for solving systems defined by Theorem~\ref{thm:squaringBased} as the main purpose is to show that the error analysis of the squaring algorithm can be expanded to normal asymmetric matrices, and we already know space and time efficient algorithms~\cite{CKKPPRS18,AKMPSV20} for solving a larger class of matrices which includes non-normal $\mw$'s. 
However, one can still get an actual squaring based solver from \Cref{thm:squaringBased}. For example, for a normal $\mw$ corresponding to a random walk matrix of an Eulerian digraph one can carry the same steps in \cite{AKMPSV20} by setting $k=O(\log n)$ and $\mproj_k = \mI - \frac{\vones \vones^\top}{n}$ to get an actual squaring based solver using \Cref{thm:squaringBased} and the derandomized square sparsification described in \Cref{sec:derandomized-square}.

We prove the statement of the theorem, following a few helper lemmas. First, we show the error of $\mw_i$'s approximating $\mw^{2^i}$'s grows additively.

\begin{lemma}\label{lem:sapprox_powers_rel}
Let $\mW_i$ be defined as in \Cref{thm:squaringBased}. Then  $\mw_{i} \svn{2 i \epsilon} \mw^{2^i}$ for all for $i \in [k]$.
\end{lemma}
\begin{proof}
Recall that $\mw_j \svn{\eps} \mw_{j-1}^2$ for every $j$. Thus, by \Cref{lem:svproducts} we have for every $j\leq i$,
\[(\mw_j)^{2^{i-j}} \svn{\eps+O(\eps^2)} (\mw_{j-1})^{2^{i-j+1}}\,.
\]
The result then follows from \Cref{lem:triangle}.
\end{proof}

Next we give a general technical tool regarding normal matrices.

\begin{lemma}\label{lem:sq_normal_prop}
For a normal matrix $\mv$ with $\|\mv\| \leq 1$  
\begin{equation}\label{eq:normal_square_lemma_ineq1}
    \|(\mI + \mv)^{\frac{1}{2}}\| \leq \sqrt{2}
\end{equation} 
and if $\tV\svn{\delta} \mv$ then 
\begin{equation}\label{eq:normal_square_lemma_ineq2}
\left\|(\mI - \mv)^{\frac{1}{2}}(\tV - \mv) ((\mI - \mv^2)^{+})^{\frac{1}{2}}\right\| \leq O(\delta).
\end{equation}
\end{lemma}
\begin{proof}
Note that since $\mv$ is normal, we can write it as $\mv = \mU \md \mU^*$ where $\mU$ is unitary and $\md$ is diagonal and $\| \md \| \leq 1$. Then it is easy to see 
\begin{equation}
    \|(\mI + \mv)^{\frac{1}{2}}\| 
    = \|\mU (\mI + \md)^{\frac{1}{2}} \mU^*\|
    = \| (\mI + \md)^{\frac{1}{2}} \|
    \leq \sqrt{2}.
\end{equation}
To prove \eqref{eq:normal_square_lemma_ineq2}, given the SV approximation definition, for unit vectors $\vx, \vy$ we have,
\begin{align*}
\left| \vx^* (\mI - \mv)^{\frac{1}{2}} (\tV - \mv) ((\mI - \mv^2)^{+})^{\frac{1}{2}} \vy
\right|
&\leq
\frac{\delta}{2} \cdot 
\vx^* (\mI - \mv)^{\frac{1}{2}} (\mI - \mv \mv^*) (\mI - \mv^*)^{\frac{1}{2}} \vx\\
&+
\frac{\delta}{2} \cdot \vy^* ((\mI - (\mv^*)^2)^+)^{\frac{1}{2}} (\mI -  \mv^* \mv) ((\mI - \mv^2)^{+})^{\frac{1}{2}} \vy
\end{align*}
We bound each term on the right hand side separately. Since $\mv$ is normal we can prove bounds on the norm of $(\mI - \mv)^{\frac{1}{2}} (\mI - \mv \mv^*) (\mI - \mv^*)^{\frac{1}{2}}$ and $((\mI - (\mv^*)^2)^+)^{\frac{1}{2}} (\mI -  \mv^* \mv) ((\mI - \mv^2)^+)^{\frac{1}{2}}$ by looking at complex numbers. Note that for $z \in \C$ and $|z| \leq 1$,
$$
| (1-z)^{\frac{1}{2}} (1-z\cdot z^*)^2 (1-z^*)^{\frac{1}{2}} | = |1-z| \cdot (1-|z|^2) \leq 2,
$$
thus $\|(\mI - \mv)^{\frac{1}{2}} (\mI - \mv \mv^*) (\mI - \mv^*)^{\frac{1}{2}}\| \leq 2$. Next,
$$
\left|(1-z^{*2})^{-\frac{1}{2}}(1-z\cdot z^*)^2(1-z^2)^{-\frac{1}{2}}\right| =  \frac{1-|z|^2}{|1-z^2|} \leq 1
$$
because $1-|z|^2 \leq |1-z^2|$ by triangle inequality. Thus, 
$$
\|
((\mI - (\mv^*)^2)^+)^{\frac{1}{2}} (\mI -  \mv^* \mv) ((\mI - \mv^2)^+)^{\frac{1}{2}}
\|
\leq 1.
$$
Completing the proof that,
\begin{align*}
\left| \vx^* (\mI - \mv)^{\frac{1}{2}} (\tV - \mv) ((\mI - \mv^2)^+)^{\frac{1}{2}} \vy
\right|
\leq
\frac{3}{2}\cdot \delta.
\end{align*}
\end{proof}

Leveraging these tools we give the central lemma towards proving \Cref{thm:squaringBased} and conclude by proving \Cref{thm:squaringBased}.

\begin{lemma}\label{lem:squaringBased}
For $\epsilon$ and $\mproj_i$'s defined as in \Cref{thm:squaringBased},
\begin{equation}
    \left\| (\mI - \mw^{2^i})^{\frac{1}{2}} \left[ \mproj_i - (\mI - \mw^{2^i})^{+} \right] (\mI - \mw^{2^i})^{\frac{1}{2}} \right\| \leq \epsilon_i 
\end{equation}
where $\epsilon_i = (1+O(i\cdot \epsilon) )\epsilon_{i+1} + O(i \cdot \epsilon)$ for $0 \leq i < k$ and $\epsilon_k =O (k \cdot \epsilon)$.
\end{lemma}
\begin{proof}
We prove the statement by backward induction. Note that the statement is true for $i=k$ by the assumption about $\mproj_k$ in \Cref{thm:squaringBased}. Let 
$$
\mq_i = \frac{1}{2} \left[\mI + (\mI + \mw_i) (\mI - \mw^{2^{i+1}})^{+} (\mI + \mw_i)\right],
$$
$$
\mr_i = \frac{1}{2} \left[\mI + (\mI + \mw_i) (\mI - \mw^{2^{i+1}})^{+} (\mI + \mw^{2^i})\right],
$$
$$
\ms_i = \frac{1}{2} \left[\mI +  (\mI + \mw^{2^i}) (\mI - \mw^{2^{i+1}})^{+} (\mI + \mw^{2^i})\right] = (\mI - \mw^{2^i})^{+}.
$$
Then by triangle inequality we have,
\begin{align}
\| (\mI - \mw^{2^i})^{\frac{1}{2}} \left[ \mproj_i - (\mI - \mw^{2^i})^{+} \right] (\mI - \mw^{2^i})^{\frac{1}{2}} \|
&\leq
\| (\mI - \mw^{2^i})^{\frac{1}{2}} (\mproj_i - \mq_i) (\mI - \mw^{2^i})^{\frac{1}{2}} \| \label{eq:sqr_first}\\
& +
\| (\mI - \mw^{2^i})^{\frac{1}{2}} (\mq_i - \mr_i) (\mI - \mw^{2^i})^{\frac{1}{2}} \| \label{eq:sqr_second} \\
& +
\| (\mI - \mw^{2^i})^{\frac{1}{2}} (\mr_i - \ms_i) (\mI - \mw^{2^i})^{\frac{1}{2}} \|.\label{eq:sqr_third}
\end{align}
Next, we bound each term separately. For \eqref{eq:sqr_first}, we have
\begin{align*}
&\left\| (\mI - \mw^{2^i})^{\frac{1}{2}} (\mproj_i - \mq_i) (\mI - \mw^{2^i})^{\frac{1}{2}} \right\| \\
=
\frac{1}{2}
&\left\| 
(\mI - \mw^{2^i})^{\frac{1}{2}}
(\mI + \mw_i)
\left(\mproj_{i+1} - (\mI - \mw^{2^{i+1}})^{+} \right)
(\mI + \mw_i)
(\mI - \mw^{2^i})^{\frac{1}{2}}
\right\| \\
\leq
\frac{1}{2}
&\left\| 
(\mI - \mw^{2^i})^{\frac{1}{2}}
(\mI + \mw_i)
((\mI - \mw^{2^{i+1}})^{\frac{1}{2}})^+
\right\| \\
\cdot
&\left\|
(\mI - \mw^{2^{i+1}})^{\frac{1}{2}}
\left(\mproj_{i+1} - (\mI - \mw^{2^{i+1}})^{+} \right)
(\mI - \mw^{2^{i+1}})^{\frac{1}{2}}
\right\|  \\
\cdot 
&\left\| 
((\mI - \mw^{2^{i+1}})^{\frac{1}{2}})^+
(\mI + \mw_i)
(\mI - \mw^{2^i})^{\frac{1}{2}}
\right\|.
\end{align*}
Note that the middle term above satisfies
\[\left\|
(\mI - \mw^{2^{i+1}})^{\frac{1}{2}}
\left(\mproj_{i+1} - (\mI - \mw^{2^{i+1}})^{+} \right)
(\mI - \mw^{2^{i+1}})^{\frac{1}{2}}
\right\| \leq \epsilon_{i+1}\] 
by the induction hypothesis. Next, we bound the first and third terms similarly:
\begin{align}
\left\| 
\frac{(\mI - \mw^{2^i})^{\frac{1}{2}}
(\mI + \mw_i)
((\mI - \mw^{2^{i+1}})^{\frac{1}{2}})^+}{\sqrt{2}}
\right\| 
\leq
&\left\| 
\frac{(\mI - \mw^{2^i})^{\frac{1}{2}}
(\mI + \mw^{2^i})
((\mI - \mw^{2^{i+1}})^{\frac{1}{2}})^+}{\sqrt{2}}
\right\| + \label{eq:sqr_lemma_w_i_bound} \\ 
&\left\| 
\frac{(\mI - \mw^{2^i})^{\frac{1}{2}}
(\mw_i - \mw^{2^i})
((\mI - \mw^{2^{i+1}})^{\frac{1}{2}})^+}{\sqrt{2}}
\right\| \nonumber \\
& \leq
1 + O(i\cdot \epsilon)\nonumber
\end{align}
as $\left\| \frac{(\mI - \mw^{2^i})^{\frac{1}{2}}
(\mI + \mw^{2^i})
((\mI - \mw^{2^{i+1}})^{\frac{1}{2}})^+}{\sqrt{2}}
\right\| \leq 1$ by \Cref{eq:normal_square_lemma_ineq1} in \Cref{lem:sq_normal_prop} and $\left\| 
\frac{(\mI - \mw^{2^i})^{\frac{1}{2}}
(\mw_i - \mw^{2^i})
((\mI - \mw^{2^{i+1}})^{\frac{1}{2}})^+}{\sqrt{2}}
\right\| \leq O(i \cdot \epsilon)$ by \Cref{lem:sapprox_powers_rel} and \Cref{eq:normal_square_lemma_ineq2} in \Cref{lem:sq_normal_prop}. 
Combining all the above,
$$
\left\| (\mI - \mw^{2^i})^{\frac{1}{2}} (\mproj_i - \mq_i) (\mI - \mw^{2^i})^{\frac{1}{2}} \right\| 
\leq
(1+O(i\cdot \epsilon))^2 \epsilon_{i+1},
$$
giving the bound for \eqref{eq:sqr_first}. Next, for \eqref{eq:sqr_second}, we have
\begin{align*}
&
\left\| (\mI - \mw^{2^i})^{\frac{1}{2}} (\mq_i - \mr_i) (\mI - \mw^{2^i})^{\frac{1}{2}}
\right\| 
= \\
&
\frac{1}{2}\left\| (\mI - \mw^{2^i})^{\frac{1}{2}} (\mI + \mw_i) (\mI - \mw^{2^{i+1}})^{+} (\mw_i - \mw^{2^i}) (\mI - \mw^{2^i})^{\frac{1}{2}}
\right\| \leq \\
&
\left\| 
\frac{(\mI - \mw^{2^i})^{\frac{1}{2}} (\mI + \mw_i) ((\mI - \mw^{2^{i+1}})^{\frac{1}{2}})^+}{\sqrt{2}}
\right\|
\cdot
\left\|
\frac{((\mI - \mw^{2^{i+1}})^{\frac{1}{2}})^+ (\mw_i - \mw^{2^i}) (\mI - \mw^{2^i})^{\frac{1}{2}}}{\sqrt{2}}
\right\|
\end{align*}
We can get bounds for the two terms above similar to \eqref{eq:sqr_lemma_w_i_bound}. Thus we get, 
\begin{align*}
 &
\left\| (\mI - \mw^{2^i})^{\frac{1}{2}} (\mq_i - \mr_i) (\mI - \mw^{2^i})^{\frac{1}{2}}
\right\| 
\leq \\
&
\left\| 
\frac{(\mI - \mw^{2^i})^{\frac{1}{2}} (\mI + \mw_i) ((\mI - \mw^{2^{i+1}})^{\frac{1}{2}})^+}{\sqrt{2}}
\right\|
\cdot
\left\|
\frac{((\mI - \mw^{2^{i+1}})^{\frac{1}{2}})^+ (\mw_i - \mw^{2^i}) (\mI - \mw^{2^i})^{\frac{1}{2}}}{\sqrt{2}}
\right\| 
\leq \\
&\left(1+O(i\cdot\epsilon)\right) \cdot O(i\cdot \epsilon) \leq O(i\cdot \epsilon).
\end{align*}
Finally, for \eqref{eq:sqr_third}, we have
\begin{align*}
&
\| (\mI - \mw^{2^i})^{\frac{1}{2}} (\mr_i - \ms_i) (\mI - \mw^{2^i})^{\frac{1}{2}} \|
= \\
&
\frac{1}{2}\left\| (\mI - \mw^{2^i})^{\frac{1}{2}} (\mw_i - \mw^{2^i})  (\mI - \mw^{2^{i+1}})^{+} (\mI + \mw^{2^i})(\mI - \mw^{2^i})^{\frac{1}{2}}
\right\| \leq \\
&
\left\| 
\frac{(\mI - \mw^{2^i})^{\frac{1}{2}} (\mw_i - \mw^{2^i}) ((\mI - \mw^{2^{i+1}})^{\frac{1}{2}})^+}{\sqrt{2}}
\right\|
\cdot
\left\|
\frac{((\mI - \mw^{2^{i+1}})^{\frac{1}{2}})^+ (\mI + \mw^{2^i})(\mI - \mw^{2^i})^{\frac{1}{2}}}{\sqrt{2}}
\right\|
\end{align*}
Further, similar to the earlier upper bounds, we have $\left\| 
\frac{(\mI - \mw^{2^i})^{\frac{1}{2}} (\mw_i - \mw^{2^i}) ((\mI - \mw^{2^{i+1}})^{\frac{1}{2}})^+}{\sqrt{2}}
\right\| \leq O(i\cdot \epsilon) $ and $\left\|
\frac{((\mI - \mw^{2^{i+1}})^{\frac{1}{2}})^+ (\mI + \mw^{2^i})(\mI - \mw^{2^i})^{\frac{1}{2}}}{\sqrt{2}}
\right\| \leq 1$. Thus
$$
\| (\mI - \mw^{2^i})^{\frac{1}{2}} (\mr_i - \ms_i) (\mI - \mw^{2^i})^{\frac{1}{2}} \| \leq O(i \cdot \epsilon).
$$
\end{proof}

\begin{proof}[Proof of \Cref{thm:squaringBased}]
For $\epsilon < \frac{1}{k^2}$, we can solve for $\epsilon_0$ from $\epsilon_i = (1+O(i\cdot \epsilon) )\epsilon_{i+1} + O(i \cdot \epsilon)$ in \Cref{lem:squaringBased} and get $\epsilon_0 = O(\epsilon \cdot k^2)$. To complete the proof, we show
$$
\| \mproj_0 (\mI - \mw) - \mI \|_{\mb}
    =
 \| (\mI - \mw)^{\frac{1}{2}} \left[ \mproj_0 - (\mI - \mw)^+ \right] (\mI - \mw)^{\frac{1}{2}} \|. 
$$
Note that,
\begin{align*}
    &\|(\mI - \mw)^{\frac{1}{2}}  \left[ \mproj_0 - (\mI - \mw)^+ \right] (\mI - \mw)^{\frac{1}{2}} \| \\
    & = \| (\mI - \mw)^{1/2} \left[ \mproj_0 (\mI - \mw) - \mI \right] (\mI - \mw)^{+/2} \| \\
    & = \left\| \left[ (\mI - \mw)^{*/2} (\mI - \mw)^{1/2}\right]^{1/2} \left[ \mproj_0 (\mI - \mw) - \mI \right] \left[(\mI - \mw)^{+/2} ((\mI - \mw)^*)^{+/2} \right]^{1/2} \right\|  \\
    & = \left\| \left[ (\mI - \mw)^{*/2} (\mI - \mw)^{1/2}\right]^{1/2} \left[ \mproj_0 (\mI - \mw) - \mI \right] \left[(\mI - \mw)^{*/2} (\mI - \mw)^{1/2} \right]^{+/2} \right\| \\
    & = \| \mproj_0 (\mI - \mw) - \mI \|_{\mb} 
\end{align*}
where we used \Cref{lem:norm_sqrt}.
\end{proof} 

\section*{Acknowledgements}
We thank Jack Murtagh for his helpful collaboration at the early stage of research on this research and we thank Jonathan Kelner for helpful conversations at various stages of work on this project.

\bibliographystyle{alpha}
\bibliography{ref}

\appendix

\section{SV Approximation Proofs}
\label{app:SVproofs}

In this section we provide proofs of equivalences for definitions related to SV-approximation (\Cref{sec:proofs:equivalences}), properties of SV-approximation (\Cref{sec:proofs:properties}), and separations of SV-approximation from other notions of approximation (\Cref{sec:proofs:separation}).

We recall results useful in the proofs:
\begin{fact}[\cite{CKPPRSV17}]\label{fact:SCPSD}
    Let $\mA \in \C^{n\times n}$ be a Hermitian matrix and $I\subset [n]$ be arbitrary. Then $\mA$ is PSD if and only if $\mA_{I,I}$ is PSD and $\SC_I(\mA)$ is PSD.
\end{fact}
We require a result that matrix approximation is approximately preserved under Schur complements:
\begin{theorem}[Theorem 4.8 \cite{AKMPSV20}]\label{thm:schurStandard}
    Given $\mN,\tN\in \C^{n\times n}$ such that $\|\mN\|\leq 1$, suppose $\tN$ is an $\eps$-approximation of $\mN$ with respect to $\mS_{\mN}$. Then for $F\subset [n]$ such that $\mI_{F,F}-\mN_{F,F}$ is invertible, we have that $\mI_{F^c}-\SC_F(\mI-\tN)$ is an $\eps+O(\eps^2)$ approximation of $\mI_{F^c}-\SC_F(\mI-\mN)$ with respect to $\mS_{\mI_{F^c}-\SC_F(\mI-\mN)}$. 
\end{theorem}

We also prove a useful property of matri approximation:
\begin{restatable}[Manipulating matrix approximation]{lemma}{leftRightError} \label{lem:leftrighterror}
    Let $\mA, \tA \in \C^{m\times n}$, let $\mEL\in \C^{m\times m},\mER\in \C^{n\times n}$ be PSD and Hermitian, and assume that $\mA$ is
    an $\eps$-approximation of $\tA$ with respect to error matrices $\mEL$ and $\mER$, for
    some $\eps\geq 0$.
    \begin{enumerate}
        \item For $\mU \in \C^{m'\times m}$, $\mV\in \C^{n\times n'}$,
        $\mU\tA\mV$ is an $\eps$-approximation of $\mU\mA\mV$ with
        respect to error matrices $\mEL'$ and $\mER'$, where
        $\mEL' = \mU\mEL\mU^*$ and $\mER' = \mV^*\mER\mV$.
        
        \item For PSD Hermitian $\mEL'\in \C^{m\times m}$ and $\mER'\in \C^{n\times n}$ a
        with $\mEL\preceq c\mEL'$ and $\mER\preceq c\mER'$ for a constant $c\geq 0$, then $\tA$ is a $c\eps$-approximation of $\mA$ with
        respect to error matrices $\mEL'$ and $\mER'$.
    \end{enumerate}
\end{restatable}
\begin{proof}
~
    \begin{enumerate}
        \item For arbitrary $x \in \C^m$ and $y \in \C^n$ the definition of $\eps$-approximation implies,
        \[\left|x^*(\mU\tA\mV-\mU\mA\mV) y\right| = \left|x^*\mU(\tA-\mA)\mV y\right|\leq \frac{\eps}{2}\left(x^*\mU\mEL\mU^*x+y^*\mV^*\mER\mV y\right).\]
        
        \item For arbitrary $x \in \C^m,y \in \C^n$,
        \[\left|x^*(\tA-\mA)y\right| \leq \frac{\eps}{2}\left(x^*\mEL x+y^*\mER y\right) \leq \frac{c\eps}{2}\left(x^*\mEL' x+y^*\mER'y\right). \qedhere
        \]
    \end{enumerate}
\end{proof}

\subsection{Equivalences}
\label{sec:proofs:equivalences}
\svDefined*
\begin{proof}
    We first show Conditions $1-5$ are equivalent.
    \begin{itemize}
        \item[\ref{itm:norm}$\ra$ \ref{itm:rightpsd}] Condition~\ref{itm:norm} implies for all $y$, $y^*\Dout^{+/2}\mA^* \Din^+\mA \Dout^{+/2} y \leq y^* y$ and letting $y\la \Dout^{1/2}y$ and using that $\ker(\Dout)\subseteq \ker(\mA)$ implies Condition \ref{itm:rightpsd}. 
        
        \item[\ref{itm:rightpsd} $\ra$ \ref{itm:norm}] Condition~\ref{itm:rightpsd} implies for all $y$ we have $y^*\mA^*\Din^+\mA y\leq y^*\Dout y$ and taking $y\la \Dout^{+/2}y$ we have $y^*\Dout^{+/2}\mA^*\Din^+\mA\Dout^{+/2} y \leq y^*\mPi_{\Dout} y \leq y^* y$ which implies Condition~\ref{itm:norm}.
        \item[\ref{itm:rightpsd}$\ra$ \ref{itm:twobytwoallz}] Letting $z$ with $|z|\leq 1$ be arbitrary, we have that by Condition~\ref{itm:rightpsd} $\Din - z^*\mA^*\Dout^{+}z\mA=\Din - \mA^*\Dout^{+}\mA \succeq \mzero$, and since $\Din \succeq \mzero$ we apply Fact~\ref{fact:SCPSD} and conclude Condition~\ref{itm:twobytwoallz}. An analogous argument shows \ref{itm:leftpsd}$\ra$ \ref{itm:twobytwoallz}. The fact that Fact~\ref{fact:SCPSD} is an equivalence then implies \ref{itm:twobytwoallz}$\ra$\ref{itm:leftpsd} and \ref{itm:twobytwoallz}$\ra$\ref{itm:rightpsd}.  
        \item[\ref{itm:twobytwoallz} $\ra$ \ref{itm:twobysomez}] Immediate.
        \item[\ref{itm:twobysomez} $\ra$ \ref{itm:leftpsd}] By Fact~\ref{fact:SCPSD} and the fact that $z^*z=1$ we have $\Din - z^*\mA^*\Dout^{+}z\mA=\Din - \mA^*\Dout^{+}\mA \succeq \mzero$.
    \end{itemize}
    We next show the latter set of conditions.
    \begin{itemize}
        \item[\ref{itm:graph}$\ra$\ref{itm:diagdom}] For all $i\in [n]$ we have $[\dout]_i=\sum_j\mA_{i,j}=\sum_j|\mA_{i,j}|$.
        \item[\ref{itm:diagdom} $\ra$ \ref{itm:twobytwoallz}] We use that a Hermitian matrix that is diagonally dominant is positive semidefinite. We have $(\Din)_{i,i}\geq \sum_{j}|z\mA_{i,j}|$ and $(\Dout)_{i,i}\geq \sum_j |z^*\mA^*_{j,i}|$ so $\begin{bmatrix}\Din & \mA\\\mA^* & \Dout\end{bmatrix}$ is diagonally dominant and thus PSD. \qedhere
    \end{itemize}
\end{proof}

\newcommand{\vt}[2]{\begin{bmatrix} #1 \\ #2\end{bmatrix}}
\newcommand{\vtt}[2]{\begin{bmatrix} #1^*  & #2^* \end{bmatrix}}
\svLifts*
\begin{proof}
    We prove $\ref{SVequiv:main}\iff \ref{SVequiv:bip}$ and $\ref{SVequiv:main}\ra \ref{SVequiv:allz}\ra \ref{SVequiv:somez}\ra \ref{SVequiv:main}$.
    \begin{itemize}
        \item[$\ref{SVequiv:main}\iff \ref{SVequiv:bip}$]  
        Suppose \Cref{SVequiv:main} holds. Then for arbitrary test vectors $x=\vt{x_1}{x_2}$, $y = \vt{y_1}{y_2}$ where $x_1,y_1 \in \C^m$ and $x_2,y_2 \in \C^n$ we have
        \begin{align*}
            \left|x^*(\mB-\tB)y\right| &\leq \left|x_1^*(\mA-\tA)y_2\right|+\left|x_2^*(\mA^*-\tA^*)y_1\right|\\
            &\leq \frac{\eps}{4}\left(x_1^*\Din x_1+y_2^*\Dout y_2+ x_1^*\mA\Dout^+\mA^* x_1+y_2^*\mA^*\Din^+\mA y_2\right)\\
            &\quad\quad +\frac{\eps}{4}\left(x_2^*\Dout x_2+y_1^*\Din y_1+ x_2^*\mA^*\Din^+\mA x_2+y_1^*\mA\Dout^+\mA^* y_1\right)\\
            &= \frac{\eps}{4}\left(x^*\mD x+y^*\mD y + x^*\mB \mD^{-1}\mB^* x+ y^* \mB^* \mD^{-1}\mB y\right).
        \end{align*}
        Furthermore, $\ker(\mD)\subseteq \lker(\mB),\ker(\mD)\subseteq \rker(\mB)$ and $\mD-\mB\mD^+\mB \succeq \mzero$ by \Cref{itm:twobytwoallz} of \Cref{lem:svdefined}. In the other direction, we obtain $\ref{SVequiv:bip} \ra \ref{SVequiv:main}$ by considering the set of test vectors $\vt{x_1}{0^n}$ and $\vt{0^m}{y_2}$.
        
        \item[$\ref{SVequiv:main}\ra$ \ref{SVequiv:allz}:] 
        Fix arbitrary $z$ with $|z|\leq 1$.
        We have that $\mE\succeq \mzero$ by \Cref{itm:twobytwoallz} of Lemma~\ref{lem:svdefined}. Then for arbitrary test vectors $x=\vt{x_1}{x_2},y=\vt{y_1}{y_2}$, let $x'=\vt{x_1}{x_2'},y'=\vt{y_1}{y_2'}$ where 
    
        \[x_2' = \argmin_{v \in \C^n}\vtt{x_1}{v}\begin{bmatrix}\Din & z\mA\\ z^*\mA^* & \Dout\end{bmatrix}\vt{x_1}{v}, \quad y_1' = \argmin_{v \in \C^m}\vtt{y_2}{v}\begin{bmatrix}\Dout & z^*\mA^*\\ z\mA & \Din\end{bmatrix}\vt{y_2}{v}.\]
        Thus
        \[x'^*\begin{bmatrix}\Din & z\mA\\ z^*\mA^* & \Dout\end{bmatrix}x' = x_1^*\SC\left(\begin{bmatrix}\Din & z\mA\\ z^*\mA^* & \Dout\end{bmatrix}\right)x_1, \quad y'^*\begin{bmatrix}\Din & z\mA\\ z^*\mA^* & \Dout\end{bmatrix}y' = y_2^*\SC\left(\begin{bmatrix}\Dout & z^*\mA^*\\ z\mA & \Din\end{bmatrix}\right)y_2.
        \]
        Then we have
        \begin{align*}
            \left|x^*(\tC-\mC)y\right|
            &=\left|(x_1z^*)^*(\tA-\mA)(y_2z)\right|\\
            &\leq \frac{\eps}{4}\left(x_1^*(\Din-z\mA\Dout^+\mA^*z^*)x_1 + y_2^*(\Dout -z^*\mA^*\Din^+\mA z)y_2\right)\\
            &= \frac{\eps}{4}\left(x_1^*\SC\left(\begin{bmatrix}\Din & z^*\mA^*\\ z\mA & \Dout\end{bmatrix}\right)x_1 +y_2^*\SC\left(\begin{bmatrix}\Dout & z\mA\\ z^*\mA^* & \Din\end{bmatrix}\right)y_2 \right)\\
            &= \frac{\eps}{4}\left(x'^*\begin{bmatrix}\Din & z\mA\\ z^*\mA^* & \Dout\end{bmatrix}x'+y'^*\begin{bmatrix}\Din & z\mA\\ z^*\mA^* & \Dout\end{bmatrix}y'\right)\\
            &\leq \frac{\eps}{4}\left(x^*\mE x+y^*\mE y\right).
        \end{align*}
        
        \item[$\ref{SVequiv:allz}\ra \ref{SVequiv:somez}$] Immediate.
        
        \item[$\ref{SVequiv:somez}\ra \ref{SVequiv:main}$] For arbitrary test vectors $x_1 \in \C^m,y_2 \in \C^n$, let $x=\vt{zx_1}{x_2},y=\vt{y_1}{z^*y_2}$ where 
    
        \[x_2 = \argmin_{v \in \C^n}\vtt{z^*x_1}{v}\begin{bmatrix}\Din & z\mA\\ z^*\mA^* & \Dout\end{bmatrix}\vt{zx_1}{v}, \quad y_1 = \argmin_{v \in \C^m}\vtt{zy_2}{v}\begin{bmatrix}\Dout & z^*\mA^*\\ z\mA & \Din\end{bmatrix}\vt{z^*y_2}{v}\]
        Then 
        \begin{align*}
            \left|x_1^*(\tA-\mA)y_2\right| &= \left|x^*(\tC-\mC)x\right|\\
            &\leq \frac{\eps}{4}\left( x^*\mE x+y^*\mE y\right)\\
            &= \frac{\eps}{4}\left(z^*x_1^*\SC\left(\begin{bmatrix}\Din & z^*\mA^*\\ z\mA & \Dout\end{bmatrix}\right)x_1z +zy_2^*\SC\left(\begin{bmatrix}\Dout & z\mA\\ z^*\mA^* & \Din\end{bmatrix}\right)z^*y_2\right)\\
            &=\frac{\eps}{4}\left(z^*x_1^*(\Din-z\mA\Dout^+\mA^* z^*)zx_1+zy_2^*(\Dout-z^*\mA^*\Din^+\mA z)z^*y_2\right)\\
            &= \frac{\eps}{4}\left(x_1^*(\Din-\mA\Dout^+\mA^* )x_1+y_2^*(\Dout-\mA^*\Din^+\mA )y_2\right).\qedhere
        \end{align*}
    \end{itemize} 
\end{proof}

\normalizedSV*
\begin{proof}
    Suppose $\tA \svgeneral{\Din}{\Dout}{\eps} \mA$, i.e., $\tA$ is an $\eps/2$-approximation of $\mA$ with respect to
    \[\mEL = \Din-\mA\Dout^+\mA^*, \quad\quad \mER = \Dout-\mA\Din^+\mA^*.
    \]
    Applying Lemma~\ref{lem:leftrighterror} with $\mU=\Din^{+/2},\mV=\Dout^{+/2}$ we obtain that
    $\tN$ is an $\eps/2$-approximation of $\mN$ with respect to error matrices
    \[\Din^{+/2}(\Din - \mA\Dout^+\mA^*)\Din^{+/2} \preceq \mI - (\Din^{+/2}\mA\Dout^{+/2})(\Dout^{+/2}\mA^*\Din^{+/2})
    = \mI-\mN\mN^*
    \]
    and 
    \[\Dout^{+/2}(\Dout - \mA^*\Din^+\mA)\Dout^{+/2} \preceq \mI - (\Dout^{+/2}\mA^*\Din^{+/2})(\Din^{+/2}\mA\Dout^{+/2})
    = \mI-\mN^*\mN\]
    and hence we obtain that $\tN \svn{\eps} \mN$. The other direction is analogous, and the kernel properties follow as $\ker(\Din)\subseteq \lker(\Din^{1/2}\Din^{+/2}\mA\Dout^{+/2}\Dout^{1/2})$ and $\ker(\Dout)\subseteq \rker(\Din^{1/2}\Din^{+/2}\mA\Dout^{+/2}\Dout^{1/2})$.
\end{proof}

\svRegular*
\begin{proof}~
    \begin{itemize}
        \item Suppose for every pair of unitary matrices $\mU,\mV$, we have for all $x,y\in \C^n$ that
        \[|x^*(\mU\tA\mV -\mU\mA\mV)y| \leq \frac{\eps}{2} \cdot \sqrt{x^*(\mI-\mS_{\mU\mA\mV}) x} \cdot \sqrt{y^*(\mI -\mS_{\mU\mA\mV})y}.
        \]
        We will pick $u,v,\mU,\mV$ to depend on $x,y$ and the result will follow immediately from invoking the above equation with $x,y \gets u,v$.
    
        Set $\mU=\mI$, $u=x/\|x\|$, and $v=\mA y/\|\mA y\|$. Note that $u^* v = (x^* \mA y)/(\|x\|\cdot\|y\|)$, so $u$ and $v$ have the same angle between them as $\mA^*x$ and $y$.
        Hence, there exists a unitary matrix $\mV$ such that $\mV u = \mA^* x / \|\mA^* x\|$ and $\mV v = y / \|y\|$. Substituting gives
        \[
        \frac{\left|x^*(\tA-\mA)y\right|}{\|x\|\|y\|} \leq \frac{\eps}{2} \cdot \sqrt{1-\frac{x^*(\mA\mA^*) x}{\|x\|\|\mA^* x\|}} \cdot \sqrt{1-\frac{y^*(\mA^*\mA)y}{\|y\|\|\mA y\|}}.
        \]
        
        Applying $\|\mA^*\|=\|\mA\| \leq 1$ gives the desired result.
        \item In the other direction, we start with \Cref{SVequiv:somez} of \Cref{SVequiv} with $z=1$ applied to the approximation statement $\tN\svn{\eps}\mN$ and apply \Cref{lem:leftrighterror} with left and right hand side matrices
        \[
            \mU' \gets
            \begin{bmatrix}
            \mU & \mzero \\
            \mzero & \mV^*
            \end{bmatrix}
            \qquad\text{and}\qquad
            \mv' \gets
            \begin{bmatrix}
            \mU^* & \mzero\\
            \mzero & \mV
            \end{bmatrix}
        \]
        This results in the approximation statement that for the matrices given below, $\mR,\tR$ $\epsilon/2$-approximate each other with respect to $\mE$ where
        \[
            \mR =
            \begin{bmatrix}
                \mzero & \mU \mN \mV \\
                \mzero & \mzero
            \end{bmatrix}
            \qquad
            \tR =
            \begin{bmatrix}
                \mzero & \mU \tN \mV \\
                \mzero & \mzero
            \end{bmatrix}
            \text{, and}\qquad
            \mE =
            \begin{bmatrix}
                \mU \mI \mU^* & \mU \mN \mV \\
                (\mU \mN \mV)^* & \mV^* \mI \mV
            \end{bmatrix} \preceq \begin{bmatrix}
                \mI & \mU \mN \mV \\
                (\mU \mN \mV)^* & \mI
            \end{bmatrix}.
            \]
            Hence, by \Cref{SVequiv:somez} of \Cref{SVequiv} we conclude $\mU\mN\mV\svn{\eps}\mU\tN\mV$. \qedhere
    \end{itemize}
\end{proof}

\svPerm*
We start with \Cref{SVequiv:somez} of \Cref{SVequiv} with $z=1$ applied to the approximation statement we are given in the lemma statement and apply \Cref{lem:leftrighterror} with left and right hand side matrices
    \[
        \mU' \gets
        \begin{bmatrix}
        \mU & \mzero \\
        \mzero & \mV^*
        \end{bmatrix}
        \qquad\text{and}\qquad
        \mv' \gets
        \begin{bmatrix}
        \mU^* & \mzero\\
        \mzero & \mV
        \end{bmatrix}
    \]
    This results in the approximation statement that for the matrices given below, $\mR,\tR$ $\epsilon/2$-approximate each other with respect to $\mE$ where
    \[
        \mR =
        \begin{bmatrix}
            \mzero & \mU \mA \mV \\
            \mzero & \mzero
        \end{bmatrix}
        \qquad
        \tR =
        \begin{bmatrix}
            \mzero & \mU \tA \mV \\
            \mzero & \mzero
        \end{bmatrix}
        \text{, and}\qquad
        \mE =
        \begin{bmatrix}
            \mU \mdin \mU^* & \mU \mA \mV \\
            (\mU \mA \mV)^* & \mV^* \mdout \mV
        \end{bmatrix}.
    \]
    Hence, \Cref{SVequiv:somez} of \Cref{SVequiv} is satisfied for SV approximation of $\mU\mA\mV$ and $\mU\tA\mV$ with respect to $\mU\Din\mU^*$ and $\mV^*\Dout\mV$.
    
    For the first \say{consequently} claim, observe that $\mU\mA\mV \vones = \mU \mA \vones$ since $\mV$ is a permutation. Hence,
    \[
        \mdiag(\mU\mA\mV \vones) = \mU \mdin \mU^*.
    \]
    A similar analysis shows
    \[
        \mdiag(\vones^\rTrans \mU\mA\mV ) = \mV^* \mdin \mV.
    \]
    the second statement follows from $\mU\mU^*=\mV^*\mV=\mI$.

\subsection{Properties}
\label{sec:proofs:properties}
We prove that SV approximation implies standard approximation:
\begin{restatable}{lemma}{SVimpliesSTD}\label{lem:SVimpliesSTD}
     Let $\mA,\tA \in \C^{n\times n}$ and suppose $\tA \svgeneral{\Din}{\Dout}{\eps} \mA$. Then, $\tA$ is an $\eps$-approximation of $\mA$ with respect to $\mEL=\Din-\mA$, $\mER=\Dout-\mA^*$.
\end{restatable}
\begin{proof}
    This follows immediately from specializing \Cref{SVequiv:somez} of \Cref{SVequiv} with $z=-1$ to test vectors of the form $\vt{x}{x},\vt{y}{y}$
\end{proof}
It likewise implies UC approximation:
\SVimpliesUC*
\begin{proof}
    By \Cref{SVequiv:allz} of \Cref{SVequiv} we have that for every pair of test vectors $\vt{x}{x},\vt{y}{y}$ and every unit magnitude $z$, we have
    \[\left|x^*(z\mA-z\tA)y\right| \leq \frac{\eps}{4}\cdot 2\left(x^*\mD x+y^*\mD y + \Re(zx^*\mA x+zy\mA y^*)\right)
    \]
    Then choosing $z$ to minimize $\Re(zx^*\mA x+zy\mA y^*)$, we obtain
    \[\left|x^*(\mA-\tA)y\right| \leq \frac{2\eps}{4}\left(x^*\mD x+y^*\mD y - |x^*\mA x+y\mA y^*|\right)
    \]
    Which is precisely the condition for $\tA \uc_{\eps} \mA$ with respect to $\mD$.
\end{proof}

\svArbLifts*
\begin{proof}
    Observe by expanding the definition of SV approximation that it is preserved under embedding into the first principal diagonal block of any larger matrix that is zeros elsewhere. That is,
    \begin{equation}
    \begin{bmatrix}\label{eq:diag-embed}
    \tA & \mzero^{m \times (j+k)} \\
    \mzero^{(i + \ell) \times n} & \mzero^{(i+\ell) \times (j+k)}
\end{bmatrix} \svgeneral{\Din'}{\Dout'}{\eps}
    \begin{bmatrix}
    \mA & \mzero^{m \times (j+k)} \\
    \mzero^{(i + \ell) \times n} & \mzero^{(i+\ell) \times (j+k)}
\end{bmatrix}
    \end{equation}
    where $\Din',\Dout'$ are the respective principle embeddings of the degree matrices.
    
    Let $\mU,\mV$ be permutation matrices such that 
    \[
\mU     \begin{bmatrix}
    \mA & \mzero^{m \times (j+k)} \\
    \mzero^{(i + \ell) \times n} & \mzero^{(i+\ell) \times (j+k)}
\end{bmatrix}
\mV
=
\begin{bmatrix}
    \mzero^{i \times j} & \mzero^{i \times n} & \mzero^{i \times k} \\
    \mzero^{m \times j} & \mA & \mzero^{m \times k} \\
    \mzero^{\ell \times j} & \mzero^{\ell \times n} & \mzero^{\ell \times k}
\end{bmatrix}.
    \]
    It suffices to argue that we can apply this transformation to both sides of \Cref{eq:diag-embed} and have the approximation still hold. This is implied by \Cref{cor:svPerm} with $\Din' \la \mU\Din'\mU^*$ and $\Dout'\la \mV^*\Dout\mV$. The \say{consequently} claim holds from inspecting the resulting block structure.
\end{proof}

\svProducts*
\begin{proof}
Using preservation under lifting and sums (\Cref{lem:svarblifts,lem:summability}) and that SV approximation implies standard approximation for square matrices (\Cref{lem:SVimpliesSTD}), we can obtain the approximation statement that the matrices $\mR,\tR$ given below $\epsilon/2$-standard-approximate each other with respect to the symmetrization $\mS_\mR$ of the first one:
    \[
    \mR = 
        \begin{bmatrix}
            \mI/2 & \mN_\ell & && \mzero \\
            \mzero & \mI & \mN_{\ell-1}  & & \vdots\\
            \vdots & & \ddots & \ddots &\\
            \vdots & & &\mI & \mN_1\\
            \mzero^{n \times n} &  & & \mzero &\mI/2
        \end{bmatrix}
        \qquad\text{and}\qquad
        \tR = 
        \begin{bmatrix}
            \mI/2 & \tN_\ell & && \mzero \\
            \mzero & \mI & \tN_{\ell-1}  & & \vdots\\
            \vdots & & \ddots & \ddots &\\
            \vdots & & &\mI & \tN_1\\
            \mzero^{n \times n} &  & & \mzero &\mI/2
        \end{bmatrix}.
    \]
    Taking the Schur complement by eliminating the center $n(l-1)\times n(l-1)$ block gives the following two matrices
    \[
    \mR = 
        \begin{bmatrix}
            \mI/2 & \mN_\ell \cdots \mN_2\mN_1 \\
            \mzero^{n \times n} &  \mI/2
        \end{bmatrix}
        \qquad\text{and}\qquad
        \tR = 
        \begin{bmatrix}
            \mI/2 & \tN_\ell \cdots \tN_2\tN_1 \\
            \mzero^{n \times n} &  \mI/2
        \end{bmatrix}.
    \]
    And we obtain by \Cref{thm:schurStandard} that $\tR$ is an $\epsilon/2 + O(\epsilon^2)$-approximation of $\mR$ with respect to $\mS_{\mR}$. Hence, by \Cref{SVequiv:somez} of \Cref{SVequiv}, we have
    \[\tN_\ell\cdots\tN_2\tN_1\svn{\eps + O(\epsilon^2)}\mN_\ell\cdots\mN_2\mN_1. \qedhere
    \]
\end{proof}

\svStationary*
\begin{proof}
    Let $x\la \delta x$ be an arbitrary test vector and let $y=\vones$. Then the approximation condition implies
    \begin{align*}
        \left|\delta x^*(\tA-\mA)\vones\right| &\leq \frac{\eps}{4}\left(\delta^2x^*\Dout x-\delta^2 x^*\mA\Din^+\mA^*x+\vones^{\rTrans}\Din\vones-\vones^{\rTrans}\mA^{\rTrans}\Dout^{+}\mA\vones\right)\\
        &= \frac{\eps}{4}\left(\delta^2x^*\Dout x-\delta^2 x^*\mA\Din^{+}\mA^*x\right)
    \end{align*}
    and so taking $\delta\ra 0$ we obtain that $\mA\vones=\tA\vones$. An analogous argument setting $x=\vones$ shows that $\vones^{\rTrans}\mA=\vones^{\rTrans}\tA$.
\end{proof}

\summability*
\begin{proof}
    This  follows immediately from \Cref{SVequiv:somez} of \Cref{SVequiv} being linear over $\mA$ and $\mD$. The \say{consequently} claim follows immediately from noting 
    \[\left(\sum_{i\in [k]}\mA\right)\vones=\sum_{i\in [k]}(\mA\vones) \text{ and }\left(\sum_{i\in [k]}\mA^\rTrans\right)\vones=\sum_{i\in [k]}(\mA^\rTrans\vones).\qedhere\]
\end{proof}

\triangle*
\begin{proof}
By \Cref{SVequiv:somez} of \Cref{SVequiv} we have
\[\begin{bmatrix}\Din & \mA_2^*\\ \mA_2 & \Dout \end{bmatrix} \preceq (1+\eps)\begin{bmatrix}\Din & \mA_1^*\\ \mA_1 & \Dout \end{bmatrix}.
\]
Thus for arbitrary $x\in \C^m,y\in \C^{n}$ we have
\begin{align*}
    \MoveEqLeft{\left|x^*\left(\begin{bmatrix}\mzero & \mA_3\\
    \mzero & \mzero\end{bmatrix}-\begin{bmatrix}\mzero & \mA_1\\
    \mzero & \mzero\end{bmatrix}
    \right)y\right|}\\ 
    &\leq \left|x^*\left(\begin{bmatrix}\mzero & \mA_3\\
    \mzero & \mzero\end{bmatrix}-\begin{bmatrix}\mzero & \mA_2\\
    \mzero & \mzero\end{bmatrix}
    \right)y\right|+\left|x^*\left(\begin{bmatrix}\mzero & \mA_2\\
    \mzero & \mzero\end{bmatrix}-\begin{bmatrix}\mzero & \mA_1\\
    \mzero & \mzero\end{bmatrix}
    \right)y\right|\\
    &\leq \frac{\delta}{4}\left(x^*\begin{bmatrix}\Din & \mA_2^*\\ \mA_2 & \Dout \end{bmatrix}x+y^*\begin{bmatrix}\Din & \mA_2^*\\ \mA_2 & \Dout \end{bmatrix}y\right)+\frac{\eps}{4}\left(x^*\begin{bmatrix}\Din & \mA_1^*\\ \mA_1 & \Dout \end{bmatrix}x+y^*\begin{bmatrix}\Din & \mA_1^*\\ \mA_1 & \Dout \end{bmatrix}y\right)\\
    &\leq \frac{\delta(1+\eps)}{4}\left(x^*\begin{bmatrix}\Din & \mA_1^*\\ \mA_1 & \Dout \end{bmatrix}x+y^*\begin{bmatrix}\Din & \mA_1^*\\ \mA_1 & \Dout \end{bmatrix}y\right)+\frac{\eps}{4}\left(x^*\begin{bmatrix}\Din & \mA_1^*\\ \mA_1 & \Dout \end{bmatrix}x+y^*\begin{bmatrix}\Din & \mA_1^*\\ \mA_1 & \Dout \end{bmatrix}y\right)\\
    &=\frac{\delta+\eps+\eps\delta}{4}\left(x^*\begin{bmatrix}\Din & \mA_1^*\\ \mA_1 & \Dout \end{bmatrix}x+y^*\begin{bmatrix}\Din & \mA_1^*\\ \mA_1 & \Dout \end{bmatrix}y\right)\\
\end{align*}
and so we conclude by \Cref{SVequiv:somez} of \Cref{SVequiv}. The first \say{consequently} claim follows from the fact that by \Cref{lem:svStationary} we have $\mA_2\vones=\mA\vones$ and $\mA_2^\rTrans\vones = \mA_1^\rTrans \vones$, and the second follows immediately.

For the \say{moreover} claim, we prove by induction on $i$ that $\mA_i \svgraph{i\delta/2\ell+i\delta^2/\ell} \mA_0$, which suffices as $\ell\delta/2\ell+\ell\delta^2/\ell\leq \delta$. Assuming this holds for $i$, we have
    \[\mA_{i+1} \svgraph{\delta/2\ell} \mA_i \svgraph{i\delta/2\ell +i\delta^2/\ell}\mA_0
    \]
    and thus we have
    \[\mA_{i+1} \svgraph{(i+1)\delta/2\ell + \gamma} \mA_0
    \]
    for 
    \[\gamma = \frac{i\delta^2}{\ell}+\frac{\delta}{2\ell}\left(\frac{i\delta}{2\ell}+\frac{i\delta^2}{\ell}\right) \leq \frac{i\delta^2}{\ell} + \delta^2/2\ell+\delta^3/2\ell^2 \leq \frac{(i+1)\delta^2}{\ell}\]
    as desired. The proof for normalized SV approximation is essentially identical, so we omit it.
\end{proof}

\expanderSVappx*
\begin{proof}
    Note that by definition of $\mA\svn{\lambda}\mJ$,
    \begin{align*}
    |x^* (\mA/d - \mJ) y| 
    &\leq 
    \frac{\lambda}{4} \left(x^* (\mI - \mJ \mJ^*) x + y^* (\mI - \mJ^*\mJ) y\right) \\
    &= \frac{\lambda}{4} \left(x^* (\mI - \mJ) x + y^* (\mI - \mJ) y\right)
    \end{align*}
    Therefore, $\mA/d\svn{\lambda}\mJ$ if and only if $\mA/d \uc_{\lambda/2}\mJ$. By~\cite[Lemma 5.2]{AKMPSV20} $\mA/d \uc_{\lambda/2}\mJ$ if and only if $\lambda(G)\leq 1-\lambda/2$, so the result follows.
\end{proof}

We show a slightly stronger version of \Cref{SVequiv:bip} of \Cref{SVequiv} for non-negative matrices, in that we show SV sparsification of the bipartite lift of a graph implies SV sparsification of the undirlying directed graph, even if the sparsifier is not promised to be bipartite-preserving, and even if the sparsifier is not itself undirected.
\begin{lemma}\label{lem:sv_undirected_suffices}
    Given non-negative $\mA \in \R_{\geq 0}^{m\times n}$, suppose there is nonnegative
    $\tA = \begin{bmatrix}
    \tA_{F,F}&\tA_{F,F^c}\\ \tA_{F^c,F}&\tA_{F^c,F^c}
    \end{bmatrix}$ such that $\tA\svgraph{\eps} \begin{bmatrix} \mzero & \mA^{\rTrans}\\ \mA&\mzero\end{bmatrix}$. Then
    \begin{enumerate}
        \item $\tA_{F,F}=\mzero$ and $\tA_{F^c,F^c}=\mzero$.
        \item $\tA_{F^c,F}\svgraph{\ep}\mA$ and $\tA_{F,F^c}\svgraph{\ep}\mA^{\rTrans}$. 
    \end{enumerate}
\end{lemma}
\begin{proof}
    Let $\Dout=\diag(\mA\vones)$ and $\Din = \diag(\vones^{\rTrans}\mA)$ and observe $\vones^{\rTrans}\Dout \vones - \vones^{\rTrans}\mA\Din^+\mA\vones = \vones^{\rTrans}\Din\vones - \vones^{\rTrans}\mA^{\rTrans}\Dout^+\mA^{\rTrans}\vones=0$.
    \begin{enumerate}
        \item Let $x\in \C^n$ be arbitrary, and consider the test vectors $x'=\delta\begin{bmatrix}x\\0^m\end{bmatrix},y' = \begin{bmatrix}\vones\\ \vones\end{bmatrix}$. Then from the definition of SV approximation we obtain
        \begin{align*}
            \left|\delta x^{\rTrans}\tA_{F,F}\vones+\delta x^*\tA_{F,F^c}\vones-\delta x\mA^{\rTrans}\vones\right| &= \left|x'^*\left(\begin{bmatrix} \mzero & \mA^{\rTrans}\\ \mA&\mzero\end{bmatrix}-\tA\right)y\right|\\
            &\leq \frac{\eps}{2}\cdot \delta^2 x^*(\Dout-\mA\Din^+\mA) x
        \end{align*}
        and hence $\tA_{F,F}\vones+x\tA_{F,F^c}\vones=x\mA^{\rTrans}\vones$ by taking $\delta \ra 0$. Considering the test vectors $x'=\delta\begin{bmatrix}x\\0^m\end{bmatrix},y'' = \begin{bmatrix}\vones\\ -\vones\end{bmatrix}$ we likewise obtain $x \tA_{F,F}\vones-x\tA_{F,F^c}\vones=-x\mA^{\rTrans}\vones$ and hence $x\tA_{F,F}\vones =0$. By an analogous argument we obtain $\tA_{F^c,F^c}=0$.

        \item For arbitrary $x \in \C^m,y\in \C^n$, apply the test vectors $[0^n,x],[y,0^m]$. Then we obtain
        \begin{align*}
            \left|x^*(\mA-\tA_{F^c,F})y\right| &\leq \frac{\eps}{2}\left(x^*(\mD-\mA\mD^{+}\mA^T)x+y^*(\mD-\mA^T\mD^{+}\mA)y\right)
        \end{align*}
        thus $\tA_{F^c,F}\svgraph{\ep} \mA$ and an analogous argument shows $\tA_{F,F^c}\svgraph{\ep} \mA^T$.  
    \end{enumerate}
\end{proof}

\subsection{Separations}
\label{sec:proofs:separation}
We prove the two separations separately, where we first construct the undirected graph example and later the arbitrary matrix example.
\SVUCsepGraphs*
The proof of~\Cref{itm:sepGraph} proceeds by considering the symmetric lifts of lazy directed cycles:
\begin{definition}[Lazy Directed Cycle]
    For $n\in \N$ and $\delta \in [0,1/2]$, let $C_{n,\delta}$ be the \emph{$(1/2+\delta)$-lazy directed cycle} on $n$ vertices, i.e., the directed graph with $V:= \{1,\ldots,n\}$, directed edges $\{(1,2),\ldots,(n-1,n),(n,1)\}$ with weight $(1/2-\delta)$, and directed edges $\{(i,i)\}_{i\in [n]}$ with weight $(1/2+\delta)$. Let $\mW_{n,\delta}$ be the random walk matrix of $C_{n,\delta}$.
\end{definition}
We first prove that the symmetric lifts of $(1/2+\delta)$-lazy directed cycles unit-circle approximate the $1/2$-lazy directed cycle with error proportional to $\delta$.
\begin{lemma}\label{lem:UCholds}
    For every $n\in \N$ and $\delta\in [0,1/2]$, we have $\slift{\mW_{n,\delta}}\uc_{\delta} \slift{\mW_{n,0}}$.
\end{lemma}
\begin{proof}
    Let $\tM:=\slift{\mW_{n,\delta}}$ and $\mM:=\slift{\mW_{n,0}}$.
    We have that $\ker(\mI-\mM) = \Span\{\vones\} \subseteq \ker(\mM-\tM)$ and $\ker(\mI+\mM) = \Span\{u\} \subseteq \ker(\mM-\tM)$ where $u_i = (-1)^i$ for $i$ in $[2n]$.
    
    Let $E_1$ be the edges in the lift corresponding to self-loops in the directed cycle, and let $E_2$ be the edges corresponding to non-self-loops in the directed cycle. Let $\L_e$ be the Laplacian of edge $e$. 
    \[\mI-\mM = \sum_{e\in E_1}(1/2)\L_e + \sum_{e\in E_2}(1/2)\L_e, \quad\quad \mI-\tM = \sum_{e\in E_1}(1/2+\delta)\L_e + \sum_{e\in E_2}(1/2-\delta)\L_e\]
    and hence  $\tM-\mM= \sum_{e\in E_1\cup E_2} \left(\pm \delta\right)\L_e$ and hence  $-2\delta(\mI-\mM) \preceq \tM-\mM \preceq 2\delta(\mI-\mM)$ and hence by~\Cref{lem:Hermitianapprox}, $\tM$ is a $2\delta$-standard approximation of $\mM$. Moreover, we have by an analogous argument that $\tM$ is a $2\delta$ standard approximation of $\mM$ with respect to $\mI+\mM$, and hence by \cite[Lemma 3.7]{AKMPSV20} we have that $\tM$ is a $\delta$-UC approximation of $\mM$.
\end{proof}
However, we now prove that a sufficiently high power of $\mW_{k,\eps}$ is not an $O(1)$-standard approximation of $\mW_{k,0}$, for $\eps=\Omega(1/\sqrt{k})$.
\newcommand{\Bin}{\mathrm{Bin}}

\begin{lemma}\label{lem:standardFails}
    For every $\eps>0$, there is $k=O(1/\eps^2)$ such that for every $n\geq 2k$, we have that $\mW_{n,\eps}^k$ is not a $.95$-standard approximation of $\mW_{n,0}^k$.
\end{lemma}
\begin{proof}
    \newcommand{\fl}[1]{\left\lfloor #1\right\rfloor}
    \newcommand{\cl}[1]{\left\lceil#1\right\rceil}
    Let $\mW:=\mW_{n,0}$ and $\tW:=\mW_{n,\eps}$. Let $d$ be a constant to be chosen later. Let $k=\lceil (10+d)^2/\eps^2\rceil$.  Let 
    \[S :=\{1,\ldots,\lceil d\sqrt{k}\rceil\} \text{ and } T:= \{\lfloor k/2-d\sqrt{k}/2\rfloor,\ldots,\lceil k/2+d\sqrt{k}/1.9\rceil\}.\]
    Let $x:=\vones_T$ and $y:=\vones_S$.

    Now observe that since $\tW^k$ is an $\eps_0$-standard approximation of $\mW^k$ (where the value of $\eps_0$ will be derived later),
    \begin{align*}
        \left|x^T(\mW^k-\tW^k)y\right| &\leq \frac{\eps_0}{2}(x^T\sym[\mI-\mW^k]x+y^T\sym[\mI-\mW^k]y)\\
        &\leq \frac{\eps_0}{2}(|S|+|T|) \leq \frac{\eps_0}{2}\sqrt{k}\cdot (3+d+d/2+d/1.9)\leq (1.03)\eps_0 d\sqrt{k}.
    \end{align*}
    where we assume $k$ and $d$ are sufficiently large.
    Next,
    \begin{align*}
        \MoveEqLeft{\left|x^T(\mW^k-\tW^k)y\right|}\\ &\geq|S|\left|\Pr_{i\in S}[k \text{ step walk goes from $i$ to $T$ in }\mW]-\Pr_{i\in S}[k \text{ step walk goes from $i$ to $T$ in }\tW]\right|
    \end{align*}
    We then bound both terms, where we use that $k\leq n/2$ so no walk can loop around the cycle, and $S$ and $T$ have no overlap. For every $i$, let 
    \[S_i:=\left\{\fl{k/2-d\sqrt{k}/2-i},\ldots,\cl{k/2+d\sqrt{k}/1.9-i}\right\}
    \]
    be the number of steps from starting vertex $i$ such that the final walk  vertex lies in $T$.
    \begin{align*}
        \Pr_{i\in S}[k \text{ step walk goes from $i$ to $T$ in }\mW]
        &= \E_{i\in S}\Pr\left[\Bin(k,1/2) \in S_i\right]\\ 
        &\geq \Pr\left[\Bin(k,1/2)\in [k/2\pm (.03)d\sqrt{k}]\right]\\ 
        &\geq (1-O(1/d^2)). &&\text{(Chebyshev)}
    \end{align*}
    whereas
    \begin{align*}
        \Pr_{i\in S}[k \text{ step walk goes from $i$ to $T$ in }\tW]
        &=\E_{i\in S}\Pr[\Bin(k,1/2-\eps) \in S_i]\\ 
        &\leq \Pr[\Bin(k,1/2-\eps) \geq k/2-d\sqrt{k}/2]\\
        &\leq .01
    \end{align*}
    where the final step follows as the event holding implies a deviation of at least $k/2-d\sqrt{k}/2-k/2+k\eps=k\eps -d\sqrt{k}/2\geq 10\sqrt{k}$, and so we can apply Chebyshev. Thus, letting $d$ be sufficiently large we derive
    \[
    d\sqrt{k}(.99-.01)\leq (1.03)\eps_0 d\sqrt{k}\implies \eps_0 \geq .95. \qedhere
    \]
\end{proof}

We can then combine these two observations and prove the first separation.
\begin{proof}[Proof of~\Cref{itm:sepGraph}]
    As we can prove the result for $n'=2\cdot \lfloor n/2\rfloor$ and this does not affect the asymptotics, we assume without loss of generality that $n$ is even. 

    Let $\mW:=\mW_{n/2,0}$ and $\tW:=\mW_{n/2,\eps}$ for $\eps=\Theta(1/\sqrt{n})$ chosen such that applying~\Cref{lem:standardFails} with $\eps=\eps$ results in a value of $k$ such that $n/2\geq 2k$. Let $\mM \defeq \slift{\mW}$ and $\tM \defeq \slift{\tW}$. 

   The fact that $\tM\uc_\eps \mM$ follows from~\Cref{lem:UCholds}. However, we claim that  $\tM$ is not a $.3$-SV approximation of $\mM$.
    Assuming for contradiction $\tM\svn{.3}\mM$, we obtain:
    \begin{align*}
        \tM\svn{.3}\mM &\implies \tW\svn{.3}\mW && \text{(\Cref{SVequiv:bip})}\\
        &\implies  \tW^k\svn{.8}\mW^k && \text{(\Cref{lem:svproducts})}\\
        &\implies \tW^k \text{ is a $.8$-standard approximation of } \mW^k
    \end{align*}
    which is a contradiction to~\Cref{lem:standardFails}, so we have the desired separation.
\end{proof}

\begin{proof}[Proof of \Cref{itm:sepMatr}]
Given $\alpha$, we define
\[\mW =
\begin{bmatrix}
     \alpha & 0 \\
     0 & -\alpha
\end{bmatrix}
\]
and
\[\tW = \cdot\begin{bmatrix}
     \alpha & \eps \sqrt{1-\alpha^2} \\
     \eps \sqrt{1-\alpha^2} & -\alpha
\end{bmatrix}.
\]
We first show that $\tW \uc_{\eps} \mW$. Since both matrices are symmetric, suffices~\cite[Lemma 3.7]{AKMPSV20} to show the equivalent statement for PSD approximation between $\mI-\tW,\mI-\mW$ and $\mI+\tW,\mI+\mW$, i.e.,
\[(1-\eps)(\mI-\mW) \preceq \mI-\tW \preceq (1+\eps)(\mI-\mW)\]
\[(1-\eps)(\mI+\mW) \preceq \mI+\tW \preceq (1+\eps)(\mI+\mW).
\]
All four inequalities are implied by 
\[\begin{bmatrix}
    \eps (1-\alpha) & \eps \sqrt{1-\alpha^2}\\
    \eps \sqrt{1-\alpha^2} & \eps (1+\alpha)
\end{bmatrix} \succeq \mzero.
\]
Using~\Cref{fact:SCPSD} this is equivalent to $\eps(1-\alpha) -\eps(1-\alpha^2)(1+\alpha)^{-1}\geq 0$ which holds with equality, so $\tW \uc_\eps \mW$. Now assume that $\tW \svn{\eps'} \mW$ for some $\eps'$. From the definition of SV approximation applied to test vector $z=\begin{bmatrix}x\\x\end{bmatrix}$, we have
\[
\pm(\tW-\mW) \preceq \eps'\left(\mI - \mW^2\right)= \eps' (1-\alpha^2) \mI.
\]
This is equivalent to
\[
\epsilon \sqrt{1-\alpha^2} \leq \eps' (1-\alpha^2).
\]
Thus for $\tW \svn{\eps}\mW$, a necessary condition is for $\eps' \geq \frac{\eps}{\sqrt{1-\alpha^2}}$.
\end{proof}

We then use these separations to prove that SV approximation enjoys properties not enjoyed by prior notions:
\UCpropFail*
\begin{proof}
    All positive claims follow from the respective lemmas, so it remains to show these properties do not hold for UC approximation. All such properties fail to hold from the observation that UC approximation is not preserved under the asymmetric lift operation, as otherwise it would imply SV approximation.

    Formally, we have from \Cref{prop:SVUCsepGraphs} that for $n\in \N$ there exist $\tW,\mW$ such that $\tW\uc_{O(1/\sqrt{n})}\mW$ yet $\tW$ is not a $.3$-SV approximation of $\mW$. Now let 
    \[\tM := \begin{bmatrix} \mzero & \mzero \\ \tW & \mzero\end{bmatrix},\quad\quad \begin{bmatrix} \mzero & \mzero \\ \mW & \mzero\end{bmatrix}.\]
    \begin{claim}
        $\tM$ is not a $.3$-UC approximation of $\mM$.
    \end{claim}
    \begin{proof}
        Assume for contradiction $\tM\uc_{.3}\mM$. Then since UC approximation implies standard approximation, we have that $\tM$ is a $.3$-approximation of $\mM$ with respect to $\sym[\mM]$. By \Cref{SVequiv:somez}, this implies that $\tW\svn{.3}\mW$ and hence $\tM\uc_{.3}\mM$ by \Cref{lem:SVimpliesUC}, which is a contradiction. 
    \end{proof}
    We can then use this to derive that all three properties do not hold. First, clearly this implies that UC is not preserved under arbitrary embeddings. Second, it is easy to show that
    \[
    \tA := \begin{bmatrix} \tW & \mzero \\ \mzero & \mzero\end{bmatrix}\uc_{.3}   \begin{bmatrix} \mW & \mzero \\ \mzero & \mzero\end{bmatrix}=:\mA.
    \]
    Let $\Pi$ be the permutation such that $\Pi\tA=\tM$ and $\Pi\mA=\mM$. We have by the claim that applying $\Pi$ cannot preserve UC approximation, and since $\Pi\uc_0\Pi$ we have that UC approximation is likewise not preserved under products.
\end{proof}

\section{Singular Values Facts}

\begin{lemma}\label{lem:singularPower}
    Let $\sigma_i(\cdot)$ denote the $i$th largest singular value. For any matrix $\mA \in \C^{n \times n}$ and any positive integer $k$,
    \[
        \sigma_2(\mA^k) \leq \sigma_2(\mA) \cdot \sigma_1(\mA)^{k-1}.
    \]
\end{lemma}

\begin{proof}
    By the variational characterization of singular values, we have
    \[
        \sigma_2(\mA^k) = \min_{v \in \C^n} \max_{x \perp v} \frac{\|\mA^k x\|}{\|x\|} \leq \min_{v \in \C^n} \max_{x \perp v} \frac{ \|\mA^{k-1} \| \cdot \|\mA x\|}{\|x\|} \leq \|\mA\|^{k-1} \cdot \min_{v \in \C^n} \max_{x \perp v} \frac{\|\mA x\|}{\|x\|}  = \sigma_2(\mA) \cdot \sigma_1(\mA)^{k-1}.
    \]
\end{proof}

\begin{lemma}\label{lem:sigmasum}
    Let $\sigma_i(\cdot)$ denote the $i$th largest singular value. For any matrix $\mA,\mB \in \C^{n \times n}$,
    \[
        \sigma_2(\mA + \mB) \leq \sigma_2(\mA) + \sigma_1(\mB).
    \]
    Furthermore, if $\mA,\mB$ share a common right singular vector that achieves the maximum singular value in each of them respectively, then
        \[
        \sigma_2(\mA + \mB) \leq \sigma_2(\mA) + \sigma_2(\mB)
    \]
\end{lemma}

\begin{proof}
    For the first part of the statement, by the variational characterization of singular values, we have
    \[
        \sigma_2(\mA + \mB) = \min_{v \in \C^n} \max_{x \perp v} \frac{\|(\mA+\mB) x\|}{\|x\|} \leq \min_{v \in \C^n} \max_{x \perp v} \frac{\|\mA x\|+\|\mB x\|}{\|x\|} \leq \min_{v \in \C^n} \max_{x \perp v} \frac{\|\mA x\|}{\|x\|}  +\|\mB\|.
    \]
    For the second part of the statement, let $v$ denote a common right singular vector of $\mA,\mB$ that respectively achieves the maximum singular value in each of them. Then we have
    \begin{align*}
        \sigma_2(\mA + \mB) &= \min_{v' \in \C^n} \max_{x \perp v'} \frac{\|(\mA+\mB) x\|}{\|x\|} 
        \leq   \max_{x \perp v} \frac{\|(\mA+\mB) x\|}{\|x\|} \\
        &\leq \max_{x \perp v} \frac{\|\mA x\|+\|\mB x\|}{\|x\|} \leq \max_{x \perp v} \frac{\|\mA x\|}{\|x\|} + \max_{x \perp v} \frac{\|\mB x\|}{\|x\|} 
        = \sigma_2(\mA) + \sigma_2(\mB).
    \end{align*}
    
    The third part of the statement has essentially the same proof as the second.
\end{proof}

We now use these properties to show that small perturbations preserve the smallest singular value:
\begin{lemma}\label{lem:singularPerturb}
    Let $\sigma_2(\cdot)$ denote the second largest singular value. Let $G$ be an Eulerian graph such that $\sigma_2(\mD_G^{-1/2}\mA_G\mD^{-1/2})\leq 1-1/\gamma$. For every $H$ such that $\mD_H=\mD_G$ and $\|\mA_H-\mA_G\| \leq \delta$, we have $\sigma_2(\mD_H^{-1/2}\mA_H\mD_H^{-1/2}) \leq 1-1/\gamma + \delta / d_\text{min}$ where $d_\text{min}$ is the minimum diagonal entry of $\mD_G$.
\end{lemma}
\begin{proof}
    By \cref{lem:sigmasum},
    \begin{align*}
\sigma_2(\mD_H^{-1/2}\mA_H\mD_H^{-1/2}) &\leq \sigma_2(\mD_G^{-1/2}\mA_G\mD^{-1/2}) + \|\mD_G^{-1/2}(\mA_H-\mA_G)\mD^{-1/2}\| \\
&\leq 1-1/\gamma + \|\mD_G^{-1/2}\|^2 \cdot \|\mA_H-\mA_G\| \\
&= 1-1/\gamma + \delta/d_\text{min}. \qedhere
    \end{align*}
\end{proof}
\begin{lemma}\label{lem:SVSingular}
    Let $\mA, \tA \in \R^{n \times n}$ be the adjacency matrices of Eulerian graphs with no isolated vertices. Suppose $\tA \svgraph{\eps} \mA$. Let $\mD$ be their diagonal matrix of degrees. Then 
    \[
        1-\sigma_2(\mD^{-1/2} \tA \mD^{-1/2}) \leq (1+2 \eps) \cdot (1-\sigma_2(\mD^{-1/2} \mA \mD^{-1/2})).
    \]
\end{lemma}

\begin{proof}
Define $\mN = \mD^{-1/2} \mA \mD^{-1/2}$ and $\tN = \mD^{-1/2} \tA \mD^{-1/2}$.

It suffices to show 
\[
\pm(\sigma_2(\tN)-\sigma_2(\mN)) \leq \sigma_2(\tN-\mN) \leq 2 \eps \cdot (1-\sigma_2(\mN)).
\]
The first inequality follows from \cref{lem:sigmasum}. For the second inequality, note that the definition of SV approximation can be rewritten as for all $x \in \C^n$,
\[
\|(\tN-\mN)x\| \leq \eps \cdot (1-\|\mN x\|^2).
\]
Note that $\tN,\mN$ have a common left and right eigenvector given by $v=\mD^{1/2} \vones$, and this achieves their respective maximum singular values of $1$.\footnote{We know it corresponds to a singular value of $1$ and this its not possible to have a singular value larger than $1$ for this matrix (\cite{CKPPRSV17} lemma B.4).} Hence, if we wish to maximize the LHS over all unit vectors $x$, it suffices to maximize over all unit vectors $x \perp v$. Doing so, we obtain
\[
\|\tN-\mN\| \leq \eps \cdot (1-\sigma_2(\mN)^2) \leq 2 \eps \cdot (1-\sigma_2(\mN))
\]
where the last inequality is because $\|\mN\| \leq 1$ (\cite{CKPPRSV17} lemma B.4).
\end{proof}

\begin{lemma}\label{lem:eulerian_lazy}
    Let $\sigma_i(\cdot)$ denote the $i$th largest singular value. Let $\mA \in \R^{n \times n}$ be the adjacency matrix of a strongly connected Eulerian graph $G$ with diagonal degree matrix $\mD$ and no isolated vertices. Let $\mN = \mD^{-1/2} \mA \mD^{-1/2}$. Let $\lambda$ denote the second largest eigenvalue of $\mS_{\mN}$.\footnote{Here we mean largest according to the actual eigenvalues, not their magnitudes.}
    Then we have
    \[
        \sigma_2( (1-\gamma) \mN + \gamma \mI) \leq 1 - (1-\lambda) \gamma + O(\gamma^2).
    \]
    In particular, if $\mA$ has weak mixing time $T$ and all weights are in $[1,U]$, then 
    \[
        \sigma_2( (1-\gamma) \mN + \gamma \mI) \leq 1 - 1/\poly(nTU/\gamma).
    \]
\end{lemma}

\begin{proof}
    Let $v$ denote a singular vector of $(1-\gamma) \mN + \gamma \mI$ corresponding to its second-largest singular value. We then have
    \begin{align*}
    [\sigma_2( (1-\gamma) \mN + \gamma \mI)]^2
    &= (1-\gamma)^2 v^* \mN^* \mN v + 2 \gamma (1-\gamma) v^* \mS_{\mN} v + \gamma^2 \\
    &\leq (1-\gamma)^2 + 2 \gamma (1-\gamma) \lambda + \gamma^2 \\
    &\leq 1 - 2(1-\lambda) \gamma + O(\gamma^2).
    \end{align*}
    In the above, we used the fact that $\|\mN\| \leq 1$ (\cite{CKPPRSV17} lemma B.4).

        To show the second part of the statement, define $\mL=\mI-\mN$. Note that $\mL$ is related to $\mI-\mA\mD^{-1}$ by a change of basis with condition number $\poly(nU)$. We may assume the graph is strongly connected as this is necessary for the weak mixing time to be finite. It suffices to prove
        \[
            \frac{1}{1-\lambda} \leq \poly(n) \cdot T.
        \]
        We do so by showing
        \[
            \frac{1}{1-\lambda} \leq \poly(n) \cdot T \leq \poly(n) \cdot \|\mS_{\mL}^+\|^2 \leq \poly(n) \cdot \|\mL^+\|^2 \leq \poly(nT).
        \]

        The first inequality is just the folklore result relating spectral gap to weak mixing time with an $O(\log n)$-factor loss. The second and fourth inequalities are from (\cite{CKPPSV16} Theorem 21).
        
        We now show the third inequality, that $\|\mS_{\mL}^+\| \leq \poly(n) \cdot \|\mL^+\|$. We have
        \[
            \|\mS_{\mL}^+\| \leq \poly(n) \cdot \|(\mL^T \mS_{\mL}^+ \mL)^+\| = \poly(n) \cdot \|\mS_{\mL^+}\| \leq \poly(n) \cdot (\|\mL^+\| + \|\mL^{T+}\|)/2 = \poly(n) \cdot \|\mL^+\|
        \]
        where the first inequality in the line immediately above is a corollary of (\cite{CKPPSV16} Lemma 13).
\end{proof}

\begin{lemma}\label{lem:norm_sqrt}
For matrices $\ma,\mb \in \C^{n\times n}$,
$$
\| \ma \mb \| = \| (\ma^* \ma)^{1/2} \mb \|, \|  \mb \ma \| = \| \mb (\ma \ma^*)^{1/2} \|.
$$
\end{lemma}
\begin{proof}
For any matrix $\mm$, 
$$\|\mm\| = \sqrt{\lambda_{\max}(\mm^* \mm)} = \max_{x \in \C^m-\{\vzero\}} \frac{\vx^* \mm^* \mm \vx}{\vx^* \vx}.$$
Thus,
\begin{align*}
\| \mb \ma \| 
= \max_{x \in \C^m-\{\vzero\}} \frac{\vx^*  \mb^* \ma^* \ma \mb \vx}{\vx^* \vx} 
= \max_{x \in \C^m-\{\vzero\}} \frac{\vx^* \mb^* (\ma^* \ma)^{1/2} (\ma^* \ma)^{1/2}  \mb \vx}{\vx^* \vx} 
 = \| (\ma^* \ma)^{1/2} \mb \|.
\end{align*}
The other equality is proved similarly.
\end{proof}

\section{Stationary Distribution Facts}

We first define the stationary norm of a graph and prove some useful properties.
\begin{definition}[Stationary Distribution] 
    Let $G$ be a strongly connected directed graph, and let $\pi = \pi(G)$ be its unique stationary distribution. We have that $\pi > 0$ entrywise, $\sum_{v\in V}\pi_v = 1$, and letting $\mW:=\mA\mD^{-1}$ be the random walk matrix of $G$, we have $\mW \pi=\pi$. Let $\pimin = \min_{v\in V}\pi_v>0$ be the minimum value in the stationary distribution.For a subset $S\subseteq V$, let $\vol(S)=\sum_{v\in S}\pi_v$ be the volume of the set under the stationary distribution.
\end{definition}
Observe that for every probability distribution $p$, we have $\|\mW p\|_\pi \leq \|p\|_\pi$, as~\cite{}
\[
\norm{\mW p}_\pi^2
= \sum_{u \in V} \pi_u \left(\sum_{v \in V} [\mW]_{uv} p_v \right)^2
\leq 
\sum_{u \in V} \pi_u \left(\sum_{v \in V} [\mW]_{uv} \pi_v^{-1} \right)
\left(\sum_{v \in V} \pi_v p_v^2 \right)
= \norm{p}_\pi^2.
\]
\begin{claim}\label{clm:xcapT} Let $G=(V,E)$ be a strongly connected graph. For every $x\in \R_{\geq 0}^V$ and $T\subseteq V$, we have $\langle 1_T,x\rangle\leq \|x\|_{\pi}\cdot \sqrt{\vol(T)}$.
\end{claim}
\begin{proof}
    By the Cauchy-Schwarz inequality
    \[
    \langle 1_T,x\rangle = \sum_v 1_{v\in T} x_v = \sum_v(x_v/\sqrt{\pi_v})(1_{v\in T}\sqrt{\pi_v}) \leq \|x\|_{\pi}\cdot \sqrt{\vol(T)}. \qedhere
    \]
\end{proof}
We now state the main proposition:
\begin{proposition}\label{prop:statNonNeg}
    Let $G$ be a strongly connected graph on $n$ vertices $V$ and edge weights in $[1,U]$. For every $\ell\in \N$ and $S,T \subseteq V$ with $\vol(S)+\vol(T)\geq 1$ we have either $\Cut_{G^\ell}(S,T)=0$ or $\Cut_{G^\ell}(S,T)\geq (\pimin/2U)^3$.
\end{proposition}
For ease of application, we state a corollary of this for $(S,S^c)$ cuts and uncuts:
\begin{corollary}\label{cor:cutNonNeg}
    Let $G$ be a strongly connected graph on $n$ vertices $V$ and edge weights in $[1,U]$. For every $\ell\in \N$ and $S\subseteq V$, we have
    \[\Cut_{G^\ell}(S) \in \{0\}\cup [(\pimin/2U)^3,1] \text{ and } \Uncut_{G^\ell}(S) \in \{0\}\cup [(\pimin/2U)^3,1].\]
\end{corollary}
\begin{proof}
    The first claim follows immediately from~\Cref{prop:statNonNeg} with $S=S,T=S^c$ as $\vol(S)+\vol(S^c)=1$. The latter claim follows as $\Uncut_{G^\ell}(S)=\Uncut_{G^\ell}(S^c)$, and either $\vol(S)\geq 1/2$ or $\vol(S^c)\geq 1/2$. Without loss of generality assuming the former, and then the claim follows from from~\Cref{prop:statNonNeg} with $S=S,T=S$.
\end{proof}

We first prove that applying $\mW$ to the stationary distribution restricted to a subset of vertices either preserves the stationary norm, or decreases it by a non-negligible amount.
\begin{lemma}\label{lem:Wtopi}
    Let $S\subseteq V$ be an arbitrary set of vertices, let $p=\pi|_S$ be the stationary distribution restricted to $S$, and let $p' = \mW p$. Then either:
    \begin{enumerate}
        \item $p' = \pi|_{B}$ for some $B\subseteq V$ with $\vol(B)=\vol(S)$. \label{itm:piCase1}
        \item $\|p'\|_{\pi}^2\leq \|p\|_\pi^2-\pimin^2/2U^2$. \label{itm:piCase2}
    \end{enumerate}
\end{lemma}
Note that it is not the case that the first case always implies $B=S$, as (for instance) we could have $S$ be one side of a bipartition and $B$ be the other side.
\begin{proof}[Proof of \Cref{lem:Wtopi}]
    First, note that $p'\leq \pi$ entrywise as $p' = \mW p =\mW (\pi|_S) \leq \pi$. Let $B = \supp(p')$. We break into cases based on the size of $\vol(B)$:
    
    \begin{enumerate}
        \item We claim $\vol(B)<\vol(S)$ can never occur. Assuming for contradiction we are in this case, we have $\|p'\|_1 = \sum_{v\in B}p'_v\leq \sum_{v\in B}\pi_v< \sum_{v\in S}\pi_v=\|p\|_1$. But this is impossible as $\|p'\|_1 = \|p\|_1$ as $\mW$ preserves the sum of entries. 
        
        \item If $\vol(B)=\vol(S)$, we have by the above argument that $p'_v=\pi_v$ for every $v\in B$, so~\Cref{itm:piCase1} holds.
        \item Otherwise, $\vol(B)>\vol(S)$, and so there is some edge $(u,v)$ such that $u\notin S$ and $v\in B$, and thus there is some $v\in B$ where $\pimin/U\leq p'_v\leq \pi_v - \pimin/U$. Therefore,
        \[
            \frac{p_v^{'2}}{\pi_v}\leq \frac{p_v^{'2}}{p'_v+\pimin/U} \leq p'_v-\frac{p'_v\cdot \pimin/U}{2}\leq p'_v-\pimin^2/2U^2
        \]
        and so
        \begin{align*}
            \|p'\|_{\pi}^2 &= \sum_{v\in B}\frac{(p'_v)^2}{\pi_v}
            \leq \sum_{v\in B}\frac{(p'_v)^2}{(p'_v)} - \pimin^2/2U^2
            =\|p\|^2_\pi-\pimin^2/2U^2
        \end{align*}
        where the last line uses that $\sum_{v\in B}p'_v = \sum_{v\in S}p_v = \|p\|^2_\pi$, so so~\Cref{itm:piCase2} holds. \qedhere
    \end{enumerate}
\end{proof}

\begin{proof}[Proof of~\Cref{prop:statNonNeg}]
    Let $p^0 = \pi|_S$ be the stationary distribution restricted to $S$, and for every $i\in [\ell]$ let $p^i = \mW^ip^0$.
    By applying~\Cref{lem:Wtopi} inductively (and using that $\|p^i\|_\pi \leq \|p^{i-1}\|_\pi$) we obtain that either one of the following two cases occurs:
    \begin{enumerate}
        \item We have $p^\ell = \pi|_B$ for some $B\subseteq V$ with $\vol(B)=\vol(S)$. Then either $B\cap T=\emptyset$ (in which case the cut value is exactly $0$) or there is some $v\in B\cap T$ (in which case the cut value is at least $p^\ell_v =\pi_v\geq \pimin$).
        \item We have $\|p^\ell\|_\pi \leq \|p^0\|_\pi-\pimin^2/2U^2$. In this case, 
        \begin{align*}
            \Pr_{(i,j)\sim \muedge(G^\ell)}[i\in S, j\in T^c] &= \langle 1_{T^c},p^\ell\rangle\\
            &\leq \|p^\ell\|_\pi \cdot \sqrt{\vol(T^c)} \tag{\Cref{clm:xcapT}}\\
            &\leq (\|p^0\|_{\pi}-\pimin^2/2U^2)\cdot \sqrt{\vol(T^c)}\\
            &\leq (\sqrt{\vol(S)}-\pimin^2/2U^2)\sqrt{\vol(S)}\leq \vol(S)-\pimin^{5/2}/2U^2
        \end{align*}
        where $\|p^0\|_{\pi} \leq \sqrt{\vol(S)}$ follows as $\pi\leq \vones$, and the third inequality uses that $1\leq \vol(S)+\vol(T)=\vol(S)+1-\vol(T^c)$. 
        Thus, we have
        \begin{align*}
            \Cut_{G^\ell}(S,T) &= \Pr_{(i,j)\sim \muedge(G^\ell)}[i \in S, j\in T]\\
            &\geq \Pr_{i\sim \pi}[i\in S]-\Pr_{(i,j)\sim \muedge(G^\ell)}[i\in S, j\in T^c]\geq (\pimin/2U)^3. \qedhere
        \end{align*}
    \end{enumerate}
\end{proof}

\end{document}